\documentclass[msom,sglanonrev]{informs4}

\OneAndAHalfSpacedXII 

\usepackage{natbib}
 \bibpunct[, ]{(}{)}{,}{a}{}{,}%
 \def\bibfont{\small}%
\TheoremsNumberedThrough     
\ECRepeatTheorems

\EquationsNumberedBySection 

\MANUSCRIPTNO{}

\usepackage{setspace}
\usepackage{etoolbox}
\usepackage{booktabs} 
\usepackage{tikz}
\usetikzlibrary{shapes,arrows}
\usepackage{pgfplots}
\pgfplotsset{compat=newest}
\usepgfplotslibrary{dateplot}
\usepackage{float} 
\usepackage{bm} 
\usepackage{pdfpages} 
\makeatletter
\@ifundefined{KV@Gin@artifact}{\define@key{Gin}{artifact}[true]{}}{}
\makeatother
\usepackage{hyperref}	
\usepackage{longtable}
\hypersetup{
	colorlinks=true,
	linkcolor=magenta,
	filecolor=magenta,      
	urlcolor=magenta,
	citecolor = magenta,
}
\usepackage{wrapfig}

\usepackage{bigstrut}
\usepackage[lined, noend]{algorithm2e}

\usepackage{array}
\newcolumntype{L}[1]{>{\raggedright\let\newline\\\arraybackslash\hspace{0pt}}m{#1}}
\newcolumntype{C}[1]{>{\centering\let\newline\\\arraybackslash\hspace{0pt}}m{#1}}
\newcolumntype{R}[1]{>{\raggedleft\let\newline\\\arraybackslash\hspace{0pt}}m{#1}}

\newcommand{\Inf}{$\infty$}

\begin{document}

%
\begingroup
\hoffset=0pt \voffset=0pt\relax


# Accelerating the Adoption of Residential Solar Power Systems: Policy Analysis using a Dynamic Structural Model

Sebastián Souyris[1], Jason A. Duan[2], Anantaram Balakrishnan[2], Varun Rai[3]

This version: June 30, 2026



**Problem definition:** Solar electricity generation is a strategic component of energy portfolios to meet growing demand and reduce carbon emissions. To accelerate the installation of household solar photovoltaic (PV) systems, governmental agencies and municipalities offer upfront incentives such as rebates and tax credits. Designing effective incentive policies with limited budgets requires principled data-driven approaches that capture the drivers of residential PV adoption and the effects of incentives on the adoption rate.

**Methodology/results:** We develop a dynamic structural model for the diffusion of residential PV systems by modeling the adoption decisions of forward-looking individual households who consider the economic trade-offs between installing a system now or later. This choice depends on both their return on investment and the influence by neighbors. Our model segments households by home values and urbanization level, incorporates unobserved heterogeneity, and captures the observed spatio-temporal dynamics of installations. We fit the model to a detailed dataset consisting of household-level information in Austin, Texas, and use Bayesian estimation to overcome the computational burden. In out-of-sample tests, our model predicts PV installations more accurately than other contemporary models. We simulate various counterfactual scenarios, with embedded dynamic equilibrium of PV diffusion, to generate insights into the design of effective rebate policies. The framework can also be adapted to study the adoption dynamics of other durable technologies.

**Managerial implications:** Our counterfactual analysis shows that a rebate program offered over a limited duration is more effective than a prolonged and costly program to incent PV adoption and reduce carbon emissions. This counterintuitive result stems from households' forward-looking behavior, influence of neighboring adopters, and the consequent acceleration in the adoptions before the rebate ends. We also explore policies which reduce rebates in multiple steps and offer differentiated rebates for different segments. We find that a two-step rebate reduction is superior to multiple small reductions and geographically differentiated policies can improve policy performance, whereas differentiation based on home value yields little improvement over a uniform rebate policy.

*Keywords*: solar photovoltaic, technology diffusion, dynamic discrete choice model, Bayesian structural estimation, counterfactual policy analysis, energy policy

---

[1] Lally School of Management, Rensselaer Polytechnic Institute, souyris@rpi.edu
[2] McCombs School of Business, The University of Texas at Austin, jason.duan@mccombs.utexas.edu, anantb@mail.utexas.edu
[3] LBJ School of Public Affairs, The University of Texas at Austin, rai@austin.utexas.edu

## 1. Introduction

With the growing demand for electricity, uncertainties in the cost of conventional fuels, and concerns about environmental impact, solar energy has emerged as an important component of energy portfolios for governments and society. By the end of 2025, cumulative U.S. solar capacity reached 279 GW with over five million installations, 97% of them residential rooftops on 7% of U.S. homes (SEIA and Wood Mackenzie 2026). Globally, installed solar capacity reached roughly 2.4 TW, with a record 511 GW added in 2025 alone (IRENA 2026). The cost of solar panels and associated hardware has decreased dramatically, from \$10 per watt in 2007 to a benchmark market price of about \$3.15 per watt by 2024 (Ramasamy et al. 2025), making the technology accessible and economically viable for many households. To accelerate solar photovoltaic (PV) adoption, governmental agencies and municipalities have provided upfront incentives such as equipment and installation rebates as well as tax credits that are proportional to the installed wattage. The design of such programs has become particularly important and relevant for state agencies and local utilities since their rebates are the principal policy levers to promote residential PV adoption after federal residential clean energy credit programs ended (SEIA and Wood Mackenzie 2026). Figure 1 shows the median cost of residential PV systems and the upfront incentives, both in dollars per watt, in the U.S. from 2000 to 2024. Since costs and incentives decreased at nearly the same rate during this period, the net installation cost remained relatively flat. Despite the stable net cost, the number of new PV installations, shown in Figure 1, follows a pattern similar to the diffusion of other new technologies. The adoption rate is relatively low during the "introduction" phase but then increases dramatically (around 2010) in the "growth" phase. Higher rebates can attract more adopters, but budgetary constraints limit both the magnitude and duration of rebate programs. Policy makers are interested in identifying effective incentive policies to maximize PV adoption and carbon savings using limited budgets. This paper seeks to address such policy-support needs in this and other similar contexts that entail deciding incentives to promote the adoption of expensive durable products.

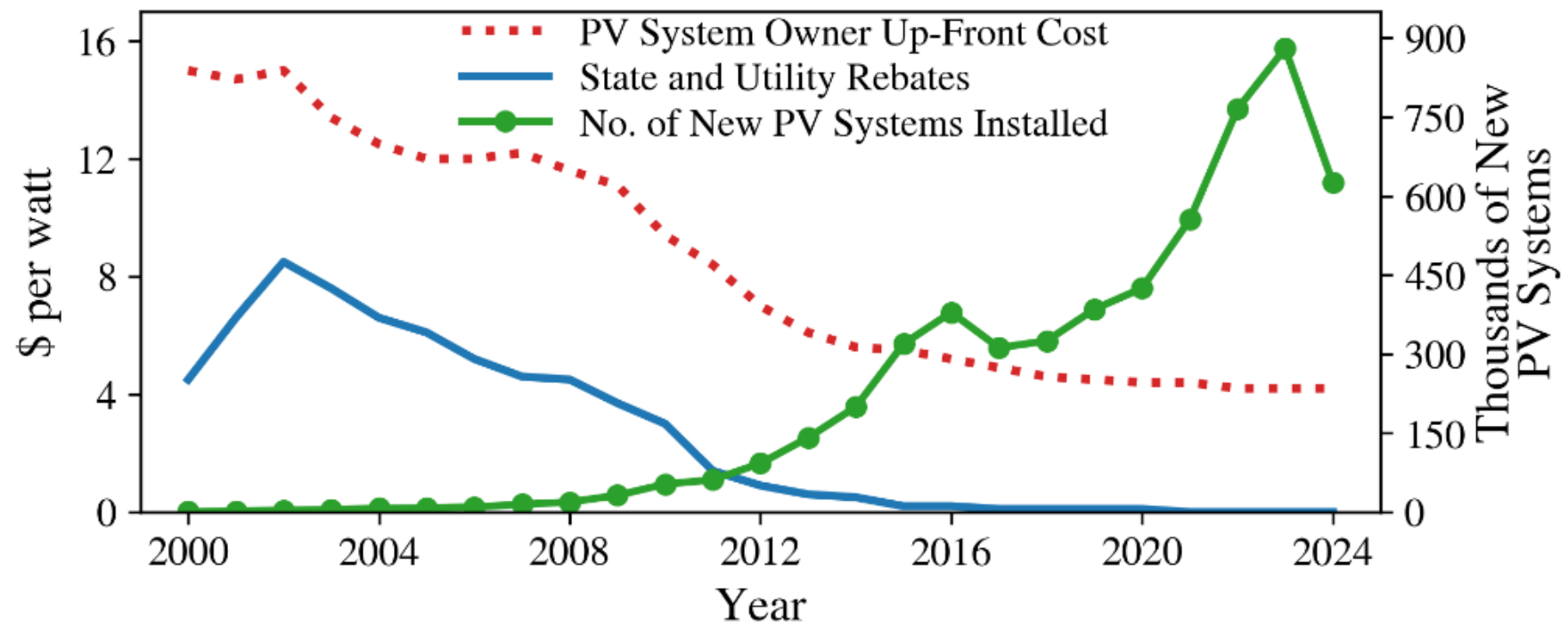


**Figure 1 Average PV cost, rebates, and new residential PV installations in the U.S. 2000-2024 (all monetary values in 2024 US\$).** Sources: Barbose et al. 2024, SEIA 2024, SEIA/Wood Mackenzie 2026.

Assessing the impact of PV adoption incentives requires developing a data-driven model that captures the main factors that drive household adoption decisions, namely, the long-term economic benefits of

installing a PV system and the influence of neighboring adopters, to predict how adoption evolves as a function of the incentive policy. Both the economic and peer influence factors vary by household and with time. Installing a PV system requires a significant initial investment that may be recouped over the lifetime of the system through energy credits or savings in other energy costs. The net present value of the investment, including rebates, depends on the system size, time-varying cost, and other attributes specific to each property. Peer effects on PV adoption has been noted by researchers (e.g., Bollinger and Gillingham 2012) and supported by observations such as the following in the trade press: “people’s decision to go solar affects other people’s decision to go solar”, producing “a kind of viral activity” in adoptions (Woody 2022). In the PV installation context, peer influence stems from the visibility of rooftop panels on nearby properties installed by neighbors who have adopted the technology. It is household-specific, depending on the property’s location, and evolves over time as neighbors adopt. Finally, the adoption dynamics also need to account for the finite duration of rebate programs. For instance, industry analysts note that PV installers may experience a “run on the bank,” a surge in demand, before the incentives expire.

To support policy decisions on upfront PV installation incentives, we first develop a dynamic structural model for each household’s adoption decision based on a detailed dataset on attributes (such as demographics, property characteristics and year it was built, and electricity consumption) of individual residences of actual or potential PV adopters in and around Austin, Texas. The utility function in this dynamic discrete choice model (DDCM) is time-varying and household-dependent. It consists of two main components: an economic component represented by the net present value (NPV) of installing a PV system (investment cost less upfront incentives and savings/credits for future PV electricity generation), and a neighborhood effect that is proportional to the number of neighbors within a specified radial distance who have already installed PV systems. In each period (quarter), households that have not yet installed a PV system are forward-looking, i.e., they decide whether or not to adopt in this period by comparing the utility of adopting now versus deferring adoption to a future period. This dynamic comparison entails considering not only the costs and benefits of later adoption but also estimating the endogenous number of future neighboring adopters. By aggregating the likelihood of adoption over all households in every period, we can estimate the total adoption rate over time.

This dynamic structural model, based on individual household decisions, provides two main advantages. First, since our model begins with estimates of primitive parameters in the consumer utility functions, it permits counterfactual analysis of alternative incentive policies that differ in the amount and duration of the incentive offered. In contrast, in reduced-form models, the parameters estimated from historical data are generally not invariant with respect to changes in policies (e.g., those outside the observed ranges of rebates and program durations), as highlighted by the Lucas critique (Lucas 1976, Nevo and Whinston 2010). Second, by treating each household as the unit of analysis, we can accommodate household-specific utility values based on house size, system capacity, costs and future

savings, actual neighborhood composition, wealth and geographic segment, and unobserved heterogeneity in attitudes towards renewable energy. Models based on more aggregate data, e.g., by zip code or region, cannot incorporate these refinements. As we elaborate in Section 2, existing empirical models of PV adoption have different objectives: most are reduced-form models that are not suited for counterfactual analysis of upfront incentive policies, and the few available structural models do not consider peer effects or household-level heterogeneity. We are not aware of any previous empirical study that provides policy insights on the level and duration of upfront PV incentives and the trade-offs among long-run adoption, carbon emissions, and program budgets.

For our empirical analysis, we first combine the data on demographics, property characteristics, period built, and electricity consumption of each household with information about the attributes of the PV system needed for each property, installation costs, and the timing of PV system installation by households that have already adopted. To estimate the structural parameters that govern adoption behavior, we design and implement a Bayesian Markov Chain Monte Carlo method that allows us to handle unobserved heterogeneity in a dynamic programming model. Within our structural modeling framework, we calibrate alternative model specifications (e.g., by varying the radial distance that defines the influencing neighborhood of each household, with and without household-level heterogeneity, and so on) and select the model with the lowest out-of-sample prediction error. Our estimation results indicate that peer influence from neighbors within a mile radius is most salient, and wealthier households are more inclined to install PV systems earlier, consistent with the large upfront investment a system requires. To assess the effectiveness of our modeling approach, we compare the predictive performance of our model with that of alternative models, including the myopic (not forward-looking) variant of our model, a hazard model recently applied by Bollinger et al. (2022) on disaggregate household-level data, and a Bass model to forecast aggregate adoption trends. Our DDCM with forward-looking households and neighborhood effects outperforms these models in out-of-sample prediction error.

We next exercise our model to identify effective upfront incentive policies for accelerating PV diffusion with a limited budget. We begin with a class of policies that has two levers, the initial net cost (net of rebates) and the period in which the rebate is terminated or equivalently the rebate duration. We conduct a counterfactual analysis by varying both, and evaluating the effectiveness of each policy on three dimensions: (a) PV penetration, i.e., the cumulative number of adopters and PV megawatts installed; (b) the total funds or budget spent on the rebate program; and (c) the total carbon savings due to PV installations. Our analysis generates an efficient frontier of policies from which policymakers can select the one that best meets their relative preferences for the three dimensions. Our simulation reveals that adoptions accelerate toward the end of a rebate program as households anticipate the rebate's expiration. Consequently, offering rebates for a longer *duration* (which is more costly) does not necessarily increase adoption or savings in carbon emissions. Two features of our model contribute to this counterintuitive result – forward-looking adoption decisions and the peer influence of neighboring adopters. Offering a

higher *initial rebate* increases PV adoption and carbon savings but requires more budget spending; under budget constraints, programs with moderate rebates and durations can outperform policies that offer smaller subsidies for longer durations or higher subsidies for shorter durations. We also study the impact of gradually reducing the rebate over multiple rebate switching periods, and offering differentiated rebate policies where the rebates vary by geographic or economic segments.

This paper focuses on residential PV adoption, using data from the Austin metropolitan area since the incentives are decided and offered by local agencies. Our model and estimation methodology readily applies to PV adoption in other geographical areas that may have different household characteristics, demographics, and attitudes towards PV technology. More broadly, the structural framework can be adapted to study durable technologies characterized by substantial upfront investment, policy-dependent incentives, peer influence, and forward-looking adoption timing. Although prior research on durable goods has recognized the importance of these features, to our knowledge, none has incorporated all four in a single empirical model. In Section 7, we elaborate on these broader application contexts and discuss how the framework can be extended to technologies such as electric vehicles and heat pumps through modest modifications to the model specification, while using the same modeling and estimation approach.

The remainder of the paper is organized as follows. Section 2 summarizes the related literature. Section 3 discusses the data used to develop and calibrate our model. Section 4 develops the dynamic structural model and the estimation algorithm. Section 5 presents and interprets the estimation results. Section 6 discusses our counterfactual analysis and policy insights, and Section 7 concludes the paper. All the Appendices are provided in the e-companion.

## 2. Literature Review

Following the sustainability research roadmap in operations management proposed by Atasu et al. (2019), this paper, as a policy-oriented study of the diffusion of PV technology, lies in the intersection of two major streams of the sustainability literature: the low-carbon economy and technological innovation. It addresses the key challenge of evaluating counterfactual incentive policies to promote solar panel adoption. For this purpose, we employ a dynamic structural modeling approach, adding to the stream of operations literature to address policy questions in other operations contexts (e.g., Olivares et al. 2008, Musalem et al. 2010, Li et al. 2014, Bray and Mendelson 2015, Wang et al. 2022, Bravo et al. 2025).

### 2.1. Empirical Studies of PV Diffusion

Traditionally, empirical modeling of product and technology diffusion uses the Bass model framework to forecast market penetration at the *aggregate* level. These approaches have been extended and applied to various contexts, incorporating factors such as pricing, marketing mix, technological innovations, consumer income levels, and network connections (e.g., Van den Bulte and Stremersch 2004, Goldenberg et al. 2010, Shaikh et al. 2010). As reduced-form models, they are not amenable to counterfactual analysis of the drivers of diffusion outside the observed ranges.

Focusing on PV adoption, a growing body of empirical papers uses household and regional-level data

to analyze various aspects of the diffusion of household PV systems. Table 1 positions our paper relative to the most closely related studies along five key modeling features.

Klein and Deissenroth (2017) study residential PV adoption in Germany under *feed-in tariffs* (FIT), i.e., credits for PV electricity generation, using a model based on prospect theory and loss aversion. Candas et al. (2019) develop a model of PV adoption in Germany and Italy to examine how profitability, public opinion, and self-consumption affect PV diffusion and the potential phase-out of FIT. To compare the effectiveness of upfront incentives versus subsidies for future electricity production, De Groote and Verboven (2019) propose a dynamic structural model with forward-looking consumers, but without peer effects and with heterogeneity only at the regional level. Their counterfactual analysis suggests that upfront incentives through subsidies can achieve greater adoption of PV technology at a much lower budgetary cost. Gillingham and Bollinger (2021) conduct a randomized experiment to compare the impact of a group pricing scheme on the adoption rate for installing PV systems. Bollinger et al. (2022) investigate the components and reasons underlying peer effects from neighbors' solar installations using an econometric model based on the single-period utility of installing PV systems. Finally, Langer and Lemoine (2022) develop a dynamic structural model to investigate household intertemporal preferences and the design of subsidies over time.

**Table 1 Comparison of model features for recent related literature**

| Paper | Data Granularity Level | Structural Model & Counterfactual Analysis | Neighborhood Effects | Forward Looking | Unobserved Heterogeneity |
|---|---|---|---|---|---|
| Klein and Deissenroth (2017) | national | no | no | yes | no |
| De Groote and Verboven (2019) | household | yes | no | yes | across regions |
| Candas et al. (2019) | state | no | yes | no | no |
| Gillingham and Bollinger (2021) | municipality | no | no | no | across regions |
| Bollinger et al. (2022) | household | no | yes | no | across regions |
| Langer and Lemoine (2022) | household | yes | no | yes | no |
| **This paper** | **household** | **yes** | **yes** | **yes** | **across households** |

Our paper generalizes this work by developing a dynamic structural econometric model based on disaggregated household-level PV adoption data, building on dynamic technology adoption models (e.g., Smith and Ulu, 2012, 2017). We estimate the utility function governing adoption decisions and use the estimated parameters to simulate policy counterfactuals. Compared to De Groote and Verboven (2019) and Langer and Lemoine (2022), our model further incorporates household-level unobserved heterogeneity and local neighborhood effects in the household's dynamic decision-making.

Other empirical papers on PV adoption include Bollinger and Gillingham (2012), an early paper in this stream, which uses a reduced-form model and ZIP-code-level data from the California Solar Initiative. Guajardo (2018) analyzes California PV installation data to investigate the relationship between the ownership structure of PV systems and their operating efficiency. Wolske et al. (2020) survey the research on peer influence on households' energy management and attitudes. Barnes et al. (2022) find that adopters in early diffusion areas report higher active peer influence, whereas potential future adopters in

later diffusion areas experience higher passive peer influence. Gillingham and Tsvetanov (2019) apply a reduced-form Poisson hurdle model of PV demand that captures both behavioral and market heterogeneity to study how policy incentives and prices influence household adoption. Bollinger, Darghouth, et al. (2023) analyze the complementarities in the demand for residential PV systems and battery energy storage. Bollinger et al. (2024) show that the duration of promotional campaigns strongly influences solar adoption by shaping word-of-mouth and social diffusion.

Another related stream of research on PV adoption in the public policy literature employs agent-based modeling and simulation. These models typically apply heuristic consumer adoption rules based on behavioral theories from sociology, and calibrate model parameters using past adoption data. For instance, Rai and Robinson (2015) examine the diffusion of PV systems in Austin, Texas, by modeling agents according to the theory of planned behavior (Ajzen 2002). Our study uses the same dataset and applies their approach to identify potential adopters.

### 2.2. Analytical Models of PV Diffusion and Incentives

A stream of analytical papers employs stylized models to investigate how government subsidies influence PV adoption and how tariff structures and compensation schemes for PV electricity impact social welfare. Cohen et al. (2016) assess the effects of government incentives on consumer adoption, considering demand uncertainty when designing subsidies to achieve PV adoption targets. Sunar and Swaminathan (2021) analyze how solar panel deployment and net-metering policy affect the electric utility business, and find that utilities can benefit from strategies that motivate their customers to adopt PV. Using a sequential game model, Agrawal et al. (2022) show that a higher generation subsidy increases total adoption whereas a higher investment subsidy can lower it, so third-party ownership (leasing) can be more effective. Singh and Scheller-Wolf (2021) incorporate competition in the solar marketplace and examine how various tariff structures align with social objectives. Yu et al. (2018) investigate structural conditions under which it is more efficient for a government to subsidize consumers or manufacturers to maximize social welfare. To understand the relationship between incentives and solar equipment supply, Chemama et al. (2019) study flexible subsidy policies when the government faces a strategic PV supplier. Alizamir et al. (2016), Babich et al. (2020), and others analyze the impact of different FIT structures to design more effective programs to incentivize PV adoption. More recently, Mamaghani and Çakanyildirim (2023) focus on pricing mechanisms to mitigate the power utility death-spiral caused by the adoption of distributed solar power generation. Kaps et al. (2023) analyze the synergy between solar power and storage for off-grid electrification.

## 3. Residential PV Systems Data

Our study is based on information from Austin, Texas, which is served by a single electric utility, Austin Energy, a publicly owned entity that provides electricity for most of the city and some suburbs. For consumers who install PV systems, the local electric utility reduces the household's electricity bill by awarding credits for the electricity generated by the PV system. For every kWh the system produces, the

utility credits the household at a residential solar energy rate, which Austin Energy calls the *Value-of-Solar* rate (a type of FIT). If the total solar credits for a month exceed the electricity consumed, the excess credits roll over to the next month. In practice, a credit rollover seldom occurs because the installed PV systems are often sized to produce less electricity per month than the household consumes. In addition to electricity bill credits, households that install PV systems also received rebates and tax credits from federal and local authorities during the period of our study.

Austin Energy's service area contains about 170,000 single-family residential households. Our dataset covers all adopters and non-adopters from 2004 to 2013. During that period, the housing stock growth was approximately 2% per year. In total, 2,724 households adopted and installed PV systems. To identify potential PV adopters among the remaining households, we employ Rai and Robinson's (2015) approach based on the value of a control variable called *perceived behavioral control* (PBC), a composite metric of home market value, house size, tree cover, and irradiance received. Installation is not economical for households with low irradiance or high tree cover (see Rai and Robinson, 2015, for details); since such households have a non-positive PBC value, we consider only households with a positive PBC value as potential adopters. For our dataset, this approach identified 42,384 households as potential adopters (including the actual adopters up to 2013). We obtained the characteristics of individual households (property values, land parcel information) from the Travis County Appraisal District's (TCAD) records (www.traviscad.org) for all appraised properties in the county. For each residential property in the Austin Energy service area, we extracted the following data from the TCAD database: property coordinates and land area; total constructed building area; period built; appraised market value (in 2012); and whether a PV system can be installed on the property. We exclude commercial buildings, apartment buildings, and other ineligible properties. Further, for each household that has already installed a PV system, we obtained from Austin Energy the following information about their systems: application date for solar rebate program; system installation date; total PV system cost (including PV panel and hardware price, and installation cost) quoted at the time of application; system size in watts; rebate in $ per watt; and, estimated annual electricity savings in kWh. The data reveal that the average lead time between the date a household applies for the rebate program and the installation date is roughly one quarter (139 days). Analogous to Figure 1 but specific to the Austin area, Figure 2(a) shows the PV system purchase price, rebate, and net cost, all in $ per watt, from 2004 to 2012, and Figure 2(b) plots the cumulative number of adopters in and around Austin. Figure 3 depicts the geographic spread of installed PV systems in various subdivisions of Austin from 2005 to 2013.

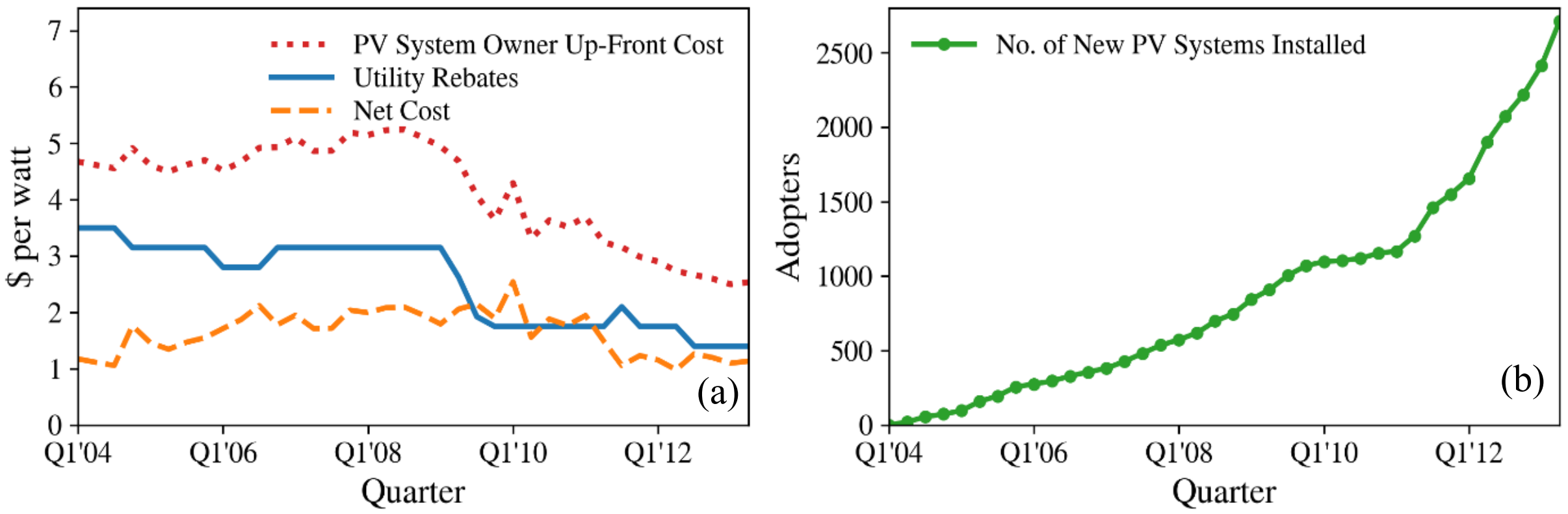


**Figure 2  PV cost and installation trends in Austin, Texas 2004-13:**

**(a) PV installation price, rebate, and net cost after ITC, and (b) Cumulative adopters**

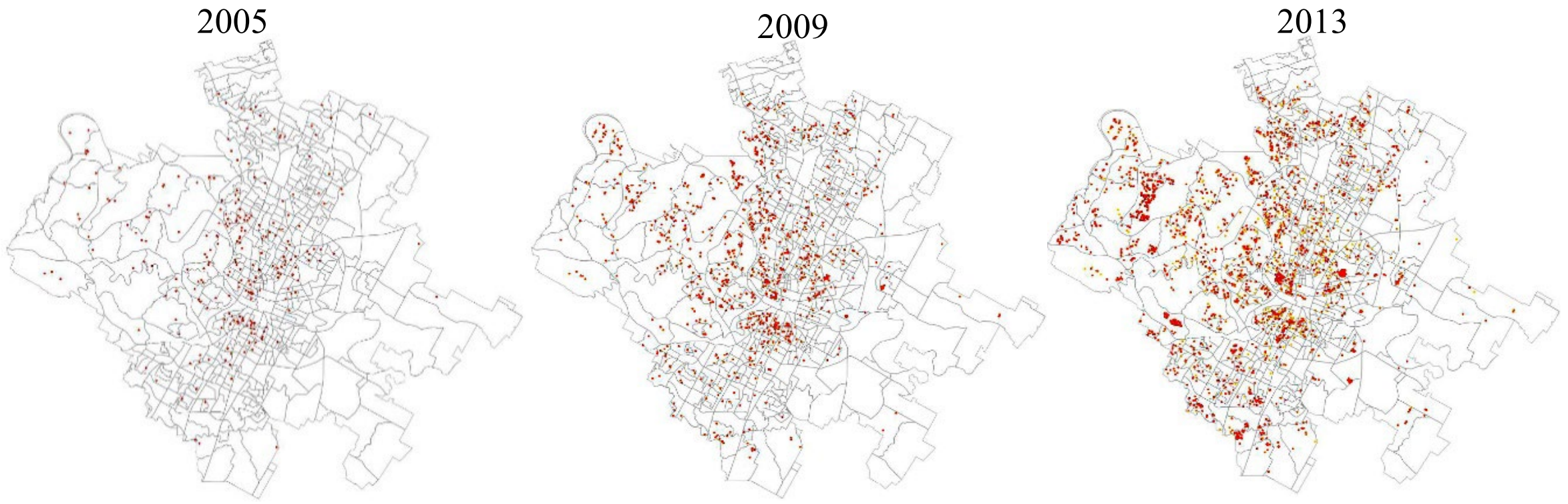


**Figure 3  Austin PV adoption dynamics, 2005–2013 (Source Austin Energy and TCAD)**

We classify the households of current and potential adopters along two dimensions: *economic* and *geographic* segments. Our *economic* segments are based on the appraised home market values (which we use as a proxy for household wealth since we do not have reliable information about each occupant's wealth). For *geographic* segments (urban, suburban, and rural), we use the Claritas PRIZM classification (claritas360.claritas.com) to group the households. Most of these homes were valued at less than $300,000 in 2013, and there are almost twice as many homes in suburban areas as in urban areas. Table 2(a) shows the number of potential adopters in the Austin Energy service area, by economic and geographic segment, and Table 2(b) shows the % of these households that installed a PV system by 2013

**Table 2  Profile of adopters by segment in Austin Energy service area**

(a) No. of households (potential adopters)

| | **Home Value** | | |
|---|---|---|---|
| **Geog. Seg.** | **< 300K** | **> 300K** | **Total** |
| **Urban** | 11,007 | 1,990 | 12,997 |
| **Suburban** | 18,888 | 2,790 | 21,678 |
| **Rural** | 4,285 | 3,424 | 7,709 |
| **Total** | 34,180 | 8,204 | 42,384 |

(b) % adopters/potential households by 2013

| | **Home Value** | | |
|---|---|---|---|
| **Geog. Seg.** | **< 300K** | **> 300K** | **Average** |
| **Urban** | 2.4% | 18.7% | 4.9% |
| **Suburban** | 2.9% | 24.9% | 5.7% |
| **Rural** | 3.7% | 20.2% | 11.0% |
| **Average** | 2.8% | 21.4% | 6.4% |

As we might expect, a larger proportion (21.4% on average) of owners with high-value homes (valued at more than $300,000) were early adopters compared to % adopters in the lowest economic segment.

## 4. Dynamic Discrete Choice Model

Based on a long research survey of PV adoption in Texas, Rai and McAndrews (2012) found that the following factors are important when a household decides whether to adopt and install a PV system:

(i) The future cash flow from generating electricity using a PV system. PV systems require a large upfront investment, so the multi-year savings in electricity charges over the system's lifetime are an important economic determinant.

(ii) Consideration of inter-temporal trade-offs between adopting a PV system today versus deferring adoption. Adopting now captures available government incentives and starts the electricity savings immediately, whereas deferring may lower the system cost but risks missing the incentives.

(iii) The predisposition of a potential adopter to install a PV system. This predisposition increases as more neighboring households install PV systems: visible installations signal that the technology is reliable and cost-effective and also serve as advertising. These neighborhood influences increase the likelihood of adoption.

Models of PV adoption should therefore incorporate inter-temporal economic trade-offs and the effects of the installed base in the decision maker's neighborhood; we next propose a dynamic discrete choice model with these features.

### 4.1. Model

We index the households by $i$, and let $M$ denote the set of households that are potential adopters who can benefit from installing PV systems. We segment these households along two dimensions: *economic segments* $(E)$ based on home value, and *geographic segments* $(G)$ based on population density. Let $e_i \in E$ be household $i$'s economic segment and let $g_i \in G$ denote its geographic segment. We consider a horizon of $T$ periods (quarters), indexed as $t \in \{1, \dots, T\}$, to implement the rebate policies and measure PV penetration. Let $\varrho_i$ be the period in which household $i$'s property was built ($\varrho_i = 1$ if it was built before or at the first period). Let $t_i$ denote the period in which household $i$ adopts a PV system ($t_i = T$ for any household that does not adopt during the horizon). Let $M_t$ denote the subset of households that have decided to adopt before period $t$, i.e., $M_t = \{i \in M: t_i < t\}$, and define $M_t^c = \{i \in M \backslash M_t, \varrho_i \leq t\}$ as the set of households whose property has been built but that have not yet adopted a PV system by period $t$.

In any period $t$, a household $i \in M_t^c$ that has not yet adopted a PV system must decide whether to adopt in the current period. Let $a_{it} = \{0,1\}$ denote the corresponding action variable that takes the value 1 if the household adopts in period $t$, and is 0 otherwise. If the household adopts in $t$, the system is installed and starts producing electricity in the next period $(t + 1)$ until the end of the system's lifespan $L$ (typically, 20 years), which, in our computations, is beyond the last period $T$ (so we do not need to consider replacement installations by the same household). Given the characteristics of household $i$'s property (e.g., square footage, location, orientation), the PV system size is generally determined by

professional installers. Hence, we treat system size as a household-dependent input parameter rather than a decision variable. Let $w_i$ denote the size (in watts) of the PV system needed by household $i$. The expected electricity generated per watt by this system in period $t$ is $q_{it}$ kWh. This value is calculated based on household $i$'s site-specific irradiance and varies seasonally, i.e., by quarter, due to variation in sunlight throughout the year. With a system size of $w_i$, the total expected electricity generated in period $t$ is $q_{it}w_i$. The household knows the values of the following parameters that determine the net cost of the PV system in period $t$: (a) the equipment (solar panels and associated hardware) and installation cost $p_t$ ($ per watt), which we henceforth call *PV cost*; (b) the incentive or *rebate* $r_t$ ($ per watt) offered by the local utility; and (c) the residential solar investment tax credit (ITC) rate $\text{ITC}_t$ (%) provided by the federal government (from 2006 to 2013, this credit was 30% of the household's investment in the PV system, www.energy.gov/eere/solar). So, if household $i$ installs a PV system in period $t$, it needs to pay a *net cost* of $c_t = (p_t - r_t)(1 - \text{ITC}_t)$ per watt. Figure 2(a) displays the PV cost per unit ($p_t$), rebate ($r_t$), and net cost per unit ($c_t$) in $ per watt in Austin from 2004 to 2013. Finally, the household also knows the residential *solar energy rate*, $v_t$ ($ per kWh), i.e., the bill credits per kWh that the household receives for the electricity generated by a PV system.

In addition to the economic costs and benefits of installing a PV system, our model incorporates the influence of the installed base in a household's neighborhood on its adoption decision as follows. At the beginning of period $t$, the installed base, $M_t$, consists of all the households that have installed PV systems *before* period $t$. Each household $i \in M_t^c$ that has not adopted by the beginning of period $t$ observes the installed systems in its neighborhood. Following the literature in economics and marketing (e.g., Manchanda et al. 2008, Conley and Udry 2010, Narayanan and Nair 2013), we partition a household's neighborhood based on their radial distance from the property. Specifically, for each household $i$, we divide $M_t$ into concentric rings of households located within $D$ specified ranges of radial distances from $i$. For $l = 1, \dots, D$, let $d_l$ be the radius of the circle that demarcates radial ring $l$ from ring $(l+1)$, with $d_1 < d_2 < \cdots < d_D$ (e.g., $d_1$= 0.5 mile, $d_2$= 1 mile, $d_3$= 2 miles). The $l^{th}$ ring spans all households whose distance from household $i$ lies in the range $(d_{l-1}, d_l]$. Let $H_{ilt} \subseteq M_t$ denote the set of installed PV systems located within this radial ring before period $t$. Define $h_{ilt} = |H_{ilt}|$ as the number of installations in the $l^{th}$ ring, and let the vector $\boldsymbol{h}_{it} = \{h_{i1t}, \dots, h_{iDt}\}$ be the number of installations in the $D$ rings. We model the neighborhood effects from household $i$'s neighbors as a linear function of the vector $\boldsymbol{h}_{it}$ (as shown in Appendix E, we also tested fitting nonlinear concave functions of $\boldsymbol{h}_{it}$, but found that the linear function provides the best out-of-sample performance). As noted in Section 3, adoption precedes completed installation by approximately one quarter. Since $\boldsymbol{h}_{it}$ represents the total number of installed PV systems at the *beginning* of period *t,* the one-period time lag between adoption and installation of a PV system resolves the reflection problem (Manski 1993, Bollinger and Gillingham 2012).

#### 4.1.1. Per-period utility for PV adoption

We formulate each household's utility for installing a PV system in each period $t$ using a random utility

model that includes a common baseline value, the NPV of cash flows, neighborhood effects, a random effect for household-level heterogeneity, and an idiosyncratic random shock. Specifically, the utility of a household that decides to adopt a PV system in period $t$ is:

$$u_{i1t}\left(c_t, \boldsymbol{h}_{it}, \xi_i, \varepsilon_{i1t}; \rho, \alpha_{e_i}, \boldsymbol{\gamma}_{g_i}\right) = \rho + \alpha_{e_i}\mathrm{NPV}_{it}(c_t) + \boldsymbol{\gamma}'_{g_i}\boldsymbol{h}_{it} + \xi_i + \varepsilon_{i1t}, \tag{1}$$

In equation (1), the five additive components of the utility are defined as follows.

(a) $\rho$ is the baseline utility, common across households and periods, for installing a PV system.

(b) $\mathrm{NPV}_{it}(c_t)$ is the discounted future bill credits for electricity generated from a PV system that household $i$ adopts in period t, minus the net cost of installation. Household $i$ pays a unit net cost $c_t$ per watt, and its PV produces electricity from $(t+1)$ until the end of its lifespan, $(t+1+L)$. Hence, in every period τ, with $t+1 \le \tau \le t+1+L$, the household receives $v_\tau q_{i\tau}$ per watt in bill credits. With a discount factor of $\beta$ for future cash flows, the discounted payoff is:

$$\mathrm{NPV}_{it}(c_t) = \left(\sum_{\tau=t+1}^{t+1+L} \beta^{\tau-t} v_\tau q_{i\tau} - c_t\right) w_i. \tag{2}$$

The coefficient $\alpha_{e_i}$ which multiplies $\mathrm{NPV}_{it}(c_t)$ represents the weight that the household ascribes to the economic benefit (versus other factors) of solar panels. We permit this coefficient to vary by economic segment $e_i \in E$, recognizing that different segments may assign more or less importance to economic benefits.

(c) The neighborhood effects are modeled as the vector of coefficients $\boldsymbol{\gamma}_{g_i} = \{\gamma_{g_i 1}, \ldots, \gamma_{g_i D}\}$ multiplied by the total number of previous installations $\boldsymbol{h}_{it}$ in each radial ring around household $i$. We permit these "influence" parameters $\boldsymbol{\gamma}_{g_i}$ to vary by geographic segment $g_i \in G$ to account, in part, for the composition and density of households in these segments.

(d) The time-invariant heterogeneity $\xi_i$, which depends on household $i$'s attitudes toward renewable energy and their economic circumstances, is known to the household but not to the modeler.

(e) The per-period idiosyncratic random shock $\varepsilon_{i1t}$ of adopting is also observable to the household but not to the modeler.

Observe that the coefficients $\alpha_{e_i}$ and $\boldsymbol{\gamma}'_{g_i}$ capture common components of utility functions among households in similar economic and geographic segments. We assume that the last two terms in equation (1) are independent across households. If household $i$ decides *not* to adopt a PV system at period $t$, it obtains a utility of

$$u_{i0t}(\varepsilon_{i0t}) = \varepsilon_{i0t}, \tag{3}$$

where $\varepsilon_{i0t}$ is the idiosyncratic random shock of not adopting a PV in period $t$, known to the household but not to the modeler.

### 4.1.2. Dynamic optimal decision for PV adoption

Let $\boldsymbol{\theta} = (\rho, \boldsymbol{\alpha}, \boldsymbol{\gamma})$ be the vector of structural parameters and $\boldsymbol{\varepsilon}_{it} = \{\varepsilon_{i0t}, \varepsilon_{i1t}\}$. Recall that $a_{it}$ is the action variable which equals 1 if household $i$ adopts in period $t$, and is 0 otherwise. We formulate the decision

problem facing a household $i$ that has not yet adopted at period $t$ as the following dynamic programming problem:

$$V_{it}(c_t, \boldsymbol{h}_{it}, \xi_i, \boldsymbol{\varepsilon}_{it}; \boldsymbol{\theta}) = \max_{a\in\{0,1\}} \mathcal{V}_{iat}(c_t, \boldsymbol{h}_{it}, \xi_i, \boldsymbol{\varepsilon}_{it}; \boldsymbol{\theta}), \quad (4)$$

where $\mathcal{V}_{iat}$ is the following action-specific value function for making the decision $a \in \{0,1\}$:

$$\begin{aligned}\mathcal{V}_{i0t}(c_t, \boldsymbol{h}_{it}, \xi_i, \boldsymbol{\varepsilon}_{it}; \boldsymbol{\theta}) &= u_{i0t}(\varepsilon_{i0t}) + \beta E_{c,\boldsymbol{h},\boldsymbol{\varepsilon}}\left[V_{it+1}(c, \boldsymbol{h}, \xi_i, \boldsymbol{\varepsilon}; \boldsymbol{\theta})|c_t, \boldsymbol{h}_{it}, \boldsymbol{\varepsilon}_{it}\right] \\ &= \varepsilon_{i0t} + \beta E_{c,\boldsymbol{h},\boldsymbol{\varepsilon}}\left[V_{it+1}(c, \boldsymbol{h}, \xi_i, \boldsymbol{\varepsilon}; \boldsymbol{\theta})|c_t, \boldsymbol{h}_{it}, \boldsymbol{\varepsilon}_{it}\right],\end{aligned} \quad (5)$$

and

$$\mathcal{V}_{i1t}(c_t, \boldsymbol{h}_{it}, \xi_i, \boldsymbol{\varepsilon}_{it}; \boldsymbol{\theta}) = u_{i1t}(c_t, \boldsymbol{h}_{it}, \xi_i, \varepsilon_{i1t}; \boldsymbol{\theta}) = \rho + \alpha_{e_i}\text{NPV}_{it}(c_t) + \boldsymbol{\gamma}'_{g_i}\boldsymbol{h}_{it} + \xi_i + \varepsilon_{i1t}. \quad (6)$$

Household $i$ will adopt at $t$ if $\mathcal{V}_{i1t}(c_t, \boldsymbol{h}_{it}, \xi_i, \varepsilon_{i1t}; \boldsymbol{\theta}) \geq \mathcal{V}_{i0t}(c_t, \boldsymbol{h}_{it}, \xi_i, \varepsilon_{i0t}; \boldsymbol{\theta})$.

#### 4.1.3. State variables, transition equations, and other model specifications

The state variables in the dynamic optimization problem defined by (4) are the unit PV net cost $c_t$ and the number of previously installed systems $\boldsymbol{h}_{it}$. We treat $c_t$ as an exogenous state variable since the net cost of installation is determined by supply and demand in the broader economy rather than the actions of individual households in a particular locality. The institutional record supports this treatment: Austin Energy sets the rebate through an annual budgeting process that accounts for changes in PV panel and installation prices, and once the city council approves the budget, the posted dollar-per-watt rebate applies uniformly to all qualifying applicants for that year. The rebate is therefore neither responsive to individual adoption decisions nor negotiated for individual households. The remaining components of the net cost, the federal ITC and the solar energy rate, were likewise fixed by statute or through citywide proceedings, outside individual household transactions.

We model the households' projection of the unit PV net cost $c_t$ (\$ per watt) as a first-order autoregressive (AR(1)) process:

$$c_{t+1} = \kappa c_t + \epsilon_t^c, \qquad \epsilon_t^c \sim N(0, \sigma_c^2), \quad (7)$$

where $\kappa$, and $\sigma_c^2$ are estimated from historical data via regression.

The number of previously installed systems $\boldsymbol{h}_{it}$ is an endogenous state variable obtained by aggregating households' actions in past periods. Because $\boldsymbol{h}_{it}$ is non-decreasing, we use an approach similar to Rust (1987) assuming that the number of new installations in the neighborhood of household $i$ in period $t$, $\Delta_{it} = \boldsymbol{h}_{it+1} - \boldsymbol{h}_{it}$, follows an exponential distribution:

$$\boldsymbol{h}_{i,t+1} = \boldsymbol{h}_{it} + \Delta_{\text{it}} \qquad \Delta_{it} \sim \text{Exponential}\left(\frac{1}{(\lambda_i - 1)\boldsymbol{h}_{it}}\right), \lambda_i > 1 \quad (8)$$

Note that, in a rational expectations equilibrium, household $i$ forms the conditional expectation $E(\boldsymbol{h}_{i,t+1}|\boldsymbol{h}_{it}) = \lambda_i \boldsymbol{h}_{it}$, which must coincide with the aggregated adoption decisions of the other households in household $i$'s neighborhood. For our analysis, $\lambda_i$ is estimated from historical data because our data is the result of equilibrium outcomes. This approach is similar to the concept of oblivious

equilibrium introduced by Weintraub et al. (2008, 2010). For our counterfactual policy simulations in Section 6, $\lambda_i$ needs to be determined by an iterative procedure until it converges.

We model households' time-invariant heterogeneity as a Normally distributed random effect:

$$\xi_i \sim N\left(0, \sigma_\xi^2\right). \tag{9}$$

The per-period idiosyncratic random shock $\varepsilon_{iat}$ is an independent and identically distributed Type I extreme value over time.

#### 4.1.4. Model identification

We conclude this section on model development by addressing the identification of model parameters. We fix the discount factor, which is generally unidentified in dynamic discrete choice models (Magnac and Thesmar 2002). The state transition densities are identified directly from the observed cost and installation histories. Given these transition densities, the conditional choice probabilities identify the utility parameters through the Hotz and Miller (1993) inversion and the nested fixed-point structure of Rust (1987). The coefficient of $\text{NPV}_{it}(c_t)$ is identified from the variation in net cost over households and time, induced by the cost and rebate path and interacted with household-specific system size. The neighborhood effect is identified from the variation in neighbor-adopter counts within specific radial rings for each household over time. The baseline utility is pinned down by the aggregate adoption level, while the heterogeneity variance is identified from residual variation in utility across households. Appendix A.1 provides the mathematical treatment of the identification conditions.

One remaining concern is that the annual rebate budget may have been set in anticipation of aggregate demand shocks in the following year. Such policy endogeneity could confound the estimated effect of $\text{NPV}_{it}(c_t)$. To assess whether the rebate could be endogenous and bias our estimates, we conduct an empirical test of whether household $\text{NPV}_{it}(c_t)$ has different effects on PV adoption depending on whether we include period-specific temporal fixed effects that control for unobserved demand shocks over time. In our model, $\text{NPV}_{it}(c_t)$ is the only variable that depends on the rebate. We find that the estimated effect of $\text{NPV}_{it}(c_t)$ remains essentially unchanged when temporal fixed effects are included, alleviating this endogeneity concern. We report the empirical test details and provide additional institutional information about the rebate in Appendix A.2.

### 4.2. Bayesian Estimation

To estimate the parameters of model (1), we use Bayesian estimation as in Imai et al. (2009) and Norets (2009), who demonstrate that combining value-function approximation with Markov Chain Monte Carlo (MCMC) estimation reduces the computational burden of estimating dynamic discrete choice models with continuous unobserved heterogeneity. In contrast to maximum likelihood estimation (MLE), which solves the dynamic programming problem to convergence in each iteration, this approach iterates over the Bellman operator once per MCMC draw. The algorithm employs Metropolis–Hastings steps within a Gibbs sampling framework. At iteration *r* of the Metropolis–within–Gibbs, the method draws each candidate parameter in the vector $\boldsymbol{\theta}^{(r)} = \left\{\rho^{(r)}, \boldsymbol{\alpha}^{(r)}, \boldsymbol{\gamma}^{(r)}\right\}$, given the other parameters and the data. Each

iteration also generates a random grid of size $J$ of the net cost state variable $c_t$, $\left\{c_t^{j(r)}\right\}_{j=1}^{J}$ and approximates the transition function with the Gaussian Kernel $f(c_{t+1}|c_t)$ from equation (7). We then approximate the value function using the sampled values from a given number $N$ of the most recent iterations. Let $\hat{\mathcal{V}}_{iat}^{(r)}\left(c_t^{(r)}, \boldsymbol{h}_{it}, \xi_i^{(r)}; \boldsymbol{\theta}^{(r)}\right)$ be the approximation of the action-specific value function for making the decision $a_{it} \in \{0,1\}$ at iteration $r$, given the state variables at the grid values and the sampled parameters. Then, $\hat{V}_{it}^{(r)}\left(c_t^{j(r)}, \boldsymbol{h}_{it}, \xi_i^{(r)}; \boldsymbol{\theta}^{(r)}\right)$ represents the value function at state $(c_t^{(r)}, h_{it}, \xi_i^{(r)})$. Let $\hat{E}\left[\hat{V}_{i,t+1}^{(r)}\left(c, \boldsymbol{h}_{i,t+1}, \xi; \boldsymbol{\theta}\middle| c_t^{j(r)}, \boldsymbol{h}_{i,t}, \xi_i^{(r)}; \boldsymbol{\theta}^{(r)}\right)\right]$ be the corresponding expected value in the next period $t+1$ at iteration $r$, computed as:

$$\hat{E}\left[\hat{V}_{i,t+1}^{(r)}\left(c, \boldsymbol{h}_{i,t+1}, \xi; \boldsymbol{\theta}\right)| c_t^{j(r)}, \boldsymbol{h}_{i,t}, \xi_i^{(r)}; \boldsymbol{\theta}^{(r)}\right]$$
$$= \frac{\sum_{n=r-N}^{r-1} K_\eta\left(\boldsymbol{\theta}^{(n)} - \boldsymbol{\theta}^{(r)}\right) \sum_{j=1}^{J} \hat{V}_{it+1}^{(n)}\left(c_{t+1}^{j(n)}, \boldsymbol{h}_{i,t+1}, \xi_i^{(n)}; \boldsymbol{\theta}^{(n)}\right) f\left(c_{t+1}^{j(n)}\middle| c_t^{j(r)}\right)}{\sum_{n=r-N}^{r-1} K_\eta\left(\boldsymbol{\theta}^{(n)} - \boldsymbol{\theta}^{(r)}\right) \sum_{j=1}^{J} f\left(c_{t+1}^{j(n)}\middle| c_t^{j(r)}\right)}, \quad (9)$$

where

$$\hat{V}_{it}^{(n)}\left(c_t^{j(n)}, \boldsymbol{h}_{it}, \xi_i^{(n)}; \boldsymbol{\theta}^{(n)}\right) = \log\left(\exp \hat{\mathcal{V}}_{i0t}^{(n)}\left(c_t^{j(n)}, \boldsymbol{h}_{it}, \xi_i^{(n)}; \boldsymbol{\theta}^{(n)}\right) + \exp \hat{\mathcal{V}}_{i1t}^{(n)}\left(c_t^{j(n)}, \boldsymbol{h}_{it}, \xi_i^{(n)}; \boldsymbol{\theta}^{(n)}\right)\right), \quad (10)$$

and

$$\hat{\mathcal{V}}_{i0t}^{(n)}\left(c_t^{j(n)}, \boldsymbol{h}_{it}, \xi_i^{(n)}; \boldsymbol{\theta}^{(n)}\right) = \beta \hat{E}\left[\hat{V}_{it+1}^{(n)}(c, \boldsymbol{h}_{it+1}, \xi; \boldsymbol{\theta})| c_t^{j(n)}, \boldsymbol{h}_{it}, \xi_i^{(n)}; \boldsymbol{\theta}^{(n)}\right], \quad (11)$$

$$\hat{\mathcal{V}}_{i1t}^{(n)}\left(c_t^{j(n)}, \boldsymbol{h}_{it}, \xi_i^{(n)}; \boldsymbol{\theta}^{(n)}\right) = \rho^{(n)} + \alpha_{e_i}^{(n)} \mathrm{NPV}_{it}\left(c_t^{j(n)}\right) + \boldsymbol{\gamma}_{g_i}^{(n)\prime} \boldsymbol{h}_{it} + \xi_i^{(n)}, \quad (12)$$

and $K_\eta\left(\boldsymbol{\theta}^{(n)} - \boldsymbol{\theta}^{(r)}\right)$ is a multivariate Gaussian kernel with bandwidth $\eta > 0$.

From these computations, we obtain the likelihood function of adoption by household $i$ in period $t$ as:

$$\hat{p}_{i1t}^{(r)}\left(c_t, \boldsymbol{h}_{it}, \xi_i^{(r)}; \boldsymbol{\theta}^{(r)}\right) = \frac{\exp \hat{\mathcal{V}}_{i1t}^{(r)}\left(c_t, \boldsymbol{h}_{it}, \xi_i^{(r)}; \boldsymbol{\theta}^{(r)}\right)}{\exp \hat{\mathcal{V}}_{i0t}^{(r)}\left(c_t, \boldsymbol{h}_{it}, \xi_i^{(r)}; \boldsymbol{\theta}^{(r)}\right) + \exp \hat{\mathcal{V}}_{i1t}^{(r)}\left(c_t, \boldsymbol{h}_{it}, \xi_i^{(r)}; \boldsymbol{\theta}^{(r)}\right)}. \quad (13)$$

Then, at iteration $r$, the log-likelihood for all the households is:

$$\mathcal{L}^{(r)} = \sum_{t=1}^{T} \sum_{i \in M_t^c} \left\{1_{\{a_{it}=1\}} \log \hat{p}_{i1t}^{(r)}\left(c_t, \boldsymbol{h}_{it}, \xi_i^{(r)}; \boldsymbol{\theta}^{(r)}\right) + 1_{\{a_{it}=0\}} \log\left(1 - \hat{p}_{i1t}^{(r)}\left(c_t, \boldsymbol{h}_{it}, \xi_i^{(r)}; \boldsymbol{\theta}^{(r)}\right)\right)\right\}, \quad (14)$$

where the inner sum in period $t$ is over households $i \in M_t^c$, those whose property is already built and that have not yet adopted a PV system, and $1_{\{a_{it}=y\}}$ is the indicator function which equals 1 if $a_{it} = y$, and is 0 otherwise. We optimize the expected value function by improving its value once at each iteration of the Bayesian sampler. Appendix C contains the pseudocode for the MCMC estimation algorithm. Imai et al. (2009) and Norets (2009) establish the convergence of the expected value function with the Bayesian sampler.

Due to the large number of households, evaluating each household's value and likelihood functions is computationally burdensome. We therefore aggregate non-adopting households into clusters with similar characteristics (home market value, house size, tree cover, received irradiance) within each census block group (www2.census.gov). Our cluster-based likelihood adapts Manski and Lerman's (1977) weighted choice-based sampling method to clusters. In Appendix B, we provide details of the clustering method and the corresponding weighted log-likelihood function. We prove in Proposition EC.1 that, asymptotically, this weighted cluster-based log-likelihood yields the same MLE as the log-likelihood function based on individual household data. The posterior mean from the Bayesian estimation also converges to the MLE asymptotically. This result applies to our analysis, since our clustering method yields a large number of clusters (over 1,100 clusters of size greater than one).

We implemented the estimation and the counterfactual simulation procedures in Julia 1.12 (Bezanson et al. 2017). Our MCMC sampler reaches its stationary distribution within the first 2,500 iterations, which we discard as burn-in (Appendix D); we then run 18,000 iterations of the sampling chain for just under 4 hours on a single processor (Intel Xeon w3-2435, 3.1 GHz, 128 GB of memory), so that the posterior summaries are computed from a large post-convergence sample. Appendix D reports the sampler's convergence diagnostics and the sensitivity of the estimates to longer runs. The workstation can run up to eight independent model estimations in parallel with adequate CPU and memory headroom, useful to concurrently estimate model variants and conduct robustness checks. For counterfactual simulations described in Section 6, a single policy path solves the rational expectations equilibrium in about five outer-loop iterations and roughly 80 minutes of elapsed time (Appendix C).

## 5. Estimation Results

To achieve good model fit, we considered many alternative model specifications in equation (1) and compared their out-of-sample predictive performance. The different models vary in: (i) the economic segmentation $E$ of households by wealth, which we proxy by the properties' appraised value; (ii) the sizes of the radial rings for models with neighborhood effects; and, (iii) the presence or omission of unobserved heterogeneity modeled by the random effect ($\xi_i$). We calibrated models with various combinations of these three factors; this section focuses on the combinations closest to the best-performing model. Appendix E reports the results for many other combinations we tested.

- Three alternative economic segmentations of house market values, in thousands of U.S. dollars: (i) $[0, \infty)$, (ii) [0,300), [300,∞), and (iii) [0,300), [300,600), [600,∞).
- Four choices for radial rings, with $[a, b)$ denoting a ring containing adopters within $a$ to $b$ miles from a given household: (i) no radial rings, i.e., corresponding to a model with no neighborhood effects, (ii) one radial ring [0,1) or [0,0.5), (iii) and two radial rings [0,0.5), [0.5,1) or [0,1), [1,$SB$), where the service boundary ($SB$) represents the border of the region that the electric utility serves and provides PV rebates.
- Two alternatives for random effects: With Heterogeneity (H) or Without Heterogeneity (NH).

We divide the data into a training set of 33 quarters (2004 Q1 to 2012 Q1) for estimating the structural parameters, and a test set of 5 quarters (2012 Q2 to 2013 Q2) for forecasting and validation (for robustness analysis, we assess the consistency of the estimates through an expanding data window re-estimation in Appendix F). We then compare the actual and predicted number of adopters in the test set to assess each model's performance and select the best model. Cash flows are discounted at an annual rate of 0.95, which corresponds to setting the quarterly discount factor $\beta$ = 0.987 in equation (2).

We first estimate the parameters $\kappa$ and $\sigma_c$ in equation (7) and $\lambda_i$ in equation (8) using regressions over the training set. The estimated values are $\kappa = 0.984$ and $\sigma_c = 0.32$, and the average value of $\lambda_i$ over all the households is 1.070. We then estimate the structural parameters for each model specification using the Bayesian estimation method discussed in Section 4.2. When calculating the value and likelihood functions, we assume that each household observes the state variables $c_t$ and $\boldsymbol{h}_{it}$ in the current period, and projects the future net cost $c_{t+1}$ using equation (7). Following the literature on approximate dynamic structural models with Markov-perfect equilibrium (Weintraub et al. 2008, 2010), for estimation purposes, we use the observed $\boldsymbol{h}_{it}$ to calculate the value function because the actual adoption data represents the equilibrium outcomes discussed in Section 4.1.

We then forecast the number of adopters for the out-of-sample *test set* periods by simulating the household's adoption decisions. To forecast, we randomly resample $K$ = 1,000 posterior sample paths from the last 3,000 iterations. For the out-of-sample periods, the state variables $c_t$ and $\boldsymbol{h}_{it}$ are projected using the estimated parameters $\kappa$ and $\lambda_i$ in equations (7) and (8). Let $p_{i1t}^{(k)}$ be the predicted choice probability of household *i* that has not yet adopted in period *t*, calculated based on the posterior sample path *k* and equation (13). Then, $P_{i1t}^{(k)} = p_{i1t}^{(k)} \times \left(1 - \sum_{t'=1}^{t-1} P_{i1t'}^{(k)}\right)$ is the probability that household *i* adopts in period *t*. Let $\widehat{N}_t^{(k)} = \sum_i P_{i1t}^{(k)}$ denote the predicted number of adopters in period *t* for posterior sample *k*, and let $N_t$ be the actual number of adopters in *t*. Then the mean absolute error of path $k$ is $\text{MAE}^{(k)} = \frac{100}{|test\ set|} \sum_{t \in test\ set} \left|\widehat{N}_t^{(k)} - N_t\right|$. We use the average over all $K$ sample paths, $\text{MAE} = \frac{1}{K} \times \sum_{k=1}^{K} \text{MAE}^{(k)}$, and the MAE sample standard deviation Std $\text{MAE} = \sqrt{\frac{\sum_{k=1}^{K}(\text{MAE}^{(k)} - \text{MAE})^2}{K-1}}$, to measure the predictive performance of each model. We also compute the mean absolute percentage error (MAPE), along with the corresponding standard deviation, for each model. We use MAE as the primary model-selection criterion to avoid the sensitivity of percentage-based errors to small denominators in low-adoption quarters. We compare both MAE and MAPE results in the Appendix E and find that the model rankings are robust across the two metrics.

Table 3 presents the results for 12 models: with one to three economic segments, no neighborhood effects or with one radial distance range [0,1) neighborhood effects, and with and without unobserved heterogeneity $\xi_i$. For each row (1) to (6) in the table, we distinguish the models with and without heterogeneity as H and NH, respectively, in the last two columns. Appendix E

presents the predictive performance results of 18 additional models, along with their estimated parameters. The MAEs for the 12 model specifications shown in Table 3 range from 31.85 to 85.51 adopters, and the model with the lowest MAE is Model (4)H. We next examine the structure of the 12 models and their out-of-sample MAEs. First, the models with neighborhood effects (rows (2), (4), and (6)) outperform the models without them (rows (1), (3), and (5)). The MAEs of models with neighborhood effects range from 31.85 to 49.14 adopters, against 67.73 to 85.51 for the models without them. This result is strong evidence that neighbors who have installed PV systems influence the adoption decisions of households that have not yet adopted. Second, Appendix E extends the comparison to additional specifications: a larger radial ring, or rings beyond one mile, adds only a negligible neighborhood effect. The out-of-sample prediction errors of the leading models using a one-mile radial ring versus a half-mile radial ring differ by less than 2 adopters. The structural parameter estimates are also stable across these alternative specifications. So, we select the one-mile radial ring as a parsimonious specification. Third, among the models with neighborhood effects, those that incorporate unobserved heterogeneity improve the fit for the leading specifications: Models (2)H and (4)H outperform their no-heterogeneity counterparts, while for the three-segment specification the two are close (Models (6)H and (6)NH). Fourth, the models with two market value segments (row 4) outperform models with just one segment (row 2) or with three segments (row 6). To summarize, Model (4)H, with one radial ring, two economic segments, neighborhood effects, and unobserved household heterogeneity, delivers the best balance between prediction performance and parsimonious specification.

**Table 3 Summary of model performance**

| | **Economic Segment ($E$)** | **Neighborhood Radial Ring ($D$)** | **MAE (Std MAE)** | |
|---|---|---|---|---|
| **Model** | Home Market Value ($'000) | Radial Distance Brackets (Miles) | **NH** | **H** |
| (1) | [0,∞) | no neighborhood effects | 67.73 (5.99) | 71.48 (4.97) |
| (2) | [0,∞) | [0,1) | 37.34 (12.25) | 32.66 (6.72) |
| (3) | [0,300), [300,∞) | no neighborhood effects | 85.00 (3.87) | 85.51 (3.66) |
| (4) | [0,300), [300,∞) | [0,1) | 45.77 (14.16) | 31.85 (2.41) |
| (5) | [0,300), [300,600), [600,∞) | no neighborhood effects | 69.01 (4.54) | 69.63 (4.35) |
| (6) | [0,300), [300,600), [600,∞) | [0,1) | 48.53 (5.59) | 49.14 (5.44) |

*MAE: mean absolute error (number of adopters). NH: without heterogeneity, H: with heterogeneity

Table 4 presents the posterior mean and Bayesian 95% credible intervals of Model (4)H's structural parameters. Appendix D shows the MCMC trace and density plots of the parameters. We observe consistent convergence for all parameters, as well as a bell-shaped density plot. Appendix F reports an expanding-window re-estimation showing that these parameter estimates are stable as additional quarters of adoption history are added to the sample. The estimates offer several insights into a household's propensity to adopt. First, the intercept $\rho$ is the baseline utility for any household for installing a PV system. For Model (4)H, this baseline utility is −6.792, implying that, for a household with no neighboring adopters, zero NPV, and no heterogeneity, the probability for adoption is $\frac{\exp(\rho)}{1+\exp(\rho)} = 0.001$.

**Table 4 Parameter estimates for Model (4)H**

| Parameter | Posterior Mean | 95% Credibility Interval | Parameter | Posterior Mean | 95% Credibility Interval |
|---|---|---|---|---|---|
| $\rho$ | −6.79158 | (−6.89602,−6.71202) | $\gamma$[0,1),urban | 0.03611 | (0.02695,0.05801) |
| $\alpha$[0,300K) | 0.00029 | (0.00026, 0.00032) | $\gamma$[0,1),suburban | 0.04932 | (0.04518,0.05853) |
| $\alpha$[300K,Inf) | 0.00064 | (0.00059,0.00070) | $\gamma$[0,1),rural | 0.12135 | (0.10514,0.15563) |
| $\sigma_\xi$ | 0.05098 | (0.04851,0.05368) | | | |

Second, households whose home value exceeds \$300,000 find more value in a PV system than households with lower-valued homes since $\alpha_{[300K,\infty)} > \alpha_{[0,300K)}$. This result confirms findings in the literature (e.g., Best et al., 2021Best et al., 2021Best et al., 2021) that wealthier households are more likely to invest in big-ticket items such as a PV system. Third, the neighborhood effect $\boldsymbol{\gamma}$ is significant for all geographic segments, but $\gamma_{[0,1.0),rural} > \gamma_{[0,1.0),suburban} > \gamma_{[0,1.0),urban}$, suggesting that rural households are influenced the most by *each* nearby installation, whereas urban households are influenced the least. However, since rural households have fewer neighbors, the total influence of all neighboring installations may not be highest for them. Lastly, the standard deviation of the heterogeneity, $\sigma_\xi$, is significant at 0.051 (recall that the mean for this component is zero) but has a much smaller magnitude than the baseline utility. So, although the heterogeneity between households is relevant for the adoption decision, this decision still depends on other explanatory factors.

### 5.1 Comparison with Benchmark Models

To evaluate our model's predictive performance, we compare it to the following benchmark models:

(a) Bollinger et al.'s (2022) hazard model, estimated at the individual household level;

(b) a myopic choice model with the same per-period utility equations (1) and (3), that assumes households are not forward-looking, i.e., using a static logistic regression;

(c) a partially forward-looking model in which households forecast future installation costs using equation (7), but only consider the currently observed neighboring adopters $h_{it}$, implying they do not account for future adopters; and

(d) a classic Bass model applied to aggregate quarterly adoption data.

Benchmark model (a) is the reduced-form model applied in Bollinger et al. (2022), but using our data. Specifically, it assumes that the PV adoption hazard rate $\lambda_{it}$ has the following model:

$$G(\lambda_{it}) = \xi_z + \eta_t + \beta_1 \text{NPV}_{it} + \beta_2 X_i + \beta_3 H_{it} + \beta_4 \text{NPV}_{Hit}, \qquad (15)$$

where $G(\lambda_{it})$ is the link function for which we test both c-log-log and logit links (Kalbfleisch and Prentice 2002) that are typically applied in discrete-time hazard models. Similar to the model in equation (1) of Bollinger et al. (2022), we test the impact of social comparisons (with neighbors) on adoption by $H_{it}$ and $\text{NPV}_{Hit}$, which represent, respectively, the cumulative number of peer installations and the cumulative NPV of peer installations within the one-mile radial ring. The household characteristic variables, denoted by $X_i$, include the total footprint of the house, ratio of tree cover to footprint, amount of sunlight received by the house, hillshade of the house, and solar irradiance in Watt-hours per square foot

(see Rai and Robinson, 2013 and 2015, for the definitions of these variables). The zip-code fixed effects $\xi_z$ capture market and geographical factors that can also affect beliefs about the economic value of solar. We test polynomial, B-spline, natural spline, and local regression for the temporal effect $\eta_t$. Details about the estimation and results of the benchmark models (a)-(d) are provided in Appendix G.

We estimated benchmarks (a) to (d) on the same 33-quarter training data, predicted the PV adoptions in the following 5 quarters, and compared their out-of-sample performance with our best Model (4)H. Table 5 reports the mean absolute error (MAE) for the four benchmarks alongside Model (4)H. Model (4)H attains the lowest MAE (31.85 adopters) and is therefore superior to every benchmark for out-of-sample prediction. The diffusion process in this data is at an early stage, so the economic and social drivers of individual households, which our model captures, are decisive for forecast accuracy. In contrast, the Bass model (d) uses only aggregate data, and although the hazard model (a) is estimated at the individual household level, neither benchmark embeds the forward-looking structure that our model exploits. Benchmarks (b) and (c), which are "nested" within our model, are also outperformed, underscoring the value of modeling the forward-looking behavior of individual households.

**Table 5 Out-of-sample predictive performance of benchmark models**

| Model | Specification | MAE |
|---|---|---|
| (4)H | Dynamic discrete choice model (best model) | 31.85 |
| (a) | Hazard model (individual level) | 41.45 |
| (b) | Myopic choice model (static NPV logit) | 78.16 |
| (c) | Partially forward-looking model | 48.17 |
| (d) | Bass diffusion model (aggregate data) | 106.17 |

*MAE: mean absolute error (number of adopters). Metrics are computed on the five out-of-sample quarters following the 33-quarter training window; Model (4)H is the dynamic discrete choice model of this paper (Table 3).

## 6. Counterfactual Policy Analysis

Because the estimated structural parameters are invariant to policy changes, our model can evaluate the impact of alternative incentive policies on metrics of interest. For our counterfactual analysis, we focus on adjusting the upfront rebate $r_t$ (\$ per watt installed) that a local government provides to households for installing a PV system and the duration $\tau$ of the rebate program with the goal of improving the program's effectiveness in stimulating PV adoption and reducing carbon emissions, given a limited budget.

Because the rebate budget is limited, rebates are offered only for a limited period. In Austin, the local utility announces its rebate schedule (the program duration and the amount of rebate offered over time) to residential customers well in advance; rebates are set higher initially and decline over time. The declining schedule induces households to advance their installations ahead of the scheduled reductions, a strategic-timing response documented in practice (Reeves and Rai 2018). We consider four key long-run metrics: (i) Cumulative Adoptions; (ii) Cumulative Megawatts (MW) of PV electricity generation installed; (iii) Cumulative Budget spent on rebates; and (iv) Carbon Emissions savings. We compute these metrics at period 56, which is less than six years beyond our estimation window (ending at quarter 33, 2012Q1).

The counterfactual simulations initialize on the first 16 quarters of observed adoption and generate new adoption from period 17 onward.

Our data show that, as PV technology costs decline, government rebates also tend to decrease correspondingly (Figures 1 and 2(a)), resulting in a net cost that remains relatively flat over time. Accordingly, to limit the number of policies to explore, we consider rebate schedules that generally maintain a constant net cost but are 'switched,' i.e., reduced in a few large steps, during the rebate duration. We organize the counterfactual analysis around the program designer's problem: choosing a rebate schedule to maximize cumulative adoption, installed capacity, or carbon savings subject to a budget. Each subsection answers one design question: Section 6.1 determines when a program with a given initial net cost should terminate (a one-step rebate switching policy); Section 6.2 examines how the initial rebate level shifts the attainable outcomes; Section 6.3 assesses whether reducing the rebate in two or three steps outperforms a single switch; Section 6.4 evaluates rebate policies differentiated across household segments; the section closes with an out-of-sample validation and a robustness check of the simulation horizon's ending period. Appendix C presents the simulation algorithm we developed for our counterfactual analysis; a key feature is the computation of the dynamic equilibrium across all households, where the state transition function for the number of neighboring adopters is determined by the equilibrium between the beliefs about future adoption and actual adoption across all households.

**6.1. One-Step Rebate Switching Policies**

In this first scenario, we set the rebate level so that the net cost starts at a "low" initial value of \$1.34 per watt, which is the net cost in the early periods in our dataset, and rises to a "high" value of \$2.49 per watt, the long-term PV cost (after it has stabilized), when the rebate program ends. The key decision variable is the timing of the "switch" from the low to the high net cost. We parametrize the alternative policies by the the Rebate Switching Period (RSP) and denote the one-step policy that switches at period $\tau$ as RSP($\tau$), for $\tau = 20, 24, \ldots, 52$, spanning early to late switching at 4-period (1-year) intervals over the 56-period (quarterly) horizon. If the RSP exceeds period 38, when the PV cost has stabilized (see Figure 2(a)), the rebate reduces from \$(2.49 − 1.34) = \$1.15 per watt to zero in this period; if $\tau < 38$, the rebate remains positive (since PV cost exceeds \$2.49) until period 38, at which time it reduces to zero.

The simulated adoption paths reveal the first feature a program designer must anticipate: households race a pre-announced deadline. Under each policy, cumulative adoption follows the typical diffusion-of-innovation curve, driven by initial lower net costs and peer influence, but the adoption rate rises as the pre-announced RSP approaches and drops once the rebate ends, producing a "kink" in the adoption curve at the switching period (Figure 4(a)). Appendix D presents the adoption trajectory under RSP(40), which reaches about 8,000 households by the end of the horizon, or about 19% of the total potential adopters. The acceleration just before rebate termination stems from the interplay between two features of our model: forward-looking households and neighborhood effects. From equation (1), the adoption utility in each period is increasing in the number of neighboring households that have installed PV systems.

Without rebates, a household with a lower value of $\xi_i$ (household-dependent heterogeneity) will tend to wait until a sufficiently large number of neighbors have adopted. As the pre-announced termination date approaches, however, the benefit of the lower net cost outweighs the benefit of waiting for more neighboring adopters, and so the adoption curve starts out flatter but accelerates as the RSP approaches.

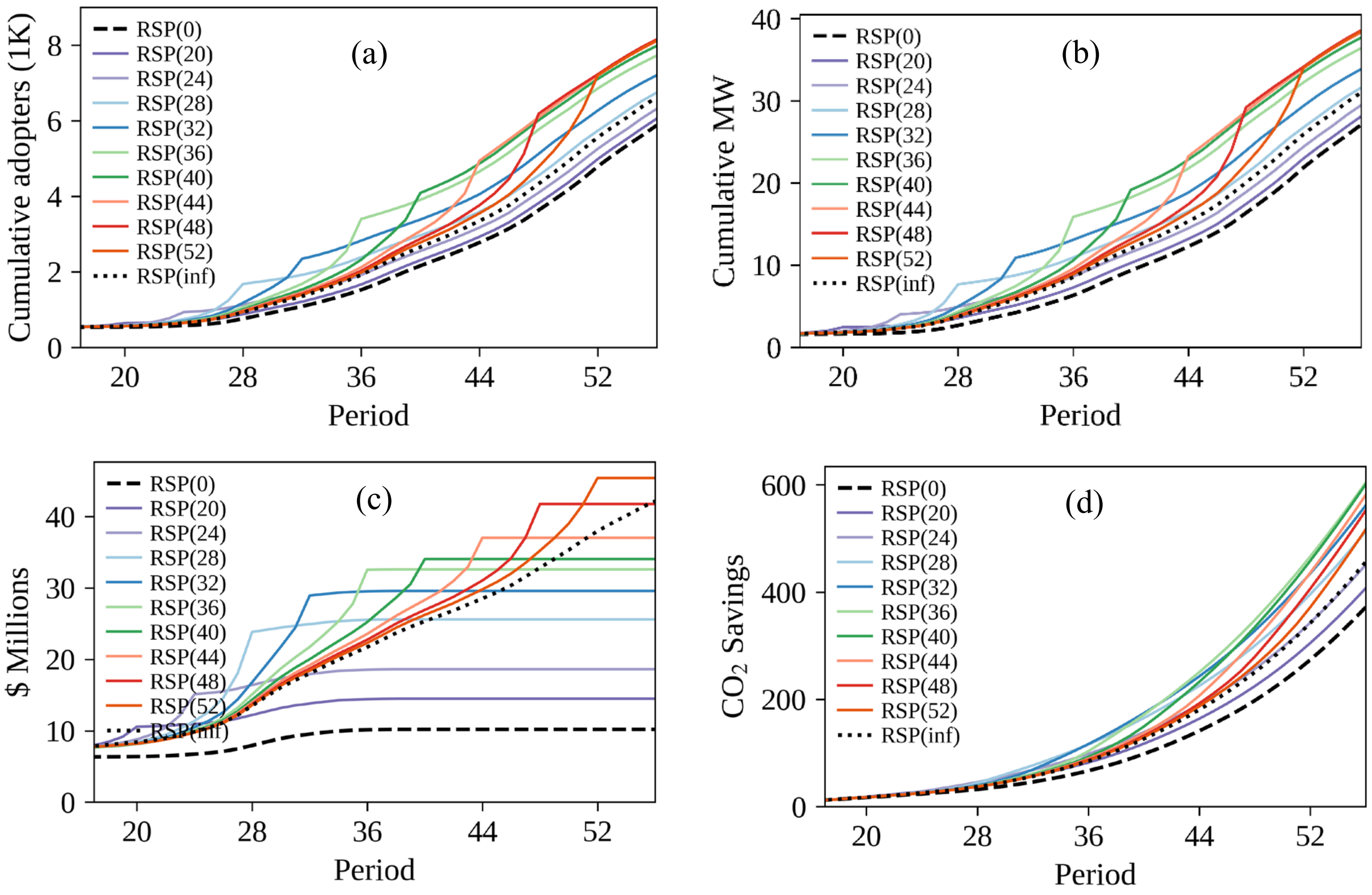


**Figure 4 Performance comparison of alternative RSP policies: (a) Cumulative Adopters. (b) Cumulative MW installed. (c) Cumulative Budget spent for rebates. (d) Cumulative Carbon Savings**

Next, we compare the long-run metrics across the one-step RSP policies, including two extreme benchmarks: RSP(∞), which offers rebates indefinitely, and RSP(0), which offers no rebate. Figures 4(a)–4(d) show, respectively, the cumulative number of adopters, cumulative MW installed, cumulative budget spent on rebates, and total carbon savings over the time horizon for the different policies evaluated. As expected, RSP(0) has the lowest adoption curve among all the policies because it entails the highest net cost for households. But, surprisingly, the adoption curve of RSP(∞) is dominated by all the other policies with fixed rebate switching periods. Although RSP(∞) offers the lowest net cost over time, households are more likely to wait because the rebates never expire; as more households wait, the neighborhood effects weaken, further reducing adoption. Offering the rebate indefinitely therefore leads to lower adoption over time and is a dominated policy.

Comparing two policies RSP($\tau$) and RSP($\tau'$) with $\tau < \tau'$ explains the orderings in Figure 4. At any period before $\tau$, RSP($\tau$) offers the rebate for fewer remaining periods, so a household is more likely to install under RSP($\tau$), and these earlier installations strengthen the neighborhood effects, further

accelerating adoption. After period τ, however, RSP(τ′) still offers rebates and may eventually overtake RSP(τ) in cumulative adopters and MW installed; for example, RSP(44) overtakes RSP(40) shortly after period 40, whereas RSP(48) dominates RSP(52) in all periods (Figures 4(a) and 4(b)). Ending the rebate either too early or too late therefore fails to maximize adoption, and longer rebate durations entail higher total budgetary expenditures (Figure 4(c)).

To gain further insight into the trade-off between incentivizing higher adoption and limiting the rebate budget, Figures 5(a) and 5(b) respectively plot the cumulative MW installed and carbon savings versus the cumulative budget spent at the end of the horizon, where each rebate switching period τ is labeled within a circle on the curve, The solid red curve corresponds to the baseline initial net cost of $1.34 per watt; Section 6.2 discusses the other two curves. Installed capacity is concave in the budget: RSP(48) achieves the highest MW installed while incurring a lower rebate budget than RSP(52), and increasing the total budget beyond RSP(36) or RSP(40) yields diminishing marginal returns in cumulative MW installed. Carbon savings, in contrast, peak at the interior policy RSP(36). Since carbon savings are proportional to the total electricity generated by the installed PV systems, the area under the cumulative MW curve is a proxy for carbon savings; the adoption curve of RSP(36) lies above those of both earlier and later RSP policies during much of the time horizon, increasing this area even though later policies install slightly more capacity by the end of the horizon. Hence, if the main goal is to increase cumulative carbon savings, a shorter rebate program with a lower budget is preferred.

To verify that the policy implications are robust to the choice of simulation horizon, we conduct additional analyses in which we vary the end period (current period 56) used to measure installed capacity, carbon savings, and budget expenditures. Specifically, we reconstruct the cost-effectiveness frontier at five horizons, using periods 40 through 56 as the end of horizon, in four-period increments. The policy implications hold at every horizon: an intermediate switch yields the most carbon savings per budget dollar, installed capacity exhibits diminishing returns to budget, and deferring the switch past the interior optimum raises spending while lowering cumulative carbon. These findings demonstrate that the policy insights obtained from our model are not sensitive to the end period used to evaluate policy performance. Appendix I provides further details.

### 6.2. Diffusion and Budgetary Implications of Varying the Initial Net Cost

Extending the analyses of RSP(τ) policies, we next vary the initial net cost (while retaining one-step switching) to examine how the rebate level shifts the attainable outcomes and to identify the least-budget design for a given capacity or carbon-savings target.

We conduct counterfactual policy simulations for the RSP(τ) policies using initial net costs that are 20% lower ($1.072 per watt) and 20% higher ($1.608 per watt) than the baseline initial net cost of $1.34 per watt from the data. Figure 5 plots the cumulative MW installed and carbon savings versus budget for the lower (dotted blue) and higher (dashed green) initial net costs alongside the baseline (solid red). The relationships retain the same shape or pattern across the different initial net costs, confirming the

robustness of our results to the initial rebate value. As we would expect, offering the highest rebate (lowest net cost) yields the maximum penetration and highest carbon savings, but requires large budgets.

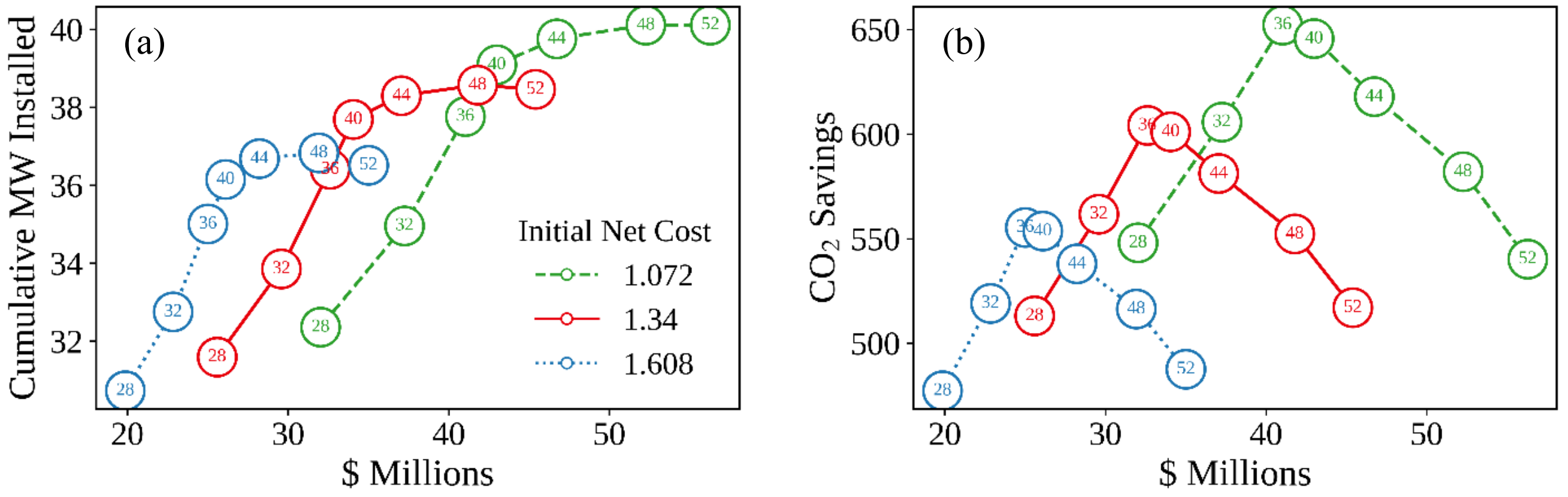


**Figure 5 (a) Cumulative MW Installed vs Budget, and (b) Cumulative Carbon savings vs Budget**

The frontier curves in Figure 5 let policymakers tune rebate policies to their objectives. For instance, suppose a policymaker has a fixed budget of, say $35 million, and wishes to maximize penetration. Then, offering an intermediate rebate level (initial net cost of $1.34) with an RSP of between 40 and 44 is most effective. Similarly, with a $35 million budget, carbon savings are maximized with moderate initial rebate and RSP around 36. That is, in the face of this budget constraint, a moderate incentive with moderate duration is better than offering a low initial rebate (high initial net cost of $1.608) for a longer duration or a high initial rebate (low initial net cost of $1.072 per watt) for a shorter duration. Of course, when budgets are highly restrictive, a low initial rebate becomes necessary. If policymakers aim to achieve a target cumulative MW installed, say 36 MW, then offering a lower rebate with a switching period of $\tau$ = 40 achieves the goal with the lowest budget. However, a higher target, say 40 MW, is attainable only by providing a high initial rebate with RSP of 48. Interestingly, for all three initial net cost alternatives, RSP(36) or RSP(40) represents the inflection point in the plot of MW installed versus total budget and also achieves the highest carbon savings. In other words, these limited duration policies best ensure a high marginal return on the incremental rebate budget, regardless of the initial value of the net cost.

### 6.3. Insights from the Multi-Step Rebate Switching Policy Simulations

Can multi-step rebate switching policies that increase the net cost in two or three steps, rather than just one, achieve better performance in terms of diffusion, carbon savings, and budget? We simulate the following two rebate switching schemes that are viable in practice:

i. A two-step policy, RSP($\tau - 4, \tau + 4$), that switches the rebate twice over 2 years (8 quarters) in periods $\tau - 4$ and $\tau + 4$ such that the initial net cost is $1.34 per watt, $1.92 after the first switch in period $\tau - 4$, and $2.49 after the second switch in $\tau + 4$. For an appropriate comparison of the switching policies, we match each two-step policy RSP($\tau - 4, \tau + 4$) with the corresponding one-step switching policy RSP($\tau$). For example, we compare the RSP(28) policy with the RSP(24,32) policy.

ii. A three-step policy RSP($\tau - 8, \tau, \tau + 8$) that switches the rebate three times (once every 2 years), resulting in three equal increments in the net cost over 4 years ($1.34 per watt initially, $1.72 after the first switch in period $\tau - 8$, $2.10 at $\tau$, and $2.49 after the final switch at $\tau + 8$).

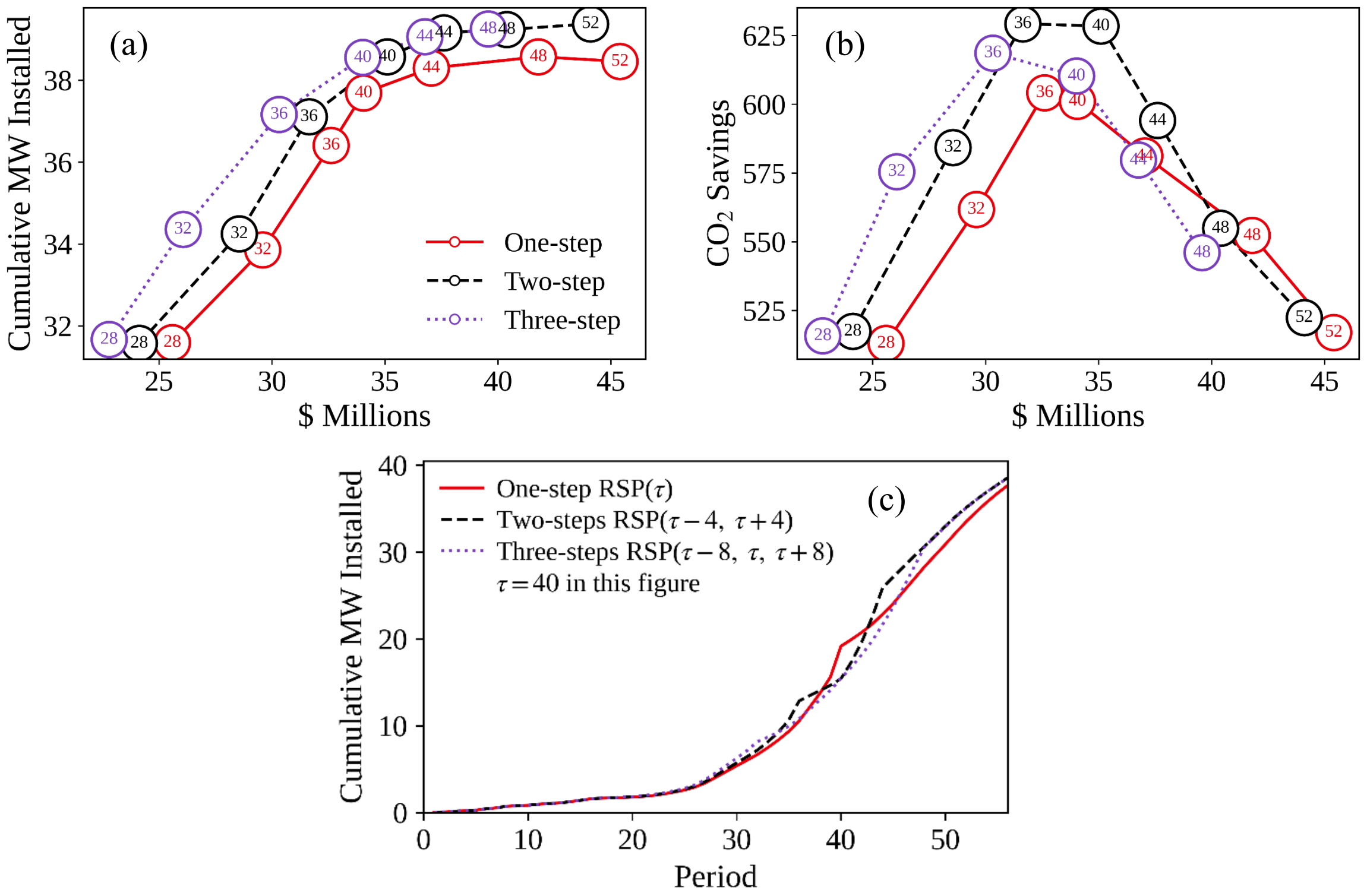


**Figure 6 (a) Cumulative MW Installed vs Budget; (b) Cumulative Carbon savings vs Budget for multi-step policies; (c) Cumulative MW installed over time for single and multi-step RSP policies with $\tau = 40$**

Figure 6 compares the total installed capacity (in MW) and total carbon savings at the end of the horizon versus the total budget for the one-step, two-step, and three-step switching policies, where τ is again marked within a circle on the curve. Two decision rules emerge. First, gradually reducing the rebate economizes the budget mainly for early-ending programs: among the policies that end the rebate in early periods, the one-step RSP(28) achieves a cumulative installed capacity of 31.6 MW with a total budget of $25.6 million, whereas the two-step RSP(24,32) achieves 31.6 MW with $24.1 million and the three-step RSP(20,28,36) achieves 31.7 MW with $22.8 million. For higher RSPs, the penetration advantage of multi-step policies at a given budget is small. Figures 6(a) and 6(b) confirm that a longer and more costly program may not add adopters or yield higher carbon savings even when the rebate declines gradually. Second, the best policy for maximizing carbon savings is the two-step RSP(36,44) policy rather than the three-step RSP(32,40,48) policy, because the anticipation mechanism of Section 6.1 operates in reverse: smaller rebate reductions in the three-step policy dampen the acceleration in adoptions before each switch. RSP(32,40,48) therefore traces a smoother and lower adoption curve than RSP(36,44) in the earlier periods (Figure 6(c)), implying lower carbon savings (as indicated by the area under the curves), even though the two policies achieve about the same cumulative installed MW by the end of the horizon.

### 6.4. Differentiated Rebate Policies

The segment structure of our model allows us to evaluate rebate policies that vary by segment, a design dimension that aggregate models cannot capture. We examine differentiation separately by geographic and by economic segment, policies that may be implementable in practice (e.g., different rebates for urban, suburban, and rural locations).

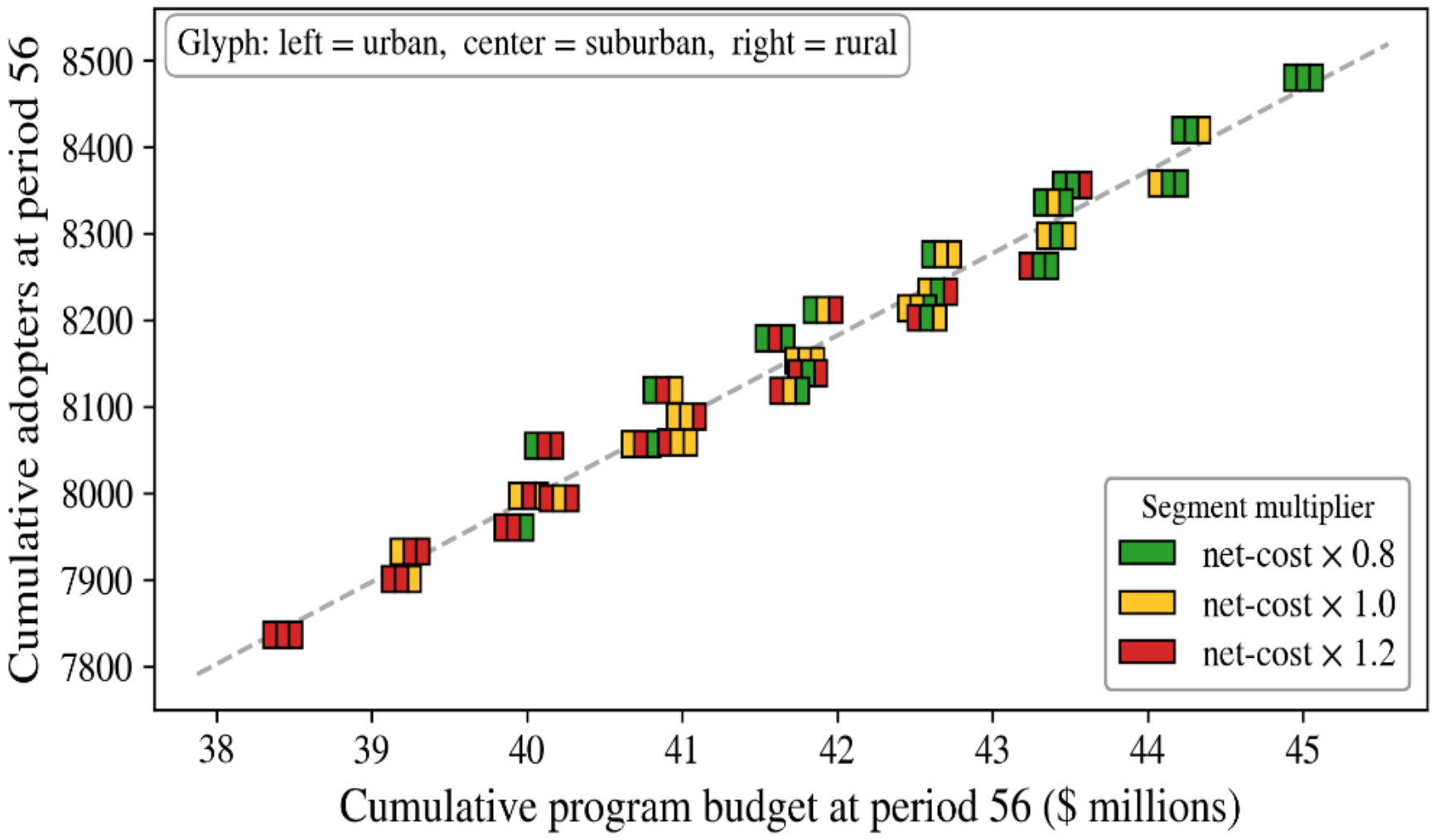


**Figure 7 Cumulative adoption versus program budget** at period 56 for the 27 geographic-targeting policies. Each point is one policy; the three squares of its glyph give the urban (left), suburban (center), and rural (right) net-cost multiplier, colored green for 0.8, yellow for 1.0, and red for 1.2

For geographic differentiation, we consider three possible rebate levels – with respective net costs of 80%, 100%, or 120% of a nominal value – for each segment. Accordingly, we simulate all combinations of net-cost multipliers in ({0.8, 1.0, 1.2}) applied separately to urban, suburban, and rural households, yielding 27 policies in total. A lower multiplier further reduces the net cost and therefore corresponds to a higher rebate rate. Figure 7 plots cumulative adoption against the rebate program's budget for each simulated policy. At each point, the three-square glyph encodes, from left to right, the net-cost multipliers for urban, suburban, and rural households, with green denoting 0.8 multiplier, yellow denoting 1.0, and red denoting 1.2. For instance, with a budget of $ 42 million, the policy that achieves maximum adoption is a green-yellow-red glyph, i.e., a policy that provides the highest rebate to urban households, intermediate rebate for suburban households, and lowest rebate to rural households. As expected, cumulative adoption rises almost linearly with the total rebate budget. The key policy question is which geographic segment should receive a higher rebate in order to improve overall adoption for a given budget. Figure 7 shows that the three-square glyphs above the fitted relationship generally have the urban square, on the left, in green, indicating a 0.8 net cost multiplier and thus a higher rebate for urban households. For example, at a total budget of about $41 million, the best differentiated rebate policy offers a higher rebate to the urban segment by setting its net cost at 0.8 times the observed level, while setting the suburban segment's net cost at 1.2 times and the rural segment's net cost at 1.0 times the

observed level. Notice that the observed rebate policy (that gives uniform rebates for all segments) is represented in Figure 7 by the glyph with all three squares in yellow. At a total budget of about $42 million, a differentiated policy that offers a higher rebate to the urban segment (in green), keeps the suburban segment's net cost at the observed level (in yellow), and assigns a lower rebate to the rural segment (in red) can increase cumulative adoption from around 8,150 to 8,210 households.

We also evaluate differentiation by economic segment, simulating the same set of net-cost multipliers applied separately to households with home values below and above $300,000. The corresponding figure (similar to Figure 7) in Appendix H shows that almost all two-square glyphs, with the left square corresponding to homes below $300,000 and the right square to homes above $300,000, fall along a nearly straight line in the relationship between cumulative adoption and budget. This pattern indicates that differentiating rebates by economic segment does not substantially improve adoption for a fixed budget. Appendix H details the factorial designs and summary statistics for both sets of simulations.

Table 6 consolidates the counterfactual findings of Sections 6.1 to 6.4 into design guidance organized by the policy objective. The comparisons hold the model structure and the estimated parameters fixed, so the table ranks policy designs rather than forecasting absolute levels.

**Table 6 Rebate program design guidance by policy objective**

| Policy objective | Recommended design | Rationale [Section number] |
|---|---|---|
| To maximize installed capacity | Selecting the single-step rebate policy RSP(48) or two-step policy RSP(44, 52) | Capacity is concave in the budget; among single step policies, RSP(48) attains the highest MW at a lower budget than RSP(52) [6.1] |
| To maximize carbon savings | Two-step switching policy RSP(36,44); one-step alternative RSP(36) | Carbon savings peak at an interior switch; smaller steps dampen anticipatory adoption [6.1, 6.3] |
| To meet a capacity target at the lowest budget | Selecting the lowest-budget point on the Figure 5 frontiers; e.g., for 36 MW use initial net cost $1.608 and $\tau = 40$ | The frontiers shift with the initial rebate level [6.2] |
| Subject to limited early-ending budget | Multi-step reductions; e.g., RSP(20,28,36) attains 31.7 MW for $22.8 million versus 31.6 MW for $25.6 million under RSP(28) | Gradual steps economize the budget for early-ending programs [6.3] |
| Geographic differentiation by urban, suburban, rural segments | Offering the higher rebate to the urban segment improves installed MW for a given budget | Urban-discount policies lie above the budget-adoption line [6.4] |
| Differentiate rebates by home value | Not effective. Installed MW tracks budget linearly | Economic differentiation does not improve adoption for a fixed budget [6.4] |
| Any objective | Avoid open-ended rebates: RSP(∞) is dominated | Households defer adoption, weakening neighborhood effects [6.1] |

The metrics reported in the counterfactual analyses above are evaluated at period 56 (2017 Q4), which is 6 years beyond the end of our estimation window. The actual installations in that quarter in Austin were about 7.4% less than the number of adopters for our best policy (RSP(48)). These numbers are comparable, implying that our model and methodology are credible, but they also demonstrate that optimizing the incentive policy can yield notable policy benefits.

## 7. Conclusions

Electricity generated by residential PV systems can lower households' utility bills while reducing carbon emissions, yet adoption remains in the early phase of its diffusion curve, held back by high purchase and installation costs and by consumers' uncertainty about the benefits. Effective incentive programs should therefore focus on reducing net costs and leveraging peer influence among consumers.

To guide the design of such programs, we developed a dynamic structural model of individual households' adoption decisions, estimated on a 9-year household-level dataset, that captures the economic benefits, neighborhood effects, forward-looking strategic behavior, and unobserved heterogeneity in households' predisposition towards solar energy across economic and geographic segments. The estimates and out-of-sample prediction performance validate the importance of peer influence from nearby installations and the inclusion of forward-looking behavior in the model. Indeed, adding these factors reduces the out-of-sample adoption prediction error compared to benchmark models. Moreover, the structural model permits the evaluation of counterfactual policies.

Our counterfactual simulations show that extending the rebate indefinitely is not optimal for either diffusion or budgets since forward-looking households postpone their adoptions, weakening neighborhood effects. For policymakers with both goals of increasing adoption and carbon savings with a limited rebate budget, the simulations characterize the trade-off curve from which agencies can select their preferred policy.

This structural framework can be applied to other durable technologies that share four key features, namely that (1) adoption requires a large one-time investment with positive NPV due to ongoing benefits; (2) public policy incentives, such as rebates and credits, shape the effective price path; (3) adoption decisions are observable to nearby decision-makers, generating accumulated local peer exposure; and (4) decision-makers are forward-looking, timing their adoption in response to anticipated prices and incentives. These features are well-recognized in the literature on durable goods adoption, where strategic consumers time their purchases in response to market dynamics. For instance, appliance buyers retime purchases to capture rebates (Houde and Aldy 2017), and air-conditioner buyers trade off purchase prices against discounted future operating costs (Rapson 2014). The literature on hybrid and electric vehicles illustrates all four features. Hybrid-vehicle buyers time purchases around pre-announced tax-credit step-downs (Sallee 2011), and both the type of incentive offered, its size, terms of payment, and charging station density matter for adoption (Gallagher and Muehlegger 2011, Li et al. 2017, Springel 2021, Avci et al. 2015, Shi et al. 2022). Visibility matters as well, since car purchases increase purchase likelihood most where they are visually observable (McShane et al. 2012) and consumers pay a premium for hybrid models whose distinctive design visibly signals their environmental commitment (Sexton and Sexton 2014). Beyond vehicles, profitability-driven adoption with learning from neighbors is documented for agricultural technologies (Foster and Rosenzweig 1995, Conley and Udry 2010), and forward-looking benefits shape mobile telephony adoption (Björkegren 2019). To our knowledge, none of the literature

above incorporates all four features in an empirical model. Our modeling framework can be adapted to provide policy support for incentive-driven adoption of durable goods with these features, such as including electric vehicles, home chargers, heat pumps, home emergency generators, and battery storage. Table 7 summarizes the simple modifications that two representative applications, electric vehicles and heat pumps, require; the estimation and counterfactual machinery remains unchanged.

**Table 7 Potential applications of the dynamic structural model**

| Model component | Solar PV (this paper) | Electric vehicles | Heat pumps |
|---|---|---|---|
| NPV calculation including price trajectory, incentives and operating savings | System cost per watt; rebate $/W, ITC; electricity bill savings via production | Vehicle price; purchase rebate or tax credit; fuel-minus-electricity savings per mile driven | Equipment plus installation cost; rebate or credit; heating-fuel or peak-price savings |
| Peer influence due to visible adoption and proximity | Rooftop panels visible from adjoining streets; counts of neighbor adopters | Vehicles in driveways and local charging stations; neighbor counts of registered EVs | Outdoor compressor units visible at grade; contractor yard signs |
| State variables and transitions | Net cost and neighbor-adopter count variations over time | Net vehicle premium trajectory; neighbor EV count; charging station density | Net equipment cost path; neighbor installation count |
| Incentive optimization | Rebate magnitude, duration, steps; geographic and income targeting | Point-of-sale rebate vs. credit; step-down schedule; targeting by income or local charger density | Rebate schedule; seasonal timing; targeting by housing vintage or fuel type |
| Data requirements | Household-level installation and rebate records, house characteristics, cost and incentives | Vehicle registrations at individual level; incentives and energy prices | Installation and rebate records at address level; incentives and fuel prices |

Future research may incorporate additional household information, such as prior investment in other green technologies or exposure to communication channels such as social media, to reduce the unobserved heterogeneity and improve model fit. Other fruitful directions include validating the PV adoption model on data from other cities, less densely populated areas, or developing countries; developing seeding strategies for faster diffusion; using our model as the basis for determining optimal incentive policies; and applying it to the diffusion of other capital-intensive technologies with peer influence.


\endgroup


\newcommand{\MyColaWidth}{1.53}
\newcommand{\MyColbWidth}{4.5}

\newcommand\smallO{
\mathchoice"
  {{\scriptstyle\mathcal{O}}}
  {{\scriptstyle\mathcal{O}}}
  {{\scriptscriptstyle\mathcal{O}}}
  {\scalebox{.7}{$\scriptscriptstyle\mathcal{O}$}}
}
\newcommand\smallM{
\mathchoice
  {{\scriptstyle\mathcal{M}}}
  {{\scriptstyle\mathcal{M}}}
  {{\scriptscriptstyle\mathcal{M}}}
  {\scalebox{.7}{$\scriptscriptstyle\mathcal{M}$}}
}

\newcommand{\sProperties}{I}

\newcommand{\pvSize}{w}
\newcommand{\pvGen}{q}
\newcommand{\df}{\beta}

\newcommand{\pvPrice}{p}
\newcommand{\rebate}{r}
\newcommand{\pvCost}{c}

\newcommand{\vos}{v}
\newcommand{\instBase}{H}
\newcommand{\sInstBase}{h}

\newcommand{\latentUtility}{\varepsilon}
\newcommand{\randomEffects}{\xi}

\newcommand{\pBaseUtility}{\rho}
\newcommand{\pMMoney}{\alpha}
\newcommand{\pMPE}{\gamma}

\newcommand{\fValue}{V} 
\newcommand{\fValueCE}{\mathcal{V}}
\newcommand{\npv}{\text{NPV}} 

\newcommand{\pSeries}{\lambda}




\newcommand{\instBaseNew}{H^{\texttt{new}}}
\newcommand{\funCost}{C}
\newcommand{\pLossFactor}{l}
\newcommand{\psigma}{\sigma}
\newcommand{\funUtilityElectricity}{u}
\newcommand{\stateElectricityRate}{c}
\newcommand{\pPremium}{\rho}
\newcommand{\iLocation}{x}
\newcommand{\pAppraisedVal}{AVal}
\newcommand{\dfTS}{F}
\newcommand{\stParDistance}{\theta}
\newcommand{\stateReliability}{R}
\newcommand{\statePerEffect}{X}
\newcommand{\randC}{\xi}
\newcommand{\stateConsumption}{q}
\newcommand{\iTime}{t}
\newcommand{\iPropertyi}{i}
\newcommand{\iPropertyj}{j}
\newcommand{\sTime}{\mathcal{T}}
\newcommand{\sWindow}{\mathcal{W}}
\newcommand{\sSubsetWindow}{\mathcal{O}}
\newcommand{\iCell}{\mathcal{A}}
\newcommand{\sX}{\mathcal{X}}
\newcommand{\sY}{\mathcal{Y}}
\newcommand{\Property}{P}
\newcommand{\sInstallers}{L}
\newcommand{\iInstaller}{l}
\newcommand{\nBlockGroups}{{Blocks}}
\newcommand{\sBlockGroups}{\mathcal B}
\newcommand{\BlockGroup}{B}
\newcommand{\iBlockGroup}{m}
\newcommand{\sSegmentationTypes}{\mathcal S}
\newcommand{\sSegments}{S}

\newcommand{\sAdopters}{A}
\newcommand{\pTimeDecision}{\hat{t}}
\newcommand{\pTimeInstallation}{\tilde{t}}
\newcommand{\pPricequote}{d}
\newcommand{\pInstalledCost}{c}
\newcommand{\pRebate}{r}
\newcommand{\pEstYrkWhSave}{ywhs}
\newcommand{\pRebatekW}{rkw}
\newcommand{\pRebateLevel}{rlevel}
\newcommand{\pCellArea}{C_A}
 
\newcommand{\pBuldingArea}{bArea} 
\newcommand{\pResidentialOrCommercial}{ifRes} 
\newcommand{\pIfCanHaveASolarPanel}{ifCanSP} 

\newcommand{\pHouseholdsBlockGroupSegment}{h} 	

\newcommand{\pRetailElectricityPrices}{e} 	

\newcommand{\nAdopters}{N}
\newcommand{\Rate}{R}
\newcommand{\intensity}{\lambda}
\newcommand{\spatialOffset}{\mu}
\newcommand{\temporalOffset}{\nu}
\newcommand{\densityInstalations}{b}
\newcommand{\gausianProcess}{u}
\newcommand{\gausianProcessCommon}{v}
\newcommand{\influenceBetweenSegments}{\gamma}
\newcommand{\influenceSegment}{\beta}

\newcommand{\fSpatialOffset}{\mu}
\newcommand{\fCombine}{\nu}
\newcommand{\fTime}{g}
\newcommand{\fSpatial}{f}

\newcommand{\pSpatialOffset}{\beta}
\newcommand{\pPowerTime}{q}
\newcommand{\pScaleTime}{\alpha}
\newcommand{\pSegmentEffect}{\gamma}
\newcommand{\pScaleSpatial}{\xi}

\newcommand{\stateHouseCharacteristics}{y}
\newcommand{\stateHomePremium}{HP}
\newcommand{\probabilityOfView}{p}
\newcommand{\vRan}{\xi}
\newcommand{\stateSpace}{\mathbb{S}}

\newcommand{\pdfTS}{f}
\newcommand{\pdfLU}{h}
\newcommand{\fUtility}{U}
\newcommand{\fUtilityObs}{u}
\newcommand{\fUtilitySubjective}{\omega}

\newcommand{\fPolicy}{\alpha}





\RUNTITLE{Accelerating Residential Solar Adoption}

\TITLE{Accelerating the Adoption of Residential Solar Power Systems: Policy Analysis using a Dynamic Structural Model}

\ARTICLEAUTHORS{%
\AUTHOR{Sebastian Souyris}
\AFF{Lally School of Business, Rensselaer Polytechnic Institute, NY 12180, 
\EMAIL{souyrs@rpi.edu}}	
\AUTHOR{Jason Duan, Anant Balakrishnan}
\AFF{McCombs School of Business, The University of Texas at Austin, Austin, TX 78712, \EMAIL{\{jason.duan@mccombs.utexas.edu, anantb@mail.utexas.edu\}}} 
\AUTHOR{Varun Rai} 
\AFF{LBJ School of Public Affairs, The University of Texas at Austin, Austin, TX 78712, \EMAIL{varun.rai@mail.utexas.edu}}
} 

\ABSTRACT{}

\KEYWORDS{renewable energy, diffusion, dynamic discrete choice model, Bayesian statistics} \HISTORY{.}


\ECSwitch

\noindent\textbf{Accelerating the Adoption of Residential Solar Power Systems: Policy Analysis using a Dynamic Structural Model}\\
\noindent\textbf{Electronic Companion}\\ 

\raggedbottom

\setcounter{equation}{0}
\renewcommand\theequation{\thesection.\arabic{equation}} 

\begin{APPENDICES}
%
%
%
%
\section{Model Identification and Rebate Exogeneity}\label{sec:ECidentification}

\subsection{Identification}

The observed data include each household's choice $a_{it}$ in every
period, the economic and geographic segment of household $i$, the
state variables $c_{t}$ and $h_{it}$, and the net present value
$NPV_{it}(c_{t})$. The parameters to be identified are $\theta_{i}=\left\{ \rho,\alpha_{e_{i}},\gamma_{g_{i}},\xi_{i}\right\} $.
Here, we discuss in detail how the parameters in the utility $\theta_{i}=\left\{ \rho,\alpha_{e_{i}},\gamma_{g_{i}},\xi_{i}\right\} $
and the variance $\sigma_{\xi}^{2}$ of the household-specific
unobserved heterogeneity term can be identified from the data.

The state transition functions $p\left(c_{t+1}|c_{t}\right),\text{ }p\left(h_{i,t+1}|h_{it}\right)$
are identified from the observed costs and photovoltaic (PV) installations over time.
For each economic and geographic segment, we observed a large number
(thousands) of households. Hence, for each segment, $p\left(a_{it}|c_{t},h_{it}\right),a_{it}=0,1$
is identifiable semiparametrically (e.g., by a logistic regression
of $a_{it}$ on the polynomials of $c_{t}$ and $h_{it}$). $p\left(a_{it}|c_{t},h_{it}\right)$
is assumed to be the outcome of the parametric structural model, i.e., $p\left(a_{it}|c_{t},h_{it},\theta_{i}\right)=p\left(a_{it}|c_{t},h_{it}\right)$.
Similar to \citet{Hotz1993}, we have
\begin{eqnarray}
	\log p\left(a_{it}=1|c_{t},h_{it},\theta_{i}\right)-\log p\left(a_{it}=0|c_{t},h_{it},\theta_{i}\right) & = & \bar{u}_{i1t}(c_{t},h_{it},\theta_{i})-\beta EV_{i}\left(c_{t},h_{it},a_{it}=0,\theta_{i}\right)\nonumber \\
	\text{where \ensuremath{\bar{u}_{i1t}}(\ensuremath{c_{t}},\ensuremath{h_{it}},\ensuremath{\theta_{i}})} & = & \rho+\alpha_{e_{i}}NPV_{it}(c_{t})+\gamma_{g_{i}}h_{it}+\xi_{i}.\label{eq:A4}
\end{eqnarray}

Because the discount factor $\beta$ in a dynamic discrete choice
model is generally not separately identifiable from the structural
parameters \citep{Magnac2002}, we fix its value. Next, we discuss
how $\theta_{i}=\left\{ \rho,\alpha_{e_{i}},\gamma_{g_{i}},\xi_{i}\right\} $
and $\sigma_{\xi}^{2}$ are identified. In general, by Corollary~3
in \citet{Magnac2002}, the dynamic discrete choice model parameters
are uniquely identified given fixed $\beta$ and the state transition
probabilities (densities), $p\left(c_{t+1}|c_{t}\right),\text{ }p\left(h_{i,t+1}|h_{it}\right),\text{ }p\left(\varepsilon_{iat}\right)$.
We then focus on how the specific parameters in our model are identified.
From equation \eqref{eq:A4}, if $EV_{i}\left(c_{t},h_{it},a_{it}=0,\theta_{i}\right)$
is known, then, for each segment, $\rho,\alpha_{e_{i}},\gamma_{g_{i}}$ are
identified by a linear regression of the log-choice-probability difference
in equation \eqref{eq:A4} adjusted for expected value function, given
that $p\left(a_{it}|c_{t},h_{it},\theta_{i}\right)=p\left(a_{it}|c_{t},h_{it}\right)$
has been consistently estimated. The household-specific $\xi_{i}$
is estimated by the residuals of the regression and $\sigma_{\xi}^{2}$
is identified by the variance of the residuals across all segments.

In practice, $EV_{i}\left(c_{t},h_{it},a_{it}=0,\theta_{i}\right)$
is unknown. Given a candidate value in the estimation iterations,
$\theta_{i}=\left\{ \rho,\alpha_{e_{i}},\gamma_{g_{i}},\xi_{i}\right\} $,
$EV_{i}\left(c_{t},h_{it},a_{it}=0,\theta_{i}\right)$ is the fixed-point
solution of the expected-value recursion in Appendix~C (equation \eqref{eq:emax_approx}). The identification and
estimation therefore follow the nested fixed-point logic of \citet{Rust1987a}. In the subsequent iteration, holding fixed $EV_{i}\left(c_{t},h_{it},a_{it}=0,\theta_{i}\right)$
from the previous step, $\theta_{i}=\left\{ \rho,\alpha_{e_{i}},\gamma_{g_{i}},\xi_{i}\right\} $
is estimated again from equation \eqref{eq:A4}. This iterative procedure
converges to the fixed point of the model \citep{Rust1987a,Magnac2002}, yielding the final estimates of $\theta_{i}=\left\{ \rho,\alpha_{e_{i}},\gamma_{g_{i}},\xi_{i}\right\} $.

\subsection{Institutional Background and Exogeneity Evidence}\label{sec:ECinstitutional}
The net cost that drives adoption in our model is determined in part
by the residential PV rebate, and the institutional record
shows that Austin Energy set this rebate through an administrative
budget process rather than in response to individual adoption decisions.
Austin Energy funded the rebate from a fixed annual budget approved
in the City of Austin budget cycle and disbursed it as a posted dollars-per-watt
rate that applied uniformly to every qualifying applicant. That is,
it did not negotiate the subsidy household by household or condition
it on a neighborhood's adoption history. The same administratively recorded
schedule is described in prior publications \citep{Rai2015,Rai2013}.
When the installation-cost benchmarks used in the budgeting process
changed, Austin Energy revised the posted rate administratively and
announced the new rate publicly for the next fiscal year. The other
components of the net cost include the federal investment tax credit
by statute \citep{EPAct2005,EESA2008} and the Value-of-Solar tariff,
which replaced residential net metering on October~1, 2012, through
a separate citywide proceeding \citep{AustinEnergyVOS2017,DSIRE5669}.
The observed posted rebate rates in our data display a small number
of discrete downward revisions over the time window. The program started
in 2004 at \$5.00 per watt and then held near \$4.50 through fiscal
years 2005--2008 \citep{SolarAustin2011}. At the beginning of fiscal
2009 the rate fell to \$3.75. It was then cut to \$2.00, effective
immediately, on June~11, 2012, and to about \$1.50 in mid-2013.

Two features of this process support treating the net cost as exogenous
to the individual adoption decision. First, the rate changes were
infrequent, discrete, and downward-only, and each was triggered by
a documented decline in installation cost (a 17\% decline in 2010
and an 11\% decline over the first half of 2011), rather than by any
household's characteristics or by local adoption \citep{SolarAustin2011}.
Second, the rate was posted and applied uniformly, so an individual
applicant could neither bargain over the subsidy nor reliably time
it.

One remaining concern is that Austin Energy may set the annual rebate
budget and rate in anticipation of aggregate demand shocks in the
following year. To further assess whether the rebate rate could be
endogenous and bias our estimates, we conduct an empirical test of
whether household net present value $\npv_{it}$ has different effects
on PV adoption depending on whether we include period-specific temporal
fixed effects that control for unobserved demand shocks over time
\citep[Chapter~10]{Wooldridge2010}. In our model, $\npv_{it}$ is the
only variable that is a function of the rebate rate. Because the dependent
variable of installation is a quarterly dummy variable over time, we
use the following reduced discrete hazard model
\[
G\left(\lambda_{it}\right)=\xi_{z}+\eta_{t}+\beta_{1}\npv_{it}+\beta_{2}H_{it}+\beta_{3}X_{i},
\]
where the hazard rate $\lambda_{it}$ of installation is modeled by
the complementary log-log ($\log[-\log[1-\lambda_{it}]]$) function.
Here, $\eta_{t}$ is the fixed effect that captures the unobserved
aggregate temporal demand shocks for all households in the city, $\xi_{z}$
is the zip-code fixed effect, $X_{i}$ represents household characteristics,
and $H_{it}$ is the number of neighborhood installations within the
one-mile radial ring. If the rebate rate is endogenous because of
the anticipation of aggregate demand shocks, we should expect $\npv_{it}$
to be correlated with $\eta_{t}$. Therefore, the model that omits $\eta_{t}$,
\[
G\left(\lambda_{it}\right)=\xi_{z}+\beta_{1}\npv_{it}+\beta_{2}H_{it}+\beta_{3}X_{i},
\]
should produce a materially different estimate of the coefficient $\beta_{1}$.
However, our estimation of these two models yields estimates of $\beta_{1}$
that differ by less than one standard error (Table~\ref{tab:ECexogeneity}), which indicates that the
variation in $\npv_{it}$ identifying the price response is orthogonal to
the citywide temporal shocks captured by $\eta_{t}$.

\begin{table}[H]
	\begin{center}
		\caption{Sensitivity of the net-present-value coefficient $\beta_{1}$ to temporal fixed effects in the discrete hazard model.\label{tab:ECexogeneity}}
		{\scriptsize
			\begin{tabular}{lrrr}
				\toprule
				$\npv_{it}$ coefficient & Estimate & S.E. & p-value\bigstrut\\
				\midrule
				Model with temporal fixed effects & $2.36\times10^{-6}$ & $0.085\times10^{-6}$ & $\simeq0$\\
				Model without temporal fixed effects & $2.31\times10^{-6}$ & $0.072\times10^{-6}$ & $\simeq0$\\
				\bottomrule
			\end{tabular}
		}
	\end{center}
\end{table}
\clearpage
\section{Proof of Consistency of Using Cluster-Based Sampling}

\label{sec:Proofclusters} In Section 4.2, we aggregate households
into clusters with similar characteristics to reduce the computational
cost of the estimation. Let the total sample size of households in
a city be $M.$ These households are partitioned into $K$ clusters,
each cluster is denoted as the set $M_{k}$, whose size is $m_{k}$.
Observe that, if the data of a household used in the likelihood-based
estimation is not from any cluster, it can be viewed as in a cluster
of size one by itself. Therefore, all the households can be viewed
as in a cluster $M_{k},k=1,...,K$. We denote household $i$'s choice
as $a_{it}\in\{0,1\}$. Notice that the cluster size $m_{k(i)},i\in M_{k}$
depends on the choice $a_{it}$. If $a_{it}=1$ for some $t=1,...,T$,
then the size $m_{k(i)}=1$. If $a_{it}=0$ for all $t=1,...,T$,
then the cluster size $m_{k(i)}$ can be greater than one. We fix the cluster composition throughout the observation window, so the cluster sizes do not depend on time $t$. 

Each household has a set of characteristic variables $X_{it}$ including
economic segment, geographic segment and NPV. We consider $X_{it}$'s
as random variables independent across households, i.e., $X_{it}=\left[X_{i1},...,X_{iT}\right]\overset{iid}{\sim}p\left(x_{1},...,x_{T}\right)$,
which is the distribution of household characteristics in the population.
The function $p\left(x_{1},...,x_{T}\right)$ can represent a pdf
if $x_{t}$ is continuous or a pmf if $x_{t}$ is discrete. (The corresponding
$\intop p\left(x_{1},...,x_{T}\right)dx_{1}\cdots dx_{T}$ is integrated
with respect to either a Lebesgue measure or a counting measure).
The choice probability defined by Equation (13) of the paper can be simplified
as $p_{iat}=p\left(a|x_{it};\vartheta\right),a=0,1,$ where $\vartheta$
represents the set of model parameters that we are interested in estimating.

The full log-likelihood function if we use the entire sample in the
estimation is 
\begin{equation}
	L^{a}=\sum_{i=1}^{M}\left\{ \sum_{t=1}^{T}\log p\left(a_{it}|x_{it};\vartheta\right)\right\}. 
\end{equation}

If we employ the cluster-based sampling method above and randomly
choose one household from each cluster, the log-likelihood based on
our cluster-based sampling method is 
\begin{equation}
	L^{c}=\sum_{k=1}^{K}\left\{ m_{k(i)}\sum_{t=1}^{T}\log p\left(a_{it}|x_{it};\vartheta\right)\right\} ,\label{eq:cluster-loglkhd}
\end{equation}
where $k\left(i\right)$ indicates the cluster $k$ from which household
$i$ is selected at random ($i\sim \text{Unif}(M_k)$).

\begin{proposition}
	Assume the probabilistic choice model described in Section 4.1 and
	the cluster-based sampling method described above. Let $L^{a}$ and
	$L^{c}$ designate the full and the weighted cluster-based sampling
	log-likelihood functions. Assume further that the probabilistic choice $p\left(a_{it}|x_{it};\vartheta\right)$,
	the structure of the attribute space, and the attribute distribution
	$p\left(x_{1},...,x_{T}\right)$ imply a unique maximum $\vartheta^{*}$ when
	the sample size $M\rightarrow\infty$. Then the maximum likelihood
	estimators of $L^{a}$ and $L^{c}$ converge to the same $\vartheta^{*}$
	almost surely when the sample size $M\rightarrow\infty$ and the number
	of clusters $K\rightarrow\infty$.
\end{proposition}

\proof{Proof}
When the cluster-based sampling method is applied, the clusters sizes
form a distribution $p\left(m_{k(i)}|a_{i1},...,a_{iT}\right)$. The
conditional distribution of household characteristics given clustering
is $p\left(x_{i1},...,x_{iT}|m_{k(i)},a_{i1},...,a_{iT}\right)$.
Notice that the random variables in the full log-likelihood $L^{a}$
and $L^{c}$ have different distributions. In the full log-likelihood
$L^{a}$, $a_{it}$ and $x_{it}$ have simply the joint distribution
\begin{equation}
	p_{i}\left(a_{i1},...,a_{iT},x_{it},...,x_{iT}\right)=\left[\prod_{t=1}^{T}p\left(a_{it}|x_{it};\vartheta\right)\right]p\left(x_{it},...,x_{iT}\right),
\end{equation}
whereas in $L^{c}$, the random variables $a_{it}$, $m_{k(i)}$ and
$x_{it}$ of household $i$ have a different joint distribution because
of cluster-based sampling. We have 
\begin{eqnarray}
	&  & p\left(a_{i1},...,a_{iT},x_{it},...,x_{i(k)T},m_{k}\right)\nonumber \\
	&  & =\frac{1}{m_{k(i)}}p\left(x_{i1},...,x_{iT}|m_{k(i)}\right)p\left(m_{k(i)}|a_{i1},...,a_{iT}\right)p\left(a_{i1},...,a_{iT}\right),
\end{eqnarray}
where 
\[
p\left(a_{i1},...,a_{iT}\right)=\int p\left(a_{i1},...,a_{iT},x_{it},...,x_{iT}\right)dx_{i1}\cdots dx_{iT},
\]
and $1/m_{k(i)}$ reflects the probability that $i$ is randomly chosen
from cluster $M_{k}$ ($m_{k(i)}=1$ when $i$ is in a cluster of
size one, which implies it will always be chosen if it is not in a
cluster with more than one household).

Let us consider the log-likelihood $L^{a}$ divided by $M$: 
\[
D_{M}^{a}\left(a,x,\vartheta\right)=\frac{1}{M}\sum_{i=1}^{M}\left\{ \sum_{t=1}^{T}\log p\left(a_{it}|x_{it};\vartheta\right)\right\} .
\]
When $M\rightarrow\infty$, by the strong law of large numbers\footnote{Here, $\int_{x_{t}}dx_{1}\cdots dx_{T}$ is an abbreviation of $\int_{x_{1}}\cdots\int_{x_{T}}dx_{1}\cdots dx_{T}$
	and $\underset{a_{t}=\left\{ 0,1\right\} }{\sum}$ is for $\underset{a_{1}=\left\{ 0,1\right\} }{\sum}\underset{a_{2}=\left\{ 0,1\right\} }{\sum}\cdots\underset{a_{T}=\left\{ 0,1\right\} }{\sum}$.}, 
\begin{eqnarray*}
	&  & D_{M}^{a}\left(a,x,\vartheta\right)\overset{a.s.}{\rightarrow}D^{a}\left(\vartheta\right)\\
	&  & =\int_{x_{t}}\sum_{a_{t}=\left\{ 0,1\right\} }\left\{ \sum_{t=1}^{T}\log p\left(a_{t}|x_{t};\vartheta\right)\right\} \prod_{t=1}^{T}p\left(a_{t}|x_{t};\vartheta\right)p\left(x_{1},...,x_{T}\right)dx_{1}\cdots dx_{T}.
\end{eqnarray*}

Now, let us consider the log-likelihood $L^{c}$ divided by $K$:
\[
D_{K}^{c}\left(a,x,\vartheta\right)=\frac{L^{c}}{K}=\frac{1}{K}\sum_{k=1}^{K}\left\{ m_{k(i)}\sum_{t=1}^{T}\log p\left(a_{it}|x_{it};\vartheta\right)\right\} .
\]
When $M\rightarrow\infty$, because the cluster sizes $m_{k}$ do not
increase with $M$, the number of clusters $K\rightarrow\infty$.
Hence, by the strong law of large numbers, 
\begin{eqnarray*}
	&  & D_{K}^{c}\left(a,x,\vartheta\right)\overset{a.s.}{\rightarrow}D^{c}\left(\vartheta\right)=\int_{x_{t}}\int_{m_{k}}\sum_{a_{t}=\left\{ 0,1\right\} }\left\{ m_{k}\sum_{t=1}^{T}\log p\left(a_{t}|x_{t};\vartheta\right)\right\} \cdot\\
	&  & \frac{1}{m_{k}}p\left(x_{1},...,x_{T}|m_{k}\right)p\left(m_{k}|a_{1},...,a_{T}\right)p\left(a_{1},...,a_{T}\right)dm_{k}dx_{1}\cdots dx_{T}\\
	&  & =\int_{x_{t}}\int_{m_{k}}\sum_{a_{t}=\left\{ 0,1\right\} }\left\{ \sum_{t=1}^{T}\log p\left(a_{t}|x_{t};\vartheta\right)\right\} p\left(x_{1},...,x_{T},m_{k},a_{1},...,a_{T}\right)dm_{k}dx_{1}\cdots dx_{T}\\
	&  & =\int_{x_{t}}\sum_{a_{t}=\left\{ 0,1\right\} }\left\{ \sum_{t=1}^{T}\log p\left(a_{t}|x_{t};\vartheta\right)\right\} p\left(x_{1},...,x_{T},a_{1},...,a_{T}\right)dx_{1}\cdots dx_{T}\\
	&  & =\int_{x_{t}}\sum_{a_{t}=\left\{ 0,1\right\} }\left\{ \sum_{t=1}^{T}\log p\left(a_{t}|x_{t};\vartheta\right)\right\} \prod_{t=1}^{T}p\left(a_{t}|x_{t};\vartheta\right)p\left(x_{1},...,x_{T}\right)dx_{1}\cdots dx_{T}.
\end{eqnarray*}
Therefore, when $M\rightarrow\infty$, $D^{c}\left(\vartheta\right)=D^{a}\left(\vartheta\right)$,
so $D_{K}^{c}\left(a,x,\vartheta\right)$ and $D_{M}^{a}\left(a,x,\vartheta\right)$
are asymptotically equivalent. Observe that 
\begin{eqnarray*}
	\int_{x_{t}}\sum_{a_{t}=\left\{ 0,1\right\} }\prod_{t=1}^{T}p\left(a_{t}|x_{t};\vartheta\right)p\left(x_{1},...,x_{T}\right)dx_{1}\cdots dx_{T} & = & 1\text{ and}\\
	\sum_{t=1}^{T}\log p\left(a_{t}|x_{t};\vartheta\right) & = & \log\prod_{t=1}^{T}p\left(a_{t}|x_{t};\vartheta\right)
\end{eqnarray*}
Hence, by Lemma 1 in \cite{Manskietal1977} and \cite{Rao1973}

$D^{c}\left(\vartheta\right)$ and $D^{a}\left(\vartheta\right)$
are maximized at the same and unique $\vartheta=\vartheta^{*}.$ And
by Lemma 2 of \cite{Manskietal1977}, both $\hat{\vartheta}_{M}$,
which maximizes $D_{M}^{a}\left(a,x,\vartheta\right)$, and $\hat{\vartheta}_{K}$,
which maximizes $D_{K}^{c}\left(a,x,\vartheta\right)$, converge to
the same $\vartheta^{*}$ almost surely.{\hfill \ensuremath{\Box}}\endproof
~\\
\noindent\textbf{Application to the dynamic structural model estimation}

We aggregate non-adopting households into clusters with similar characteristics within each census block group (www2.census.gov) using the k-means clustering technique, with home market value, house size, tree cover, and received irradiance as the clustering features.

Figure \ref{fig:aggbkgexample} shows the households before and after the clustering. Figures \ref{fig:aggbkgexample}(a) and (b) illustrate the Austin Energy service area, where the red dots represent adopters during the analysis period (2004-2013) and each black dot is a potential adopter in Figure \ref{fig:aggbkgexample}(a) and a cluster of potential adopters in Figure \ref{fig:aggbkgexample}(b). 

Figure \ref{fig:aggbkgexample}(c) shows magnified plots of two clusters within a particular census block group, where the black dots represent potential adopters and the two green dots mark the centroids of these two clusters. The clusters remain constant over the time frame of our observations. 

The clustering yields 1,185 clusters of size greater than one, with an average cluster size of 26 and a median of 7. Given this large number of clusters, the asymptotic result above applies, and the clustering does not introduce appreciable bias into our estimation.
 
\begin{figure}[H]
{	\centering
	\begin{tabular}{C{5cm} C{5cm} C{7cm}}
		{\footnotesize(a) Non-adopters disaggregated}&{\footnotesize(b) Non-adopters in clusters} & {\footnotesize(c) Non-adopters disaggregated and cluster centroids}\\
		\includegraphics[width=1\linewidth]{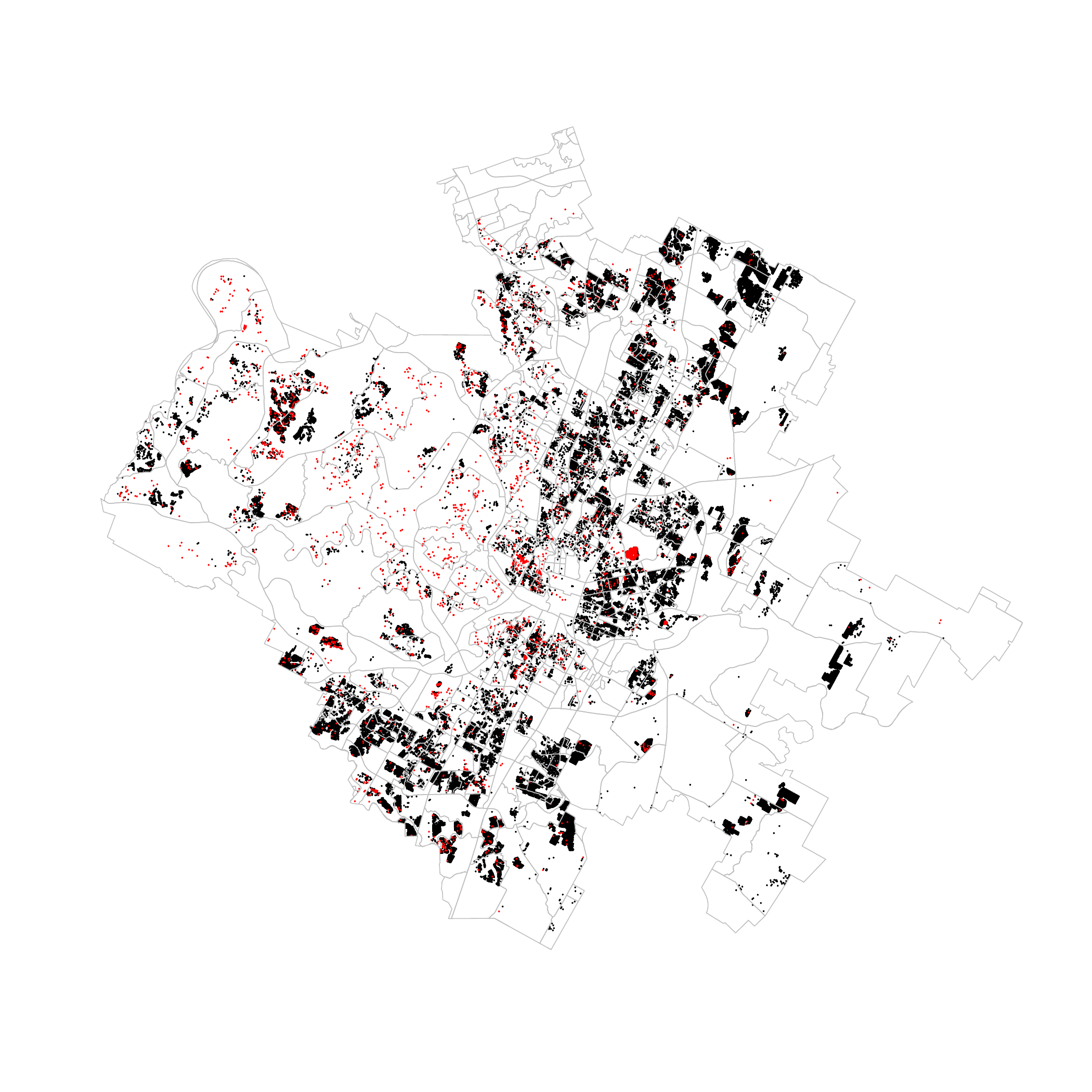}	&  
		\includegraphics[width=1\linewidth]{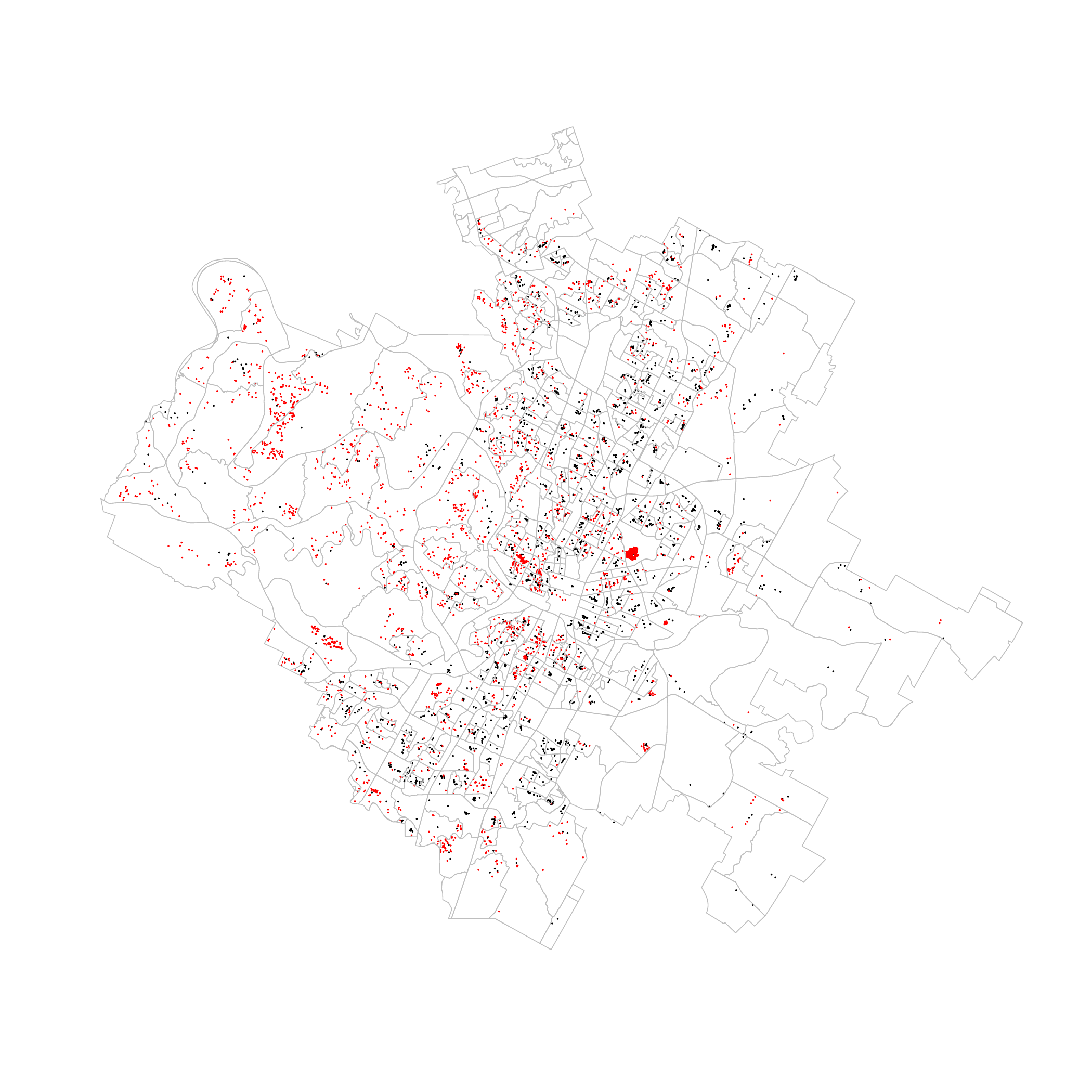} &  
		\includegraphics[width=1\linewidth]{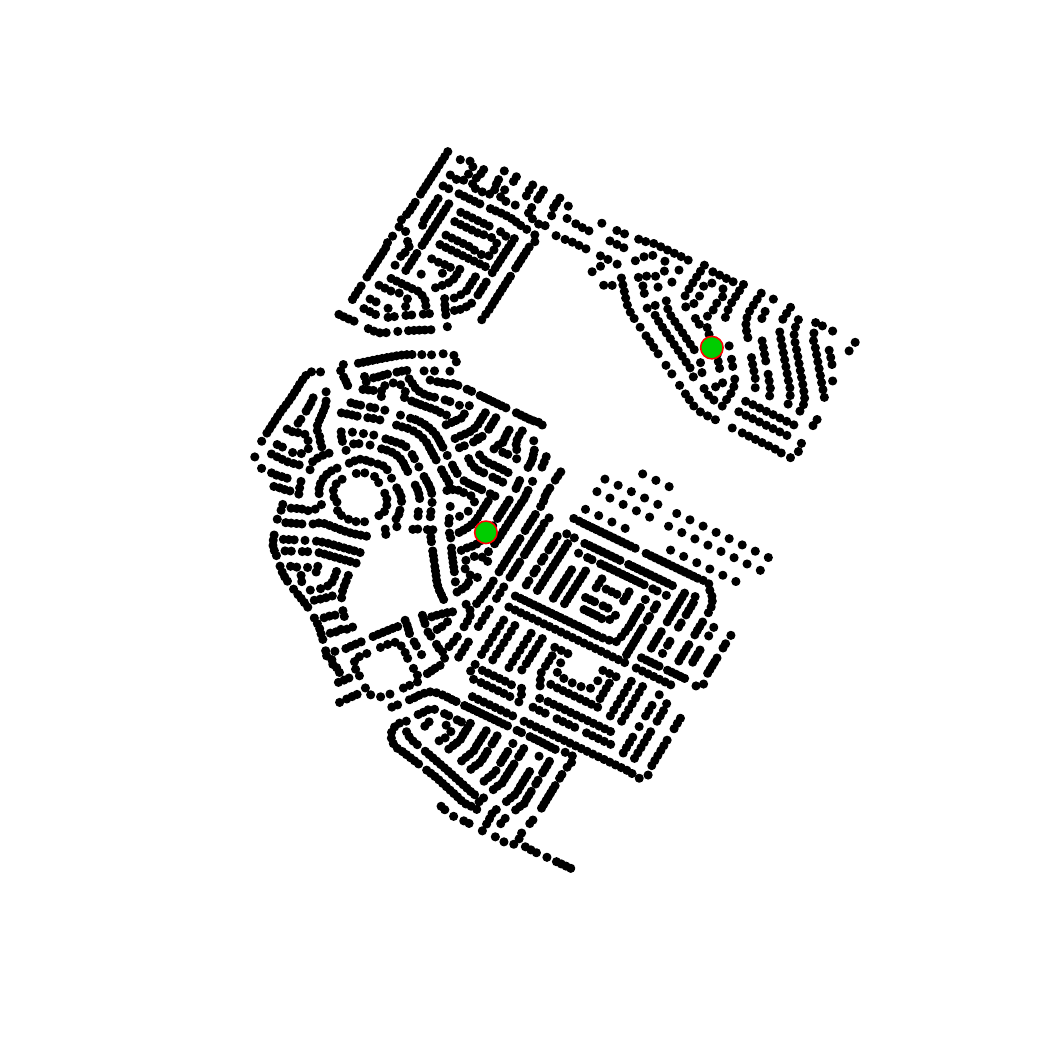}\\
	\end{tabular}
	\caption{PV current adopters and potential future adopters in 2013, Austin, TX.}
	\label{fig:aggbkgexample}
}

{\footnotesize 
	\setlength{\baselineskip}{9pt}
	\noindent \emph{Note.} Red dots are adopters. Black dots are non-adopters. Green dots mark the cluster centroids.
	\par}

\end{figure}



\clearpage
\section{Estimation and Simulation Algorithms}
We reproduce in this appendix the equations presented in the paper so that the algorithm's equation references are self-contained. See the paper for the detailed explanations.
\vspace{0.5cm}

\begin{minipage}{\textwidth}
\subsubsection*{Dynamic discrete choice model notation}
	\begin{center}
		\vskip 10pt

\newcommand{\MyColanWidth}{1.3}
\newcommand{\MyColbnWidth}{4.5}

{\footnotesize
\begin{tabular} {p{\MyColanWidth in}  p{\MyColbnWidth in}}
\toprule
\multicolumn{2}{l}{\textbf{Indices and Sets}}\bigstrut\\
\midrule
$i \in M$ & Set of potential households adopters.\\

$e \in E$ & Set of economic segments.\\

$g\in G$ & Set of geographic segments.\\

$t \in \{1,\ldots, T\}$ &  Periods (quarters). \\

$d_l \in D$ & Set of discrete distances to define neighborhoods by distance radius rings. $D = \{d_1<\ldots< d_{|D|}\}$ and $d_0 = 0$.\\

\midrule
\multicolumn{2}{l}{\textbf{Input Data}}\bigstrut\\
\midrule



$\pvSize_{i}$ & PV system size of household $i$.\\ 

$\pvGen_{it}$ & Expected electricity generated per watt by household $i$'s PV system in period $t$.\\ 

$\vos_t$ &  Bill credits (i.e., solar energy rate) for the electricity produced by PV in period $t$.\\ 

$\pvPrice_t$ & PV installation cost in period $t$.\\

$\rebate_{t}$ & PV rebate available in period $t$.\\

$\text{ITC}_{t}$ & Federal Investment tax credit in period $t$.\\


$e_i \in E$ & Economic segment of household $i$.\\

$g_i \in G$ & Geographic segment of household $i$.\\


$t_i$ & Period of household $i$'s adoption decision. The PV system is installed in the following period. And $t_i = \infty$ implies $i$ does not adopt.\\

$\tau_i$ & Period of property $i$'s built. If the property is built before the beginning of the estimation window, then $\tau_i = 1$.\\

${M}_t$ & Set of households that have adopted before period $t$. ${M}_t = \{i \in M: t_i < t\}$.\\

${M}^c_t$ & Set of households that have not adopted in period $t$. ${M}^c_t = {M} \setminus {M}_t$.\\

$\df$& Household discount factor.\\

$L$ &  Lifespan of a PV system (80 quarters).\\

\midrule
\multicolumn{2}{l}{\textbf{Decision variable}}\bigstrut\\
\midrule
$a_{it} \in \{0,1\}$& Household $i$'s adoption decision in period $t$; i.e., it equals 1 if household $i$ adopts and 0 otherwise.\\

\midrule
\multicolumn{2}{l}{\textbf{State variables}}\bigstrut\\
\midrule


$\pvCost_{t}$ & Net cost per watt of PV in period $t$. $\pvCost_{t} = (\pvPrice_t - \rebate_t)(1- \textrm{ITC}_t)$.\\

$\bm{\sInstBase}_{it} = \{\sInstBase_{i1t}, \ldots, \sInstBase_{i|D|t}\}$ & Previous installations within radial distance $d\in (d_{l-1}, d_l]$ of household $i$, $l = 1,\ldots,|D|$.\\




$\latentUtility_{iat}$ & Additive random shock in household $i$'s utility in period $t$.\\




\midrule
\multicolumn{2}{l}{\textbf{Structural parameters}}\bigstrut\\
\midrule


%
$\pBaseUtility$ & Constant intercept of the utility function.\\

$\pMMoney_{e}$ & Weight of solar economic benefit for economic segment $e \in E$.\\

$\bm{\pMPE}_{g} = \{\pMPE_{g1}, \ldots, \pMPE_{gD}\}$ &  Neighborhood effects of previous installations within radial distance $d\in (d_{l-1}, d_l]$, $l = 1,\ldots,|D|$, for geographic segment $g \in G$.\\

$\bm\theta = \{\pBaseUtility, \bm\pMMoney, \bm{\pMPE}\}$ & Vector of structural parameters.\\

$\randomEffects_{i}$ & Household unobserved heterogeneity.\\

$\sigma_\randomEffects$ & Standard deviation of household heterogeneity $\randomEffects_i$.\\

\bottomrule
\end{tabular}
}
	\end{center}
\end{minipage}	

\clearpage
\subsubsection*{Per-Period Utility for PV Adoption}

\begin{description}
\item[]Utility for adopting:

\begin{align}
	&& u_{i1t}(\pvCost_t,\bm\sInstBase_{it},\randomEffects_{i},\latentUtility_{i1t};\pBaseUtility, \pMMoney_{e_i},\bm\pMPE_{g_i})&=   \pBaseUtility + \pMMoney_{e_i} \npv_{it}(\pvCost_t) +\bm\pMPE_{g_i}'\bm\sInstBase_{it} + \randomEffects_{i}+\latentUtility_{i1t},\label{eq:ui1t}\\
\text{where}& && \nonumber\\
&	&\npv_{it}(\pvCost_t)&=   
	%
	%
	\left(\sum\limits_{\tau = t+1}^{t+1+L} \df^{\tau-t}\vos_\tau\pvGen_{i\tau} - \pvCost_t\right)\pvSize_i. \label{eq:vpnit}
\end{align}

\item[]Utility for not adopting:
\begin{align}
	u_{i0t}(\latentUtility_{i0t}) &= \latentUtility_{i0t}.\label{eq:ui0t}
\end{align}
\end{description}

\subsubsection*{Dynamic Optimal Decision of Adoption}

\begin{align}
	\fValue_{it}(\pvCost_t,\bm\sInstBase_{it},\randomEffects_{i} ,\bm{\latentUtility}_{it}; \bm\theta) & = \max\limits_{a \in \{0,1\}}\fValueCE_{iat} (\pvCost_t,\bm\sInstBase_{it}, \randomEffects_{i} ,\bm{\latentUtility}_{it}; \bm\theta),\label{eq:vit}
\end{align}

where $\fValueCE_{iat}$ is the choice specific value function for making the decision $a \in \{0,1\}$:
\begin{align}
	\fValueCE_{i0t} (\pvCost_t,\bm\sInstBase_{it}, \randomEffects_{i} ,\bm{\latentUtility}_{it}; \bm\theta)& = 
	u_{i0t}(\latentUtility_{i0t})  + \df E_{\pvCost, \bm\sInstBase,\bm{\latentUtility}}\left[\fValue_{it+1}(\pvCost, \bm\sInstBase,\randomEffects_i,\bm{\latentUtility}; \bm\theta)|\pvCost_t,\bm\sInstBase_{it},\bm{\latentUtility}_{it}\right]\nonumber\\
	&= \latentUtility_{i0t}  + \df E_{\pvCost,\bm\sInstBase ,\bm{\latentUtility}_i}\left[\fValue_{it+1}(\pvCost, \bm\sInstBase,\randomEffects ,\bm{\latentUtility}; \bm\theta)|\pvCost_t,\bm\sInstBase_{it},\bm{\latentUtility}_{it}\right].
	\label{eq:vi0t}\\
	\fValueCE_{i1t} (\pvCost_t,\bm\sInstBase_{it}, \randomEffects_{i} ,\bm{\latentUtility}_{it}; \bm\theta)& = u_{i1t}(\pvCost_t,\bm\sInstBase_{it},\randomEffects_{i},\latentUtility_{i1t} ; \bm\theta) \nonumber\\
	&= \pBaseUtility + \pMMoney_{e_i} \npv_{it}(\pvCost_t) +\bm\pMPE_{g_i}'\bm\sInstBase_{it} + \randomEffects_{i}+\latentUtility_{i1t}.\label{eq:vi1t}
\end{align}

Household $i$ will adopt at $t$ if $\fValueCE_{i1t} (\pvCost_t,\bm\sInstBase_{it}, \randomEffects_{i} ,\bm{\latentUtility}_{it}; \bm\theta) \geq \fValueCE_{i0t} (\pvCost_t,\bm\sInstBase_{it}, \randomEffects_{i} ,\bm{\latentUtility}_{it}; \bm\theta)$.

\subsubsection*{State Variables, Their Transition Equations and Other Model Specifications}

\begin{align}
	\pvCost_{t+1} &= \kappa \pvCost_{t} + \epsilon^{\pvCost}_t,  & \epsilon^{\pvCost}_t &\sim N(0,\sigma^2_{\pvCost}), \label{eq:ct}\\
	\sInstBase_{i,t+1} &= \sInstBase_{it} + \Delta_{it},  & \Delta_{it} &\sim \text{Exponential}\left(\tfrac{1}{(\lambda_i-1)\sInstBase_{it}}\right), \quad \lambda_i>1,\label{eq:ht}\\
	\randomEffects_{i} &\sim N(0,\sigma^2_{\randomEffects}),  & \sigma^2_{\randomEffects} &\sim InverseGamma(\pi_1,\pi_2). \label{eq:xii}
\end{align}

\subsection*{Bayesian Estimation}

\vskip 10pt
\begin{tabular} {p{1.3 in}  p{4.5 in}}
$r = 1,\ldots R$ & Iterations of the Metropolis–within–Gibbs method.\\

$\left\{c_t^{j(r)}\right\}_{j=1}^J$  & Random grid of size $J$ of the net cost state variable $c_t$ in iteration $r$ using \eqref{eq:ct}.\\

$f(c_{t+1}| c_t)$ & Gaussian Kernel resulting from \eqref{eq:ct}.\\

$K_\eta\left(\bm\theta^{(n)} - \bm\theta^{(r)} \right)$ & Multivariate Gaussian kernel with bandwidth $\eta > 0$.\\ 

\end{tabular}

\begin{align}
	\hat E\Big[\hat V^{(r)}_{it+1}\left(c, \bm{h}_{it+1},\xi ; \bm\theta\right)&|c^{j(r)}_{t},\bm{h}_{it},\xi^{(r)}_{i}; \bm\theta^{(r)}\Big] \nonumber\\
	 &= \frac{\sum_{n = r-N}^{r-1} K_\eta\left(\bm\theta^{(n)} - \bm\theta^{(r)} \right) \sum_{j=1}^J \hat V^{(n)}_{it+1}\left({c}^{j(n)}_{t+1},\bm{h}_{it+1},\xi^{(n)}_{i} ; \bm\theta^{(n)}\right)f\left({c}^{j(n)}_{t+1}|{c}^{j(r)}_{t}\right)}{\sum_{n = r-N}^{r-1} K_\eta\left(\bm\theta^{(n)} - \bm\theta^{(r)} \right)\sum_{j=1}^J f\left({c}^{j(n)}_{t+1}|{c}^{j(r)}_{t}\right)}, \label{eq:emax_approx} 
\end{align}	
where
\begin{align}
	\hat V^{(n)}_{it}({c}^{j(n)}_{t},\bm{h}_{it},\xi^{(n)}_{i} ; \bm\theta^{(n)}) &
	= \log\left(\exp\hat{\mathcal{V}}^{(n)}_{i0t}\left({c}^{j(n)}_{t},\bm{h}_{it},\xi^{(n)}_{i}; \bm\theta^{(n)}\right)+ \exp\hat{\mathcal{V}}^{(n)}_{i1t}\left({c}^{j(n)}_{t},\bm{h}_{it},\xi^{(n)}_{i}; \bm\theta^{(n)}\right)\right), \label{eq:pseudo_vf}
\end{align}
and
\begin{align}
	\hat{\mathcal{V}}^{(n)}_{i0t}\left({c}^{j(n)}_{t},\bm{h}_{it},\xi^{(n)}_{i} ; \bm\theta^{(n)}\right)& 
	=  \df \hat E\left[\hat V^{(n)}_{it+1}\left({c}, \bm{h}_{it+1},\xi; \bm\theta\right)|{c}^{j(n)}_{t},\bm{h}_{it},\xi^{(n)}_{i}; \bm\theta^{(n)}\right]\\
	\hat{\mathcal{V}}^{(n)}_{i1t}\left({c}^{j(n)}_{t},\bm{h}_{it},\xi^{(n)}_{i} ; \bm\theta^{(n)}\right)& 
	=\rho^{(n)} + \alpha^{(n)}_{e_i} \textrm{NPV}_{it}\left({c}^{j(n)}_t\right) +\bm\gamma^{(n)'}_{g_i}\bm{h}_{it} + \xi^{(n)}_{i}. 
\end{align}

Then, at iteration $r$, the probability of adopting at period $t$ is:
\begin{align}
	\hat p^{(r)}_{i1t}(c_t,\bm{h}_{it}, \xi^{(r)}_{i};\bm\theta^{(r)}) &= \frac{\exp\hat{\mathcal{V}}^{(r)}_{i1t} \left({c}_{t},\bm{h}_{it},\xi^{(r)}_{i}; \bm\theta^{(r)}\right)}{\exp\hat{\mathcal{V}}^{(r)}_{i1t}\left({c}_{t},\bm{h}_{it},\xi^{(r)}_{i} ; \bm\theta^{(r)}\right)  +\exp\hat{\mathcal{V}}^{(r)}_{i0t}\left({c}_{t},\bm{h}_{it},\xi^{(r)}_{i} ; \bm\theta^{(r)}\right)}.
\end{align}

Recall that $\tau_i$ is the period in which household $i$'s property is built ($\tau_i=1$ if built before or at the first period) and $t_i$ is the adoption period ($t_i=T$ if $i$ does not adopt during the horizon). Then, at iteration $r$, household $i$'s log-likelihood is:
\begin{align}
l_i^{(r)}&\left(\bm{c},\bm{h}_{i}, \xi^{(r)}_{i};\bm{\theta}^{(r)}\right) \nonumber\\
&= \sum_{t=\tau_i}^{t_i}\left\{1_{\{a_{it}=1\}}\log\left({\hat{p}}_{i1t}^{(r)}\left(c_t,\bm{h}_{it}, \xi^{(r)}_{i};\bm{\theta}^{(r)}\right)\right)+1_{\{a_{it}=0\}}\log\left(1-{\hat{p}}_{i1t}^{(r)}\left(c_t,\bm{h}_{it}, \xi^{(r)}_{i};\bm{\theta}^{(r)}\right)\right)\right\},\nonumber
\end{align}

And the full likelihood is:
\begin{align}
\mathcal{L}^{(r)}\left(\bm{c},\bm{h}, \bm\xi^{(r)};\bm{\theta}^{(r)}\right) &= \sum_{i \in M} l^{(r)}_i\left(\bm{c},\bm{h}_{i}, \xi^{(r)}_{i};\bm{\theta}^{(r)}\right).\label{eq:loglik}
\end{align}

\begin{center}
\textbf{Markov Chain Monte Carlo (MCMC) Estimation Algorithm}
{	
{\scriptsize
	\begin{algorithm}[H]
		\SetAlgoLined
		\DontPrintSemicolon
		\SetKwData{llk}{llk}
		
		\SetKwData{Up}{up}
		
		\SetKwFunction{EmaxApprox}{$\hat E$}
		\SetKwFunction{VApprox}{$\hat V$}
		\SetKwFunction{LogLik}{$\mathcal{L}$}
		\SetKwFunction{CSimulate}{SimulateCostFunction}
		\SetKwFunction{InverseGamma}{$InverseGamma$}
		\SetKwFunction{Normal}{$N$}
		\SetKwFunction{U}{U}
		\SetKwFunction{mean}{mean}
		\SetKwInOut{Input}{Input}\SetKwInOut{Output}{Output}\SetKwInOut{Variables}{Variables}\SetKwInOut{Initialization}{Initialization}
		\Input{$Data$: household characteristics, time of adoption, costs, value of solar. $\beta$ discount factor. $R$ number of MCMC iterations. $\pi_1, \pi_2, \sigma_{\pBaseUtility}, \sigma_{\pMMoney}$ and $\sigma_{\pMPE}$ hyperparameters to control acceptance rate and convergence of MCMC.}
		\Output{$ \left\{\bm\theta^{(r)}\right\}_{r=1}^R = \left\{({\pBaseUtility}^{(r)}, {\pMMoney}^{(r)}, {\pMPE}^{(r)}, {\sigma}_{\xi}^{(r)})\right\}$ MCMC and posterior mean of structural parameters. For the sake of exposition, assume only one economic and one geographic segment.} 
		\Variables{$\pBaseUtility^{(r)}, \pMMoney^{(r)},\pMPE^{(r)}, \xi^{(r)}_i$ accepted values at iteration $r$. $\pBaseUtility', \pMMoney',\pMPE', \xi'_i$ proposal. $\pBaseUtility^{0}, \pMMoney^{0},\pMPE^{0}, \xi^{0}_i = 0$}
		\For{$r \in 1$ \KwTo $R$}{
			\tcp{ ********** HOUSEHOLDS RANDOM EFFECTS **********}
			Simulate $c$ and obtain $f$ using Eq \eqref{eq:ct}\;
			$\hat E_{it} = \EmaxApprox(\hat V_{it}, f)$, for all $i \in M, t = 1,\ldots, T$, using  Eq. \eqref{eq:emax_approx}\;
			$l = \LogLik(\pBaseUtility^{(r-1)}, \pMMoney^{(r-1)},\pMPE^{(r-1)}, \xi^{(r-1)}_i,\hat E_i, Data, \beta, c, f)$, using Eq. \eqref{eq:loglik}\;
			$\sigma_{\xi} \sim \InverseGamma(\pi_1, \pi_2)$, using Eq. \eqref{eq:xii}\;
			$\xi'_i \sim \Normal(0, \sigma_{\xi})$ using Eq. \eqref{eq:xii}\;
			$l'_i = \LogLik(\pBaseUtility^{(r-1)}, \pMMoney^{(r-1)},\pMPE^{(r-1)}, \xi'_i, \hat E_i, Data, \beta, c, f)$, using Eq. \eqref{eq:loglik}\;
			$u \sim \U(0,1)$ \; 
			\leIf {$\log(u) \leq \min(l'_i - l_i,0)$}{$\xi^{(r)}_i = \xi'_i$}{$\xi^{(r)}_i = \xi^{(r-1)}_i$}
			$\hat V_i = \VApprox(\pBaseUtility^{(r-1)}, \pMMoney^{(r-1)},\pMPE^{(r-1)}, \xi^{(r)}_i,\hat E_i, Data, \beta, c,f)$, using Eq. \eqref{eq:pseudo_vf}\;
			\tcp{ ********** BASE UTILITY **********}
			Simulate $c$ and obtain $f$ using Eq \eqref{eq:ct}\;
			\For{$i \in$ Households}{
				$\hat E_i = \EmaxApprox(\hat V_i, f)$, using  Eq. \eqref{eq:emax_approx}\;
				$l_i = \LogLik(\pBaseUtility^{(r-1)}, \pMMoney^{(r-1)},\pMPE^{(r-1)}, \xi^{(r)}_i,\hat E_i, Data, \beta, c, f)$ , using Eq. \eqref{eq:loglik}\;
			}
			$l = \sum l_i$\;
			$\pBaseUtility' \sim \Normal(\pBaseUtility^{(r-1)}, \sigma_{\pBaseUtility})$\;
			\lFor{$i \in$ Households}{
				$l'_i = \LogLik(\pBaseUtility', \pMMoney^{(r-1)},\pMPE^{(r-1)}, \xi^{(r)}_i,\hat E_i, Data, \beta, c,f)$, using Eq. \eqref{eq:loglik}
			}	
			$l' = \sum l'_i$\;
			$u \sim \U(0,1)$ \; 
			\leIf{$\log(u) \leq \min(l' - l,0)$}{$\pBaseUtility^{(r)} = \pBaseUtility'$}{$\pBaseUtility^{(r)} = \pBaseUtility^{(r-1)}$}
			\lFor{$i \in$ Households}{			
				$\hat V_i = \VApprox(\pBaseUtility^{(r)}, \pMMoney^{(r-1)},\pMPE^{(r-1)}, \xi^{(r)}_i,\hat E_i, Data, \beta, c, f)$, using Eq. \eqref{eq:pseudo_vf}\;
			}
			\tcp{ ********** ECONOMIC FACTOR **********}
			Simulate $c$ and obtain $f$ using Eq \eqref{eq:ct}\;
			\For{$i \in$ Households}{
				$\hat E_i = \EmaxApprox(\hat V_i, f)$, using  Eq. \eqref{eq:emax_approx}\;
				$l_i = \LogLik(\pBaseUtility^{(r)}, \pMMoney^{(r-1)},\pMPE^{(r-1)}, \xi^{(r)}_i,\hat E_i, Data, \beta, c, f)$, using Eq. \eqref{eq:loglik}\;
			}
			$l = \sum l_i$\;
			$\pMMoney' \sim \Normal(\pMMoney^{(r-1)}, \sigma_{\pMMoney})$\;
			\lFor{$i \in$ Households}{
				$l'_i = \LogLik(\pBaseUtility^{(r)}, \pMMoney',\pMPE^{(r-1)}, \xi^{(r)}_i,\hat E_i, Data, \beta, c, f)$, using Eq. \eqref{eq:loglik}
			}	
			$l' = \sum l'_i$\;
			$u \sim \U(0,1)$ \; 
			\leIf {$\log(u) \leq \min(l' - l,0)$}{$\pMMoney^{(r)} = \pMMoney'$}{$\pMMoney^{(r)} = \pMMoney^{(r-1)}$}	
			\lFor{$i \in$ Households}{			
				$\hat V_i = \VApprox(\pBaseUtility^{(r)}, \pMMoney^{(r)},\pMPE^{(r-1)}, \xi^{(r)}_i,\hat E_i, Data, \beta, c, f)$, using Eq. \eqref{eq:pseudo_vf}\;
			}
			\tcp{ ********** NEIGHBORHOOD EFFECTS **********}
			Simulate $c$ and obtain $f$ using Eq \eqref{eq:ct}\;
			\For{$i \in$ Households}{
				$\hat E_i = \EmaxApprox(\hat V_i, f)$, using  Eq. \eqref{eq:emax_approx}\;
				$l_i = \LogLik(\pBaseUtility^{(r)}, \pMMoney^{(r)},\pMPE^{(r-1)}, \xi^{(r)}_i,\hat E_i, Data, \beta, c, f)$, using Eq. \eqref{eq:loglik}
			}
			$l = \sum l_i$\;
			$\pMPE' \sim \Normal(\pMPE^{(r-1)}, \sigma_{\pMPE})$\;
			\lFor{$i \in$ Households}{
				$l'_i = \LogLik(\pBaseUtility^{(r)}, \pMMoney^{(r)},\pMPE', \xi^{(r)}_i,\hat E_i, Data, \beta, c, f)$, using Eq. \eqref{eq:loglik}
			}
			$l' = \sum l'_i$\;
			$u \sim \U(0,1)$ \; 
			\leIf {$\log(u) \leq \min(l' - l,0)$}{$\pMPE^{(r)} = \pMPE'$}{$\pMPE^{(r)} = \pMPE^{(r-1)}$}
			\lFor{$i \in$ Households}{$\hat V_i = \VApprox(\pBaseUtility^{(r)}, \pMMoney^{(r)},\pMPE^{(r)},\xi^{(r)}_i,\hat E_i, Data, \beta, c, f)$, using Eq. \eqref{eq:pseudo_vf}\;}
		}\;
		Posterior mean of structural parameters\; 
		$\bar{\pBaseUtility} = \mean(\pBaseUtility^{(r)})$; 
		$\bar{\pMMoney} = \mean(\pMMoney^{(r)})$; 
		$\bar{\pMPE} = \mean(\pMPE^{(r)})$; 
		$\bar{\sigma}_{\xi} = \mean(\sigma^{(r)}_{\xi})$; 
	\end{algorithm}
}
}
\end{center}


\subsection*{Counterfactual Policy Analysis: Simulation Algorithm}

Simulations based on Bayesian estimation usually propagate every posterior draw through the simulation. However, we use the posterior means $\bar{\bm\theta} = \{\bar{\pBaseUtility}, \bar{\bm\pMMoney}, \bar{\bm{\pMPE}},  \bar{\bm{\xi}} \}$ as the estimates of the structural parameters to reduce the computational burden because our policy simulation algorithm involves solving for an equilibrium that requires a large number of iterations, as we explain next. 

To simulate a specific policy scenario's outcome, we compute the value function and the probabilities of adoption using the fitted value iteration algorithm (see, for example, \cite{Stachurski2009}). Since our best Model (4)H includes only one radial ring, the state variable representing the number of neighboring adopters is a scalar $h_{it}$. Note that, in the dynamic decision problem, $h_{it}$ is endogenous because it is the aggregation of the adoption decisions made by the neighbors of household $i$. Since $h_{it}$ is unobserved in the simulated counterfactual scenarios, we need to compute the dynamic equilibrium to obtain the state transition function for $h_{it}$ based on the household’s rational expectations. 

Our simulation algorithm uses an outer loop to compute this state transition equilibrium. We assume that, at each period $t$, household $i$ uses the transition function $h_{i,t+1} = \lambda_i h_{it}$, with $\lambda_i \geq 1$, the conditional mean of the transition in equation \eqref{eq:ht}, to predict the cumulative number of its neighboring adopters in the next period until the total number of adopters equals the number of potential adopters $\hat{h}_i$ in its neighborhood; after reaching this capacity, $h_{it}$ remains constant. Observe that we permit the \textit{growth} rate $\lambda_i$ to vary by household. After simulating the new adopters in each period $t$, we re-calculate $h_{it}$ and re-estimate $\lambda_i$ as the updated state transition function. We iterate this process until the number of adopters per period and $\lambda_{i}$ converge; these iterations constitute the outer loop of our simulation algorithm. 

Because the PV system has a finite lifespan of 20 years (80 quarters), inside each outer loop we compute the finite-horizon value function using backward induction. For every household $i$ and $h=1,\ldots,\hat{h}_i$, we first use linear interpolation and the fitted value iteration algorithm to calculate the value function in the last period $t=T$. Then, for any period $t < T$, we use backward induction to calculate the value function. Specifically, the recursive representation of the value function is
\begin{align}
	\fValue_{it}(\pvCost_t,\sInstBase_{it},\bm{\latentUtility}_{it}; \bar{\bm\theta})
	= \max \left\{\fValueCE_{i0t} (\pvCost_t,\sInstBase_{it},\bm{\latentUtility}_{it}; \bar{\bm\theta}),\fValueCE_{i1t} (\pvCost_t,\sInstBase_{it},\bm{\latentUtility}_{it}; \bar{\bm\theta})\right\}.
\end{align}
In the simulation, the policy's net-cost schedule $c_t$ is a known input, and household $i$ predicts its next-period neighbor count with the outer-loop transition  $\min\{\lambda^{(k)}_i h^{(k)}_{it}, \hat{h}_i\}$, so the continuation value requires no expectation over states. At iteration $k$ of the ENAE algorithm below, the value function is therefore approximated by
\begin{align}
	\hat V^{(k)}_{it}({c}_{t},h^{(k)}_{it} ; \bar{\bm\theta}) &
	= \log\left(\exp\hat{\mathcal{V}}^{(k)}_{i0t}\left({c}_{t},h^{(k)}_{it}; \bar{\bm\theta}\right)+ \exp\hat{\mathcal{V}}^{(k)}_{i1t}\left({c}_{t},h^{(k)}_{it}; \bar{\bm\theta}\right)\right),
\end{align}
and
\begin{align}
	\hat{\mathcal{V}}^{(k)}_{i0t}\left({c}_{t},h^{(k)}_{it} ; \bar{\bm\theta}\right)&
	=  \df\, \hat V^{(k)}_{it+1}\left({c}_{t+1}, \min\{\lambda^{(k)}_i h^{(k)}_{it}, \hat{h}_i\}; \bar{\bm\theta}\right)\\
	\hat{\mathcal{V}}^{(k)}_{i1t}\left({c}_{t},h^{(k)}_{it} ; \bar{\bm\theta}\right)&
	=\bar{\pBaseUtility} + \bar{\pMMoney}_{e_i} \textrm{NPV}_{it}\left({c}_t\right) +\bar{\pMPE}_{g_i} h^{(k)}_{it} + \bar{\randomEffects}_{i}.
\end{align}

Starting at the last period $T$, we compute each iteration's value function by backward induction. The equilibrium itself is computed by the outer loop: across iterations $k$, the value functions and the implied adopter paths converge to the fixed point of the ENAE algorithm below, $V_i^{(k)} \to V_i$ for each household $i$ and each value of the state $(\pvCost_t, \sInstBase_{it})$. At each iteration, the algorithm records the per-period adoption probabilities $P_{it}$, which embed the probability of not having adopted before $t$, and aggregates them into the expected number of adopters $A^{(k)}_t$.

In our estimation of the model detailed in Section~4.2, we cluster the households with similar characteristics in each census block group, and treat each cluster as a household weighted by the cluster size. Similarly, in the simulation, we compute the adoption probability of the representative household in each cluster, and then calculate the expected number of adopters in the cluster by multiplying this probability by the cluster size.

\textbf{Simulation cost.} A single counterfactual policy path solves the adoption equilibrium below (the fixed point of the ENAE algorithm) in about five outer-loop iterations and roughly 80 minutes of elapsed time on the workstation used for estimation; each policy evaluated in the counterfactual analysis is an independent run of this kind.

\begin{center}
\textbf{Expected Number of Adopters in Equilibrium (ENAE) Algorithm}
\vskip 10pt
{	
{\scriptsize 
	\begin{algorithm}[H]
		\SetAlgoLined
		\DontPrintSemicolon
		\SetKwInOut{Input}{Input}\SetKwInOut{Output}{Output}\SetKwInOut{Variables}{Variables}\SetKwInOut{Initialization}{Initialization}
		\Input{Schedule of net cost $c_t, t = 1,\ldots,T$.}
		\Output{Expected number of adopters for periods $t = 1,\ldots,T$.}
		\Initialization{No household has adopted in period $t = 1$. Iteration $k \leftarrow 1$. $\lambda_i = 1 $ for each household $i$}
		\While{the numbers of adopters per period have not converged}{
			\For{$i \in$ Households}{
				household $i$ adopts in $t = 1$ with probability $P_{i1} =  p_{i}(\pvCost_1,0;\bar{\bm\theta})$\;
			}
			\For(\tcp*[f]{the previous loop must finish first}){$i \in$ Households}{
				$h^{(k)}_{i1} =$ expected number of adopters at a 1-mile radius of household $i$ at the end of period $t = 1$\;
			}
			$A^{(k)}_{1} =$ expected number of adopters in period $1$ of iteration $k$\;
			\For{$t \in 2$ \KwTo $T$}{
				\For{$i \in$ Households that have not yet adopted}{
					household $i$ adopts in period $t$ with probability $P_{it} =  p_{i}(\pvCost_t,h^{(k)}_{it-1};\bar{\bm\theta}) * (1 - \sum_{\tau = 1}^{t-1}P_{i\tau})$\;
				}
				\For{$i \in$ Households}{
					$h^{(k)}_{it} =$ expected number of adopters at a 1-mile radius of household $i$ at the end of period $t$\;
				}
				$A^{(k)}_{t} =$ expected number of adopters in period $t$ of iteration $k$\;
			}
			\For{$i \in$ Households}{
				Estimate linear regression to update $\lambda_i$: $h^{(k)}_{it+1} = \lambda_i h^{(k)}_{it}$\; 
			}
			$k \leftarrow k+1$;
		}	
	\end{algorithm}
}
}
\end{center}

\clearpage
\section{Supporting Figures for Bayesian Posterior Sampling and Policy Simulations} \label{sec:SupportingFigures}

\textbf{Retained sample and convergence diagnostics.} We run the sampler in Appendix C for 18{,}000 iterations under a 4-hour wall-clock limit on a single processor. The chain reaches its stationary distribution within the first 2{,}500 iterations, which we discard as burn-in, and we compute the posterior means and 95\% credible intervals of the structural parameters from the retained draws; the parameter trace and density plots below display the chains and confirm convergence. Mixing differs across parameters: the heterogeneity scale and the majority-segment net-present-value coefficient yield effective sample sizes in the thousands, whereas the remaining coefficients, and the neighborhood-effect coefficients in particular, are more strongly autocorrelated and yield effective sample sizes in the tens to low hundreds; we therefore run each chain well past burn-in rather than stopping once stationarity is reached. The reported estimates are stable under longer sampling: extending the run from 4 to 8 hours leaves each posterior mean within roughly a quarter of a posterior standard deviation of its reported value, with the exception of the weakly-mixed urban neighborhood-effect coefficient, whose mean shifts by about a third of a standard deviation, consistent with its low effective sample size.

\begin{figure}[H]
{	\centering
		\includegraphics[scale=0.42]{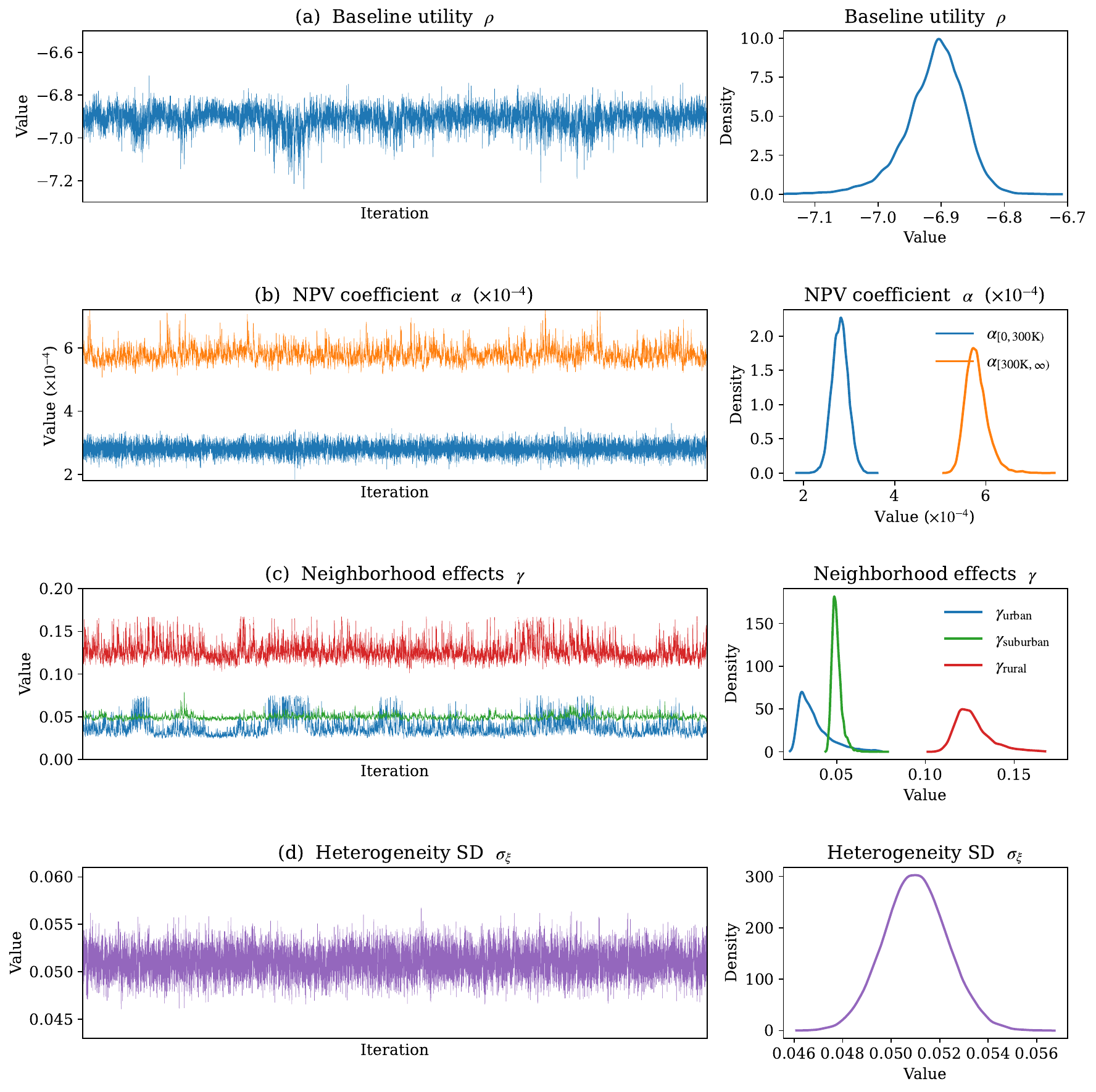}
		\caption{MCMC trace plots of the best model, Model (4)H, structural parameters.\label{fig:mcmc1_exVIa-RE}}
	}
	{\footnotesize \setlength{\baselineskip}{10pt}
		\noindent \emph{Note.} From top to bottom: (a) $\pBaseUtility$, the constant term of the utility function; (b) $\pMMoney$, the preference for the economic benefit of solar per economic segment; (c) $\pMPE$, the neighborhood effects per geographic segment; and (d) $\sigma_{\xi}$, the standard deviation of the household random effects. \textit{Note}: the first 2{,}500 iterations, the \textit{burn-in} period, are discarded because the initial samples may not accurately represent the target distribution due to the starting point of the chain. 
		\par}
\end{figure}

\begin{figure}[H]
	\begin{center}
		\includegraphics[width=0.95\textwidth]{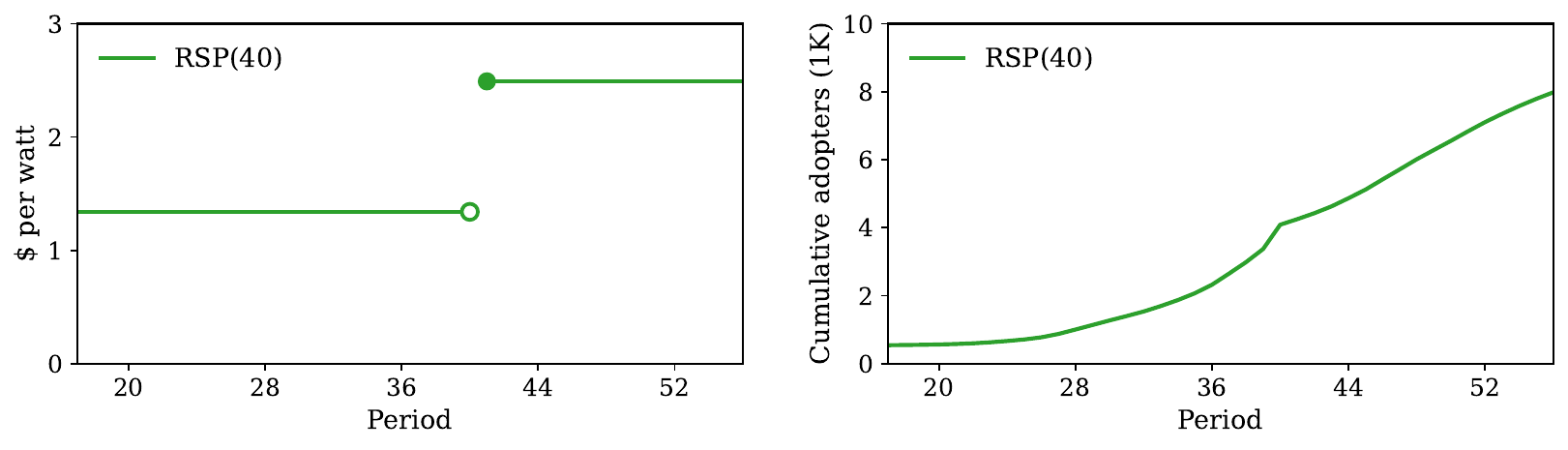}
		\caption{Counterfactual trajectory under the one-step rebate switching policy RSP(40).\label{fig:ECrsp40}}
	\end{center}
	
	{\footnotesize 	\setlength{\baselineskip}{10pt}
	\noindent \emph{Note.} (a) net cost per watt; (b) cumulative adopters. The adoption rate rises as the pre-announced rebate switching period approaches and drops once the rebate ends, producing the kink at period~40 discussed in Section~6.1 of the paper.
\par}

\end{figure}
\clearpage
\section{Additional Model Comparison and Selection Results}\label{sec:DetailResultsTables}

In this appendix, we report the results from the additional models we tested. To assess the robustness of our model selection, we compare out-of-sample predictive performance using both the mean absolute error (MAE) and the mean absolute percentage error (MAPE). The results are reported in Table~\ref{tab:summaryResultsAllModels}.

\begin{table}[H]
	\begin{center}
		\caption{Comparison of various model specifications.\label{tab:summaryResultsAllModels}}
		{\scriptsize
			\begin{tabular}{lllrrrr}
				\toprule
				& \multicolumn{1}{c}{\textbf{Economic Segment (${E}$)}} & \multicolumn{1}{c}{\textbf{Neighborhood Radial Ring (${D}$)}} & \multicolumn{2}{c} {\textbf{MAE (StdMAE)}} & \multicolumn{2}{c} {\textbf{MAPE (StdMAPE) [\%]}}\bigstrut\\
				\cline{4-5} \cline{6-7} 
				\multicolumn{1}{c}{Model} & \multicolumn{1}{p{3cm}}{\centering Home Mkt Value (\$’000) } & \multicolumn{1}{p{3cm}}{\centering Radial Distance Brackets (Miles)} & \multicolumn{1}{c}{NH} & \multicolumn{1}{c}{H}& \multicolumn{1}{c}{NH} & \multicolumn{1}{c}{H}  \bigstrut\\\midrule
				(1)  &  [0,\Inf)  &  no neighborhood effects &  
				67.726 (5.994) & 71.482 (4.974) &  
				29.778 (2.852) & 31.552 (2.373)
				\bigstrut\\
				(2)  &  [0,\Inf)  &  [0,1)  &
				37.343 (12.247) & 32.656 (6.721) &
				19.314 (6.036) & 16.956 (3.496)
				\bigstrut \\
				(3)  &  [0,300),[300,\Inf)  &   no neighborhood effects &
				85.001 (3.871) & 85.511 (3.661) & 
				38.144 (1.861) & 38.386 (1.764)   
				\bigstrut \\
				(4)  &  [0,300),[300,\Inf)  &  [0,1)  &
				45.768 (14.155) & 31.845 (2.412) &
				23.407 (6.641) & 15.599 (1.552) 
				\bigstrut \\
				(5)  &  [0,300),[300,600),[600,\Inf)  &   no neighborhood effects &
				69.011 (4.542) & 69.627 (4.353) &
				30.889 (2.151) & 31.177 (2.061) 
				\bigstrut\\
				(6)  &   [0,300),[300,600),[600,\Inf)  &  [0,1)  &
				48.534 (5.587) & 49.142 (5.437) & 
				23.904 (2.925) & 24.118 (2.797) 
				\bigstrut\\\hline
				(7)  &  [0,\Inf)  & [0,0.5) & 
				37.331 (4.904) & 32.106 (2.006) &
				16.591 (1.703) & 15.074 (1.161) 
				\bigstrut \\
				(8)  &  [0,\Inf)  &  [0,0.5),[0.5,1)      & 
				33.516 (6.168) & 30.550 (4.012) &
				16.709 (3.001) & 15.925 (2.168)
				\bigstrut \\
				(9)  &  [0,\Inf)  &  [0,1),[1,SB) & 
				37.969 (6.740) & 38.463 (6.621) &
				16.375 (2.936) & 16.503 (2.942) 
				\bigstrut \\
				(10)  &  [0,300),[300,\Inf)  &  [0,0.5) & 
				31.651 (1.597) & 31.373 (1.180) &
				14.479 (0.931) & 14.197 (0.638) 
				\bigstrut \\
				(11) &  [0,300),[300,\Inf)  &  [0,0.5),[0.5,1)  & 
				33.055 (3.929) & 31.826 (2.613)&
				17.141 (2.244) & 16.028 (1.534)
				\bigstrut \\
				(12)  &  [0,300),[300,\Inf)  &  [0,1),[1,SB) &  
				38.152 (6.754) & 39.365 (7.304)&
				16.576 (3.046) & 17.112 (3.319)
				\bigstrut \\
				(13)  &  [0,300),[300,600),[600,\Inf)  &  [0,0.5)  &  
				44.348 (3.174) & 43.184 (3.069) &
				20.972 (1.779) & 20.386 (1.741) 
				\bigstrut \\
				(14)   &  [0,300),[300,600),[600,\Inf)  &  [0,0.5),[0.5,1)  &
				51.830 (7.820) & 48.625 (6.994)&
				26.044 (3.714) & 24.492 (3.421)
				\bigstrut \\
				(15)  &   [0,300),[300,600),[600,\Inf)  &  [0,1),[1,SB) & 
				37.077 (3.032) & 36.472 (2.768) &
				16.473 (1.699) & 16.166 (1.577) 
				\bigstrut[b] \\
				\bottomrule
				\multicolumn{7}{l}{Note: (a) NH: without heterogeneity, H: with heterogeneity;}\\ 
				\multicolumn{7}{l}{$\quad\quad~~$(b) Training periods: 1-33, test periods: 34-38;}\\
				\multicolumn{7}{l}{$\quad\quad~~$(c) In estimation: $\hat{h}_{it+1} = h_{it+1}$, in prediction: $\hat{h}_{it+1} = \lambda_i \tilde{h}_{it}$;}\\ 
				\multicolumn{7}{l}{$\quad\quad~~$(d) MAE is calculated from 1000 simulations for out-of-sample periods set $\mathcal{T}$ as the average absolute error between} \\
				\multicolumn{7}{l}{$\quad\quad\quad~~~~$the actual number of adopters $N_t$ in period $t$ and the predicted number $\hat{N}^k_t$ in simulation $k$ and period $t$,} \\
				\multicolumn{7}{l}{$\quad\quad\quad~~~~$measured in number of adopters; MAPE is the corresponding mean absolute percentage error.
				} \\
			\end{tabular}\\
		}
	\end{center}
\end{table}
\clearpage
\subsection*{Stability of the structural parameters across specifications}\label{sec:ECparamstability}

This section tests the robustness of our results to the model specification, extending the comparison in Table~\ref{tab:summaryResultsAllModels} across the economic segmentations $E$ and the neighborhood radial rings $D$. Two facts anchor the comparison. First, extending the neighborhood beyond one mile adds only a negligible neighborhood effect: in the two-ring models (9), (12), and (15), the coefficient on adopters beyond one mile is near $-4\times10^{-4}$, three orders of magnitude below the inner-ring coefficient (Table~\ref{tab:ECparamGamma}). Second, the leading one-mile and half-mile specifications predict almost equally well: their out-of-sample MAEs differ by less than 2\% (for the two-segment models, 31.85 for Model~(4)H against 31.37 for Model~(10)H). Since predictive accuracy alone does not separate these specifications, we examine the structural parameter estimates, reported for the twelve heterogeneity (H) specifications with neighborhood effects, models (2), (4), (6), and (7)--(15).

The estimates are stable across these alternative specifications (Tables~\ref{tab:ECparamRhoAlpha} and~\ref{tab:ECparamGamma} and Figure~\ref{fig:ECparamstability}). The heterogeneity standard deviation $\sigma_{\xi}$ is statistically indistinguishable across the twelve specifications, its $95\%$ credible intervals all overlapping near $0.051$. The baseline utility $\rho$ stays within the narrow band $[-6.85,-6.64]$ and is precisely estimated (low posterior variance) in every specification. The net-present-value coefficient $\alpha$ is larger for higher-value home segments wherever the segment boundaries align, so higher-value households ascribe more weight to the economic benefit of a PV system: the low-value coefficient $\alpha_{[0,300)}$ has eight mutually overlapping $95\%$ credible intervals around $3\times10^{-4}$, the two-segment high-value coefficient overlapping intervals around $6.2\times10^{-4}$, and the three-segment $[300,600)$ coefficient overlapping intervals around $8.5\times10^{-4}$, so no specification's estimate is distinguishable from the others at the $95\%$ level; the single-segment models, which pool all home values into one economic segment, estimate a single coefficient near $4.1\times10^{-4}$. The only imprecisely identified coefficient is the top $[600,\infty)$ bin, whose credible intervals are wide and rest on few very-high-value adopters, consistent with the expanding-window evidence in Appendix~F. The neighborhood effect preserves the ordering $\gamma_{\mathrm{rural}}>\gamma_{\mathrm{suburban}}>\gamma_{\mathrm{urban}}$ in every one-mile-ring specification, and a ring standing alone recovers the same one-mile coefficients as when it is the inner ring of a two-ring model, the standalone and inner-ring $\gamma_{\mathrm{rural}}$ intervals overlapping (means $0.121$ and $0.135$).

The neighborhood coefficient does scale with the ring width, which is mechanical: a half-mile ring aggregates far fewer installed systems, so the per-neighbor coefficient is larger and, in sparsely built rural areas, noisier (the half-mile rural coefficient carries wider credible intervals, for example $(0.13,0.40)$ in Model~(8) against $(0.11,0.16)$ for the one-mile ring in Model~(4)). This is why the half-mile specifications that match Model~(4)H on out-of-sample error, such as Model~(10), do so with a less stable and less interpretable near-neighbor coefficient and a geographic ordering that is no longer monotone. The one-mile ring therefore delivers the same forecast accuracy to within the noise while giving a well-identified, monotone neighborhood structure, and we use the one-mile radial ring, with two economic segments (Model~(4)H), as the parsimonious specification of our model.

\begin{table}[H]
	\caption{Posterior mean and 95\% credible interval of $\rho, \sigma_{\xi},\alpha$ of Models (2), (4), and (6) to (15).\label{tab:ECparamRhoAlpha}}
	{\centering
		\footnotesize
	\setlength{\tabcolsep}{4pt}\renewcommand{\arraystretch}{1.1}
	\resizebox{\textwidth}{!}{%
		\begin{tabular}{clccccc}
			\toprule
			Model & Economic segments $E$ & $\rho$ & $\sigma_{\xi}$ & $\alpha_1$ & $\alpha_2$ & $\alpha_3$ \\
			\midrule
				(2) & $[0,\infty)$ & $-$6.80990 & 0.05094 & 0.00041 &  &  \\
				 &  & {\scriptsize ($-$6.94269, $-$6.70937)} & {\scriptsize (0.04839, 0.05364)} & {\scriptsize (0.00038, 0.00047)} &  &  \\
				\addlinespace[3pt]
				(4) & $[0,300),[300,\infty)$ & $-$6.79158 & 0.05098 & 0.00029 & 0.00064 &  \\
				 &  & {\scriptsize ($-$6.89602, $-$6.71202)} & {\scriptsize (0.04851, 0.05368)} & {\scriptsize (0.00026, 0.00032)} & {\scriptsize (0.00059, 0.00070)} &  \\
				\addlinespace[3pt]
				(6) & $[0,300),[300,600),[600,\infty)$ & $-$6.81902 & 0.05180 & 0.00029 & 0.00086 & 0.00052 \\
				 &  & {\scriptsize ($-$6.91498, $-$6.73623)} & {\scriptsize (0.04924, 0.05443)} & {\scriptsize (0.00026, 0.00033)} & {\scriptsize (0.00081, 0.00093)} & {\scriptsize (0.00024, 0.00097)} \\
				\addlinespace[3pt]
				(7) & $[0,\infty)$ & $-$6.74818 & 0.05100 & 0.00043 &  &  \\
				 &  & {\scriptsize ($-$6.83022, $-$6.67359)} & {\scriptsize (0.04852, 0.05378)} & {\scriptsize (0.00039, 0.00049)} &  &  \\
				\addlinespace[3pt]
				(8) & $[0,\infty)$ & $-$6.81625 & 0.05102 & 0.00041 &  &  \\
				 &  & {\scriptsize ($-$6.90705, $-$6.72905)} & {\scriptsize (0.04848, 0.05360)} & {\scriptsize (0.00037, 0.00045)} &  &  \\
				\addlinespace[3pt]
				(9) & $[0,\infty)$ & $-$6.63613 & 0.05189 & 0.00041 &  &  \\
				 &  & {\scriptsize ($-$6.72372, $-$6.54712)} & {\scriptsize (0.04929, 0.05467)} & {\scriptsize (0.00038, 0.00045)} &  &  \\
				\addlinespace[3pt]
				(10) & $[0,300),[300,\infty)$ & $-$6.74915 & 0.05109 & 0.00031 & 0.00062 &  \\
				 &  & {\scriptsize ($-$6.82449, $-$6.67770)} & {\scriptsize (0.04850, 0.05379)} & {\scriptsize (0.00028, 0.00034)} & {\scriptsize (0.00058, 0.00068)} &  \\
				\addlinespace[3pt]
				(11) & $[0,300),[300,\infty)$ & $-$6.81931 & 0.05107 & 0.00030 & 0.00062 &  \\
				 &  & {\scriptsize ($-$6.90434, $-$6.74042)} & {\scriptsize (0.04863, 0.05375)} & {\scriptsize (0.00026, 0.00033)} & {\scriptsize (0.00058, 0.00068)} &  \\
				\addlinespace[3pt]
				(12) & $[0,300),[300,\infty)$ & $-$6.64954 & 0.05114 & 0.00029 & 0.00060 &  \\
				 &  & {\scriptsize ($-$6.73396, $-$6.56284)} & {\scriptsize (0.04864, 0.05385)} & {\scriptsize (0.00026, 0.00032)} & {\scriptsize (0.00057, 0.00065)} &  \\
				\addlinespace[3pt]
				(13) & $[0,300),[300,600),[600,\infty)$ & $-$6.78787 & 0.05099 & 0.00032 & 0.00086 & 0.00034 \\
				 &  & {\scriptsize ($-$6.86890, $-$6.71528)} & {\scriptsize (0.04842, 0.05369)} & {\scriptsize (0.00029, 0.00035)} & {\scriptsize (0.00081, 0.00093)} & {\scriptsize (0.00017, 0.00074)} \\
				\addlinespace[3pt]
				(14) & $[0,300),[300,600),[600,\infty)$ & $-$6.84883 & 0.05097 & 0.00030 & 0.00085 & 0.00026 \\
				 &  & {\scriptsize ($-$6.93669, $-$6.76698)} & {\scriptsize (0.04848, 0.05362)} & {\scriptsize (0.00027, 0.00034)} & {\scriptsize (0.00080, 0.00091)} & {\scriptsize (0.00015, 0.00041)} \\
				\addlinespace[3pt]
				(15) & $[0,300),[300,600),[600,\infty)$ & $-$6.68338 & 0.05106 & 0.00029 & 0.00082 & 0.00033 \\
				 &  & {\scriptsize ($-$6.76686, $-$6.59498)} & {\scriptsize (0.04857, 0.05379)} & {\scriptsize (0.00026, 0.00033)} & {\scriptsize (0.00078, 0.00087)} & {\scriptsize (0.00016, 0.00061)} \\
			\bottomrule
		\end{tabular}}}
	
\vspace{0.2cm}
{\footnotesize 
	\setlength{\baselineskip}{10pt}
	\noindent \emph{Note.}
{Posterior mean and 95\% credible interval of the baseline utility $\rho$, the heterogeneity standard deviation $\sigma_{\xi}$, and the net-present-value coefficient $\alpha$ across specifications. The coefficients $\alpha_1,\alpha_2,\alpha_3$ correspond, in order, to the economic segments listed in $E$.}
{Empty cells indicate segments not present in that model.}
\par}	
	
\end{table}

\begin{table}[H]
	\caption{Posterior mean and 95\% credible interval of $\gamma$ of Models (2), (4), and (6) to (15).\label{tab:ECparamGamma}}
{\centering\footnotesize
	\setlength{\tabcolsep}{4pt}\renewcommand{\arraystretch}{1.1}
	\resizebox{\textwidth}{!}{%
		\begin{tabular}{clcccccc}
			\toprule
			& & \multicolumn{3}{c}{Inner ring} & \multicolumn{3}{c}{Outer ring}\\
			\cmidrule(lr){3-5}\cmidrule(lr){6-8}
			Model & Rings $D$ & urban & suburban & rural & urban & suburban & rural \\
			\midrule
				(2) & $[0,1)$ & 0.04703 & 0.05135 & 0.13372 &  &  &  \\
				 &  & {\scriptsize (0.02800, 0.08690)} & {\scriptsize (0.04685, 0.05984)} & {\scriptsize (0.11403, 0.19105)} &  &  &  \\
				\addlinespace[3pt]
				(4) & $[0,1)$ & 0.03611 & 0.04932 & 0.12135 &  &  &  \\
				 &  & {\scriptsize (0.02695, 0.05801)} & {\scriptsize (0.04518, 0.05853)} & {\scriptsize (0.10514, 0.15563)} &  &  &  \\
				\addlinespace[3pt]
				(6) & $[0,1)$ & 0.03726 & 0.04689 & 0.10768 &  &  &  \\
				 &  & {\scriptsize (0.02696, 0.06488)} & {\scriptsize (0.04338, 0.05249)} & {\scriptsize (0.09368, 0.13662)} &  &  &  \\
				\addlinespace[3pt]
				(7) & $[0,0.5)$ & 0.14120 & 0.06114 & 0.35359 &  &  &  \\
				 &  & {\scriptsize (0.12194, 0.17290)} & {\scriptsize (0.05654, 0.06987)} & {\scriptsize (0.30477, 0.46793)} &  &  &  \\
				\addlinespace[3pt]
				(8) & $[0,0.5),[0.5,1)$ & 0.12687 & 0.05445 & 0.21106 & 0.01648 & 0.04654 & 0.10549 \\
				 &  & {\scriptsize (0.09766, 0.16437)} & {\scriptsize (0.04946, 0.06066)} & {\scriptsize (0.12858, 0.40122)} & {\scriptsize (0.00339, 0.03799)} & {\scriptsize (0.03185, 0.07586)} & {\scriptsize (0.05398, 0.16656)} \\
				\addlinespace[3pt]
				(9) & $[0,1),[1,\infty)$ & 0.03440 & 0.04932 & 0.13863 & $-$0.00037 & $-$0.00015 & $-$0.00035 \\
				 &  & {\scriptsize (0.02905, 0.04308)} & {\scriptsize (0.04576, 0.05456)} & {\scriptsize (0.12058, 0.16632)} & {\scriptsize ($-$0.00057, $-$0.00019)} & {\scriptsize ($-$0.00029, $-$0.00001)} & {\scriptsize ($-$0.00060, $-$0.00009)} \\
				\addlinespace[3pt]
				(10) & $[0,0.5)$ & 0.13585 & 0.05878 & 0.32864 &  &  &  \\
				 &  & {\scriptsize (0.11765, 0.17079)} & {\scriptsize (0.05438, 0.06551)} & {\scriptsize (0.28539, 0.40965)} &  &  &  \\
				\addlinespace[3pt]
				(11) & $[0,0.5),[0.5,1)$ & 0.11830 & 0.05309 & 0.19035 & 0.01387 & 0.04687 & 0.09366 \\
				 &  & {\scriptsize (0.09466, 0.14641)} & {\scriptsize (0.04874, 0.05987)} & {\scriptsize (0.11599, 0.27162)} & {\scriptsize (0.00445, 0.03165)} & {\scriptsize (0.03127, 0.07442)} & {\scriptsize (0.05600, 0.14257)} \\
				\addlinespace[3pt]
				(12) & $[0,1),[1,\infty)$ & 0.03578 & 0.04764 & 0.13672 & $-$0.00039 & $-$0.00013 & $-$0.00046 \\
				 &  & {\scriptsize (0.02994, 0.04855)} & {\scriptsize (0.04436, 0.05201)} & {\scriptsize (0.11822, 0.16704)} & {\scriptsize ($-$0.00067, $-$0.00020)} & {\scriptsize ($-$0.00026, 0.00001)} & {\scriptsize ($-$0.00074, $-$0.00016)} \\
				\addlinespace[3pt]
				(13) & $[0,0.5)$ & 0.13146 & 0.05690 & 0.30728 &  &  &  \\
				 &  & {\scriptsize (0.11391, 0.16309)} & {\scriptsize (0.05219, 0.06477)} & {\scriptsize (0.26092, 0.42666)} &  &  &  \\
				\addlinespace[3pt]
				(14) & $[0,0.5),[0.5,1)$ & 0.11554 & 0.05017 & 0.17273 & 0.01577 & 0.04415 & 0.08735 \\
				 &  & {\scriptsize (0.09113, 0.14541)} & {\scriptsize (0.04676, 0.05502)} & {\scriptsize (0.08742, 0.26271)} & {\scriptsize (0.00499, 0.03407)} & {\scriptsize (0.03084, 0.07583)} & {\scriptsize (0.04775, 0.14083)} \\
				\addlinespace[3pt]
				(15) & $[0,1),[1,\infty)$ & 0.03388 & 0.04539 & 0.13092 & $-$0.00037 & $-$0.00011 & $-$0.00052 \\
				 &  & {\scriptsize (0.02893, 0.04551)} & {\scriptsize (0.04277, 0.04905)} & {\scriptsize (0.11465, 0.15367)} & {\scriptsize ($-$0.00061, $-$0.00017)} & {\scriptsize ($-$0.00023, 0.00002)} & {\scriptsize ($-$0.00078, $-$0.00025)} \\
			\bottomrule
		\end{tabular}}}
	
	\vspace{0.2cm}
{\footnotesize 
\setlength{\baselineskip}{10pt}
\noindent \emph{Note.} 		Posterior mean and 95\% credible interval of the neighborhood effect $\gamma$ by geographic segment and radial ring across specifications. Inner and outer denote the first and second rings in $D$; single-ring models have no outer ring, and the beyond-one-mile (outer) effects in models (9), (12), and (15) are negligible, three orders of magnitude below the inner-ring effects.
\par}
	
\end{table}

\begin{figure}[H]
{	\centering
	\includegraphics[width=\textwidth]{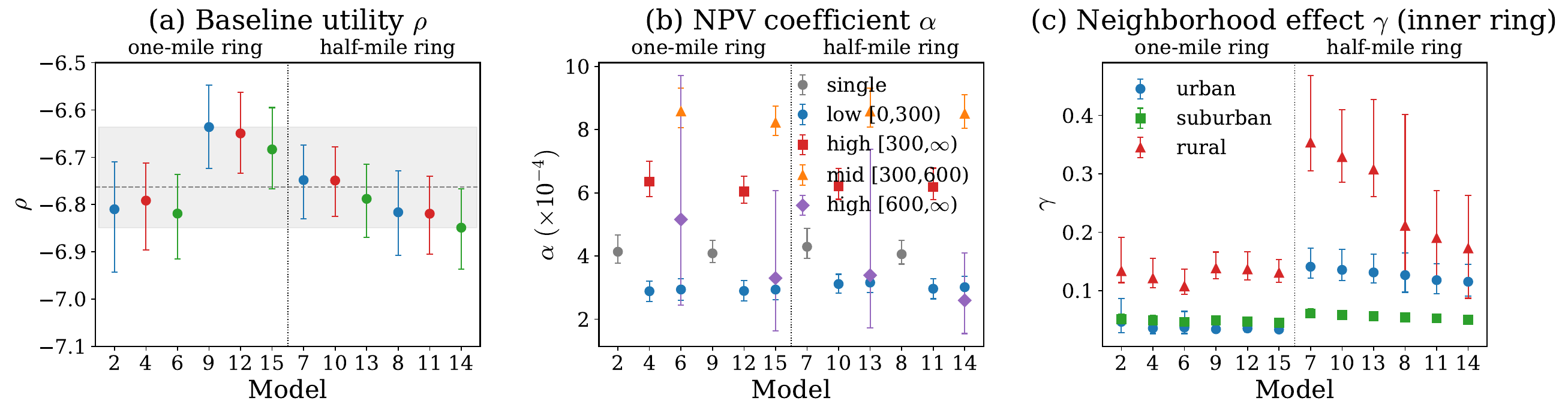}
	\caption{Stability of the structural parameters across economic-segment and radial-ring specifications. \label{fig:ECparamstability}}
}
{\footnotesize 
	\setlength{\baselineskip}{9pt}
(a) baseline utility $\rho$, with the dashed line the mean and the shaded band the range; (b) the net-present-value coefficient $\alpha$ by economic segment; (c) the neighborhood effect $\gamma$ for the inner ring by geographic segment. Whiskers are $95\%$ credible intervals. Models left of the dotted line use a one-mile inner ring; those to the right use a half-mile inner ring.	
	\noindent \emph{Note.}
\par}

\end{figure}

\begin{table}[H]
	\begin{center}
		\caption{Predicted performance of models using alternative function specifications of $\textbf{h}$ in the utility function.\label{tab:summaryResults3}}
		{\scriptsize 
			\begin{tabular}{ccc}
				\toprule
				\multicolumn{1}{c}{\textbf{Specification of $\textbf{h}$ in utility}} & \multicolumn{1}{c} {\textbf{MAE (StdMAE)}} & \multicolumn{1}{c} {\textbf{MAPE (StdMAPE) [\%]}}\bigstrut\\
				\midrule
				$h$ & 31.845 (2.412)& 15.599 (1.552) 
				\bigstrut\\
				$\log(1+h)$ & 106.901 (6.648) & 54.746 (4.221) \bigstrut \\
				$\sqrt{h}$ & 74.434 (5.631) & 37.672 (4.208)  \bigstrut \\
				\bottomrule
				\multicolumn{3}{l}{Note: (a) Training periods: 1-33, test periods: 34-38;}\\
				\multicolumn{3}{l}{$\quad\quad~~$(b) In estimation: $\hat{h}_{it+1} = h_{it+1}$, in prediction: $\hat{h}_{it+1} = \lambda_i \hat{h}_{it}$;}\\
				\multicolumn{3}{l}{$\quad\quad~~$(c) Economic Segments (${E}$): $[0,300),[300, \infty)$;}\\
				\multicolumn{3}{l}{$\quad\quad~~$(d) Neighborhood Radial Ring (${D}$): $[0,1)$.}\\
			\end{tabular}
		}
	\end{center}
\end{table}


\clearpage
\section{Estimation with Expanding Data Window}\label{sec:ECexpandingwindow}
We assess the consistency of the estimates through an expanding-window re-estimation. We re-estimate the selected model, Model~(4)H, on the first $T$ quarters of the sample at even windows from $T=26$ (2010Q2) through $T=36$ (2012Q4), together with the full estimation window $T=33$ (2012Q1) reported in the paper, and track the posterior mean and 95\% credible interval of each structural parameter as the window grows. Figure~\ref{fig:ECexpwindow} plots the trajectories grouped by parameter family, on axes wide enough to span the estimates at every window; Table~\ref{tab:ECexpwindow} reports the underlying per-window estimates.

The net-present-value coefficients $\alpha$ of the two economic segments cross between windows $T=30$ and $T=32$ (Figure~\ref{fig:ECexpwindow}b). In the early windows ($T\le 30$, through 2011Q2) the coefficient for the high-value $[300\mathrm{K},\infty)$ segment sits below that of the majority $[0,300\mathrm{K})$ segment, estimated near $2\times10^{-4}$: few high-value households have adopted by then, so their net-present-value coefficient has greater estimation error, reflected in wider interval estimates. As the window extends through the high-adoption quarters of 2011--2012, more high-value adoptions enter the sample and $\alpha_{[300\mathrm{K},\infty)}$ rises, crossing above $\alpha_{[0,300\mathrm{K})}$ at $T=32$ and settling near its full-sample value of $6.4\times10^{-4}$ at $T=33$. The crossover marks improved identification of the high-value segment's price response rather than a change in preferences: once enough high-value adoptions are observed, the ordering $\alpha_{[300\mathrm{K},\infty)}>\alpha_{[0,300\mathrm{K})}$ reported for the full sample is recovered and holds for every window from $T=32$ onward.

The remaining parameters support the stability reported in Section~5 of the paper (Figure~\ref{fig:ECexpwindow}; Table~\ref{tab:ECexpwindow}). The standard deviation of household heterogeneity $\sigma_{\xi}$ varies only between $0.0504$ and $0.0515$, the baseline utility $\rho$ stays within a narrow band ($-7.19$ to $-6.76$), and the coefficient of the majority segment $\alpha_{[0,300\mathrm{K})}$ holds its order of magnitude, near $3$ to $4\times10^{-4}$. The coefficients $\gamma$ for neighborhood effects are a little more sensitive to the sample sizes: their estimates are larger in the short early windows, with wider interval estimates, and gradually level off as the larger sampling window absorbs more data, when falling net costs account for a growing share of adoption. Their economic interpretation is nevertheless invariant over time: the ordering $\gamma_{\mathrm{rural}}>\gamma_{\mathrm{suburban}}>\gamma_{\mathrm{urban}}$ is preserved and no coefficient changes sign. By the full estimation window $T=33$, the heterogeneity scale, the baseline utility, and the majority-segment coefficient have settled and the high-value coefficient is identified, the stability a policymaker needs before relying on the estimates for counterfactual analysis.

\begin{table}[!h]
\caption{Posterior mean and 95\% credible interval of the structural parameters of Model~(4)H across expanding estimation windows\label{tab:ECexpwindow}}
{\centering
{\footnotesize
\setlength{\tabcolsep}{4pt}\renewcommand{\arraystretch}{1.1}
\resizebox{\textwidth}{!}{%
\begin{tabular}{lccccccc}
\toprule
Parameter & $T=26$ & $T=28$ & $T=30$ & $T=32$ & \textbf{$T=33$} & $T=34$ & $T=36$ \\
\midrule
$\rho$ & $-$7.10438 & $-$7.14538 & $-$7.19179 & $-$7.04018 & \textbf{$-$6.79158} & $-$6.76025 & $-$6.84665 \\
 & {\scriptsize ($-$7.19273, $-$7.01396)} & {\scriptsize ($-$7.22766, $-$7.07685)} & {\scriptsize ($-$7.26856, $-$7.10772)} & {\scriptsize ($-$7.11829, $-$6.96704)} & {\scriptsize ($-$6.89602, $-$6.71202)} & {\scriptsize ($-$6.82481, $-$6.69228)} & {\scriptsize ($-$6.90845, $-$6.78334)} \\
\addlinespace[3pt]
$\alpha_{[0,300\mathrm{K})}$ & 0.00040 & 0.00039 & 0.00042 & 0.00033 & \textbf{0.00029} & 0.00028 & 0.00040 \\
 & {\scriptsize (0.00036, 0.00044)} & {\scriptsize (0.00036, 0.00043)} & {\scriptsize (0.00038, 0.00046)} & {\scriptsize (0.00030, 0.00037)} & {\scriptsize (0.00026, 0.00032)} & {\scriptsize (0.00025, 0.00031)} & {\scriptsize (0.00037, 0.00043)} \\
\addlinespace[3pt]
$\alpha_{[300\mathrm{K},\infty)}$ & 0.00023 & 0.00018 & 0.00025 & 0.00065 & \textbf{0.00064} & 0.00076 & 0.00088 \\
 & {\scriptsize (0.00017, 0.00031)} & {\scriptsize (0.00012, 0.00027)} & {\scriptsize (0.00019, 0.00036)} & {\scriptsize (0.00060, 0.00070)} & {\scriptsize (0.00059, 0.00070)} & {\scriptsize (0.00071, 0.00082)} & {\scriptsize (0.00085, 0.00090)} \\
\addlinespace[3pt]
$\gamma_{\mathrm{urban}}$ & 0.07531 & 0.05378 & 0.05164 & 0.04569 & \textbf{0.03611} & 0.02430 & 0.01330 \\
 & {\scriptsize (0.06744, 0.08689)} & {\scriptsize (0.04792, 0.06109)} & {\scriptsize (0.04627, 0.05929)} & {\scriptsize (0.03765, 0.06005)} & {\scriptsize (0.02695, 0.05801)} & {\scriptsize (0.01856, 0.03579)} & {\scriptsize (0.00901, 0.02063)} \\
\addlinespace[3pt]
$\gamma_{\mathrm{suburban}}$ & 0.11858 & 0.08577 & 0.11679 & 0.07235 & \textbf{0.04932} & 0.03979 & 0.02209 \\
 & {\scriptsize (0.10389, 0.14250)} & {\scriptsize (0.07663, 0.09954)} & {\scriptsize (0.10990, 0.12572)} & {\scriptsize (0.06793, 0.07957)} & {\scriptsize (0.04518, 0.05853)} & {\scriptsize (0.03732, 0.04318)} & {\scriptsize (0.02000, 0.02493)} \\
\addlinespace[3pt]
$\gamma_{\mathrm{rural}}$ & 0.21835 & 0.15478 & 0.15084 & 0.13673 & \textbf{0.12135} & 0.17655 & 0.07860 \\
 & {\scriptsize (0.19393, 0.25007)} & {\scriptsize (0.14232, 0.17272)} & {\scriptsize (0.13646, 0.17118)} & {\scriptsize (0.12302, 0.15766)} & {\scriptsize (0.10514, 0.15563)} & {\scriptsize (0.16047, 0.20891)} & {\scriptsize (0.07068, 0.09236)} \\
\addlinespace[3pt]
$\sigma_{\xi}$ & 0.05051 & 0.05039 & 0.05052 & 0.05079 & \textbf{0.05098} & 0.05126 & 0.05153 \\
 & {\scriptsize (0.04802, 0.05313)} & {\scriptsize (0.04806, 0.05283)} & {\scriptsize (0.04801, 0.05320)} & {\scriptsize (0.04824, 0.05364)} & {\scriptsize (0.04851, 0.05368)} & {\scriptsize (0.04856, 0.05392)} & {\scriptsize (0.04909, 0.05421)} \\
\addlinespace[3pt]
\bottomrule
\end{tabular}}}}

\vspace{0.2cm}
{\footnotesize 
\setlength{\baselineskip}{10pt}
\noindent \emph{Note.}
Column $T$ re-estimates the model on quarters $1$ to $T$; the $T=33$ column (boldface) reproduces the full-sample estimates reported in Table~5 of the paper.
$\rho$ baseline utility; $\alpha$ net-present-value coefficient by home-value segment ($[0,300\mathrm{K})$ and $[300\mathrm{K},\infty)$); $\gamma$ neighborhood effect for a one-mile ring by geographic segment(urban, suburban, rural); $\sigma_{\xi}$ standard deviation of household heterogeneity. Calendar mapping: $T=26$ is 2010Q2, $T=33$ is 2012Q1, and $T=36$ is 2012Q4.
\par}

\end{table}

\begin{figure}[H]
\begin{center}
	\includegraphics[width=0.8\textwidth]{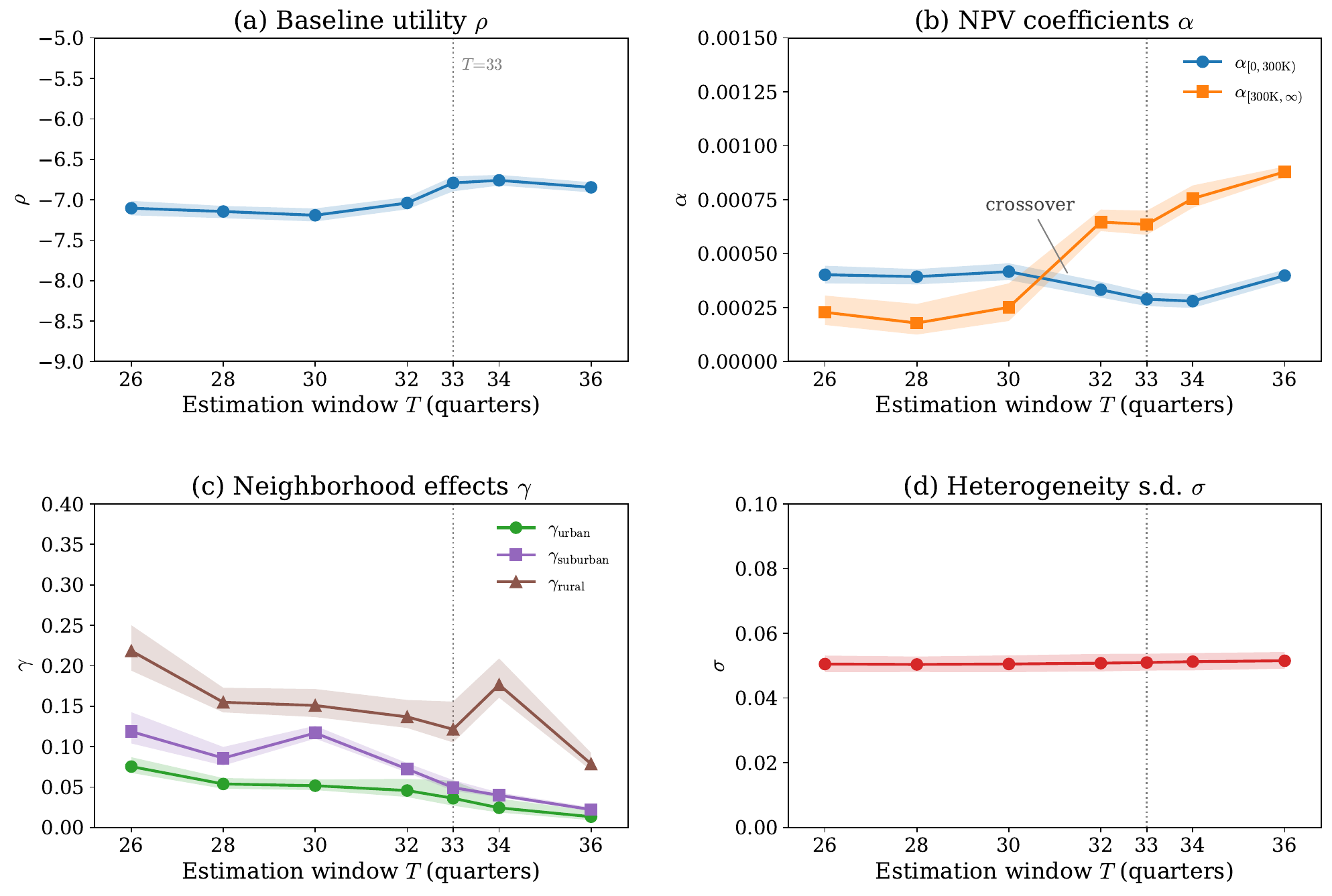}
	\caption{Expanding-window estimates of the structural parameters of Model~(4)H. \label{fig:ECexpwindow}}
\end{center}
{\footnotesize 
	\setlength{\baselineskip}{10pt}
	\noindent \emph{Note.}
	Expanding-window estimates of the structural parameters of Model~(4)H over even windows $T=26$ to $36$ together with the full estimation window $T=33$ reported in the paper. Each panel shows the posterior mean (markers) and 95\% credible interval (shaded band).
	\par}
\end{figure}
\clearpage
\section{Benchmark Model Details}\label{sec:BassModel}


\subsection*{Benchmark (a): Hazard model (individual level)}
This is the reduced-form hazard model applied in \citet{Bollinger2022visibility}, estimated on our data. The discrete hazard benchmark model (15) in Section 5.1 is estimated
by maximum likelihood estimation (MLE), treating PV adoption as $\{0,1\}$ dummy variables over time
in a discrete-time hazard process. For each household-period observation
in which household $i$ has not yet adopted, define $y_{it}=1$
if the household adopts in period $t$, and $y_{it}=0$ otherwise.
The model specifies the conditional adoption probability, or hazard
rate, as
\[
G\left(\lambda_{it}\right)=\xi_{z}+\eta_{t}+\beta_{1}\npv_{it}+\beta_{2}X_{i}+\beta_{3}H_{it}+\beta_{4}\npv_{H_{it}},
\]

\noindent where $G(\cdot)$ is either the logit link $\log\!\left(\lambda_{it}/(1-\lambda_{it})\right)$
or the complementary log-log link $\log\!\left(-\log(1-\lambda_{it})\right)$, both standard
in discrete-time hazard models \citep{Kalbfleisch2002a}; $\xi_{z}$ are zip-code fixed
effects that capture market and geographic factors, $\eta_{t}$ is the temporal effect,
modeled flexibly (polynomial, B-spline, natural spline, or local regression), and $X_{i}$
collects the household characteristics: the total footprint of the house, the ratio of tree
cover to footprint, the sunlight received, the hillshade, and the solar irradiance. The
variables $H_{it}$ and $\npv_{H_{it}}$ denote, respectively, the cumulative number of peer
installations and the cumulative economic value (NPV) of those installations within the
one-mile radial ring. The likelihood contribution for household $i$ in period $t$ is
therefore $\lambda_{it}^{y_{it}}(1-\lambda_{it})^{1-y_{it}}$.
The MLE maximizes the log-likelihood
\[
\sum_{i,t}\left(y_{it}\log(\lambda_{it})+(1-y_{it})\log(1-\lambda_{it})\right),
\]
over all household-period observations in the risk set by choosing
the optimal coefficients $\left\{ \beta_{1},\ldots,\beta_{4}\right\}$ and
fixed effects $\left\{ \xi_{z},\eta_{t}\right\} $.

The estimated coefficients of the hazard benchmark, which uses the complementary log-log link function and fixed effects for $\xi_{z}$ and $\eta_{t}$, are provided in Table~\ref{tab:hazardCoeffs} (we omit the fixed-effect estimates).

\begin{table}[H]
	\begin{center}
		\caption{Estimated coefficients of the hazard benchmark model (b).\label{tab:hazardCoeffs}}
		{\scriptsize
			\begin{tabular}{lrrr}
				\toprule
				Coefficients & Estimate & S.E. & p-value\bigstrut\\
				\midrule
				Intercept & $-7.23$ & $1.61$ & $7.34\times10^{-6}$\\
				$\npv_{it}$ & $2.15\times10^{-4}$ & $0.12\times10^{-4}$ & $\simeq0$\\
				$\npv_{H_{it}}$ & $-9.19\times10^{-6}$ & $1.08\times10^{-6}$ & $\simeq0$\\
				$H_{it}$ & $3.57\times10^{-2}$ & $0.18\times10^{-2}$ & $\simeq0$\\
				$\text{FootPrint2TreeCover}_{i}$ & $2.93\times10^{-4}$ & $0.72\times10^{-4}$ & $4.81\times10^{-5}$\\
				$\text{SunlightReceived}_{i}$ & $4.47\times10^{-4}$ & $0.64\times10^{-4}$ & $2.71\times10^{-12}$\\
				$\text{Hillshade}_{i}$ & $1.76\times10^{-2}$ & $0.84\times10^{-2}$ & $0.036$\\
				$\text{SolarIrradiance}_{i}$ & $-0.41\times10^{-2}$ & $0.18\times10^{-2}$ & $0.027$\\
				\bottomrule
			\end{tabular}
		}
	\end{center}
\end{table}

\subsection*{Benchmark (b): Myopic choice model (static NPV logit)}
\begin{description}
	\item[]Utility for adopting:
 $u_{i1t}(\pvCost_t,\bm\sInstBase_{it},\randomEffects_{i},\latentUtility_{i1t};\pBaseUtility, \pMMoney_{e_i},\bm\pMPE_{g_i})=   \pBaseUtility + \pMMoney_{e_i} \npv_{it}(\pvCost_t) +\bm\pMPE_{g_i}'\bm\sInstBase_{it} + \randomEffects_{i}+\latentUtility_{i1t}$.
\item[]Utility for not adopting: $u_{i0t}(\latentUtility_{i0t}) = \latentUtility_{i0t}$.
\end{description}

Then, at iteration $(r)$ of the MCMC algorithm, the probability that household $i$ adopts in period $t$ is:
\begin{align*}
	\hat p^{(r)}_{i1t}(c_t,\bm{h}_{it}, \xi^{(r)}_{i};\bm\theta^{(r)}) &= \frac{\exp u^{(r)}_{i1t} \left({c}_{t},\bm{h}_{it},\xi^{(r)}_{i}; \bm\theta^{(r)}\right)}{ 1 +\exp u^{(r)}_{i1t}\left({c}_{t},\bm{h}_{it},\xi^{(r)}_{i} ; \bm\theta^{(r)}\right)}.
\end{align*}

\subsection*{Benchmark (c): Partially forward-looking model}
This model assumes that each household $i$ at period $t$ projects the neighborhood installed base as constant over the planning horizon, i.e., $\bm\sInstBase_{i,t+1} = \bm\sInstBase_{it}$, while remaining forward-looking about the net present value of its own system, forecasting the future net cost with equation (7) of the paper.

\subsection*{Benchmark (d): Bass diffusion model (aggregate data)}
We estimated a classic Bass model at the aggregate level, and we projected the number of adopters in out-of-sample quarters. Let $N_t$ be the number of adopters in period $t$, $M$ the market size, $p$ the coefficient of innovation, and $q$ the coefficient of imitation. The regression to estimate is:  
\begin{align}
	N_t &= M \frac{(p+q)^2}{p} \frac{\exp(-(p+q)t)}{\left(1 + \frac{q}{p} \exp(-(p+q)t)\right)^2}\nonumber.
\end{align}
We use a nonlinear least squares regression (R software, nls \{stats\} function) to estimate the parameters of this model. We obtain $p = 0.0005006$ and $q = 0.0464439$. The out-of-sample MAE is 106.172 adopters and the MAPE is 47.208\%. The Bass model performs poorly because the diffusion process in this dataset is at an early stage and does not yet attain the \textit{S}-shape that the Bass equation fits.

Table~\ref{tab:ECbenchmarks} below reports the out-of-sample MAE and MAPE for the benchmark models (a)--(d) and Model~(4)H of Section~5.1 of the paper.

\begin{table}[H]
	\begin{center}
		\caption{Out-of-sample predictive performance of the benchmark models.\label{tab:ECbenchmarks}}
		{\footnotesize
			\begin{tabular}{clcc}
				\toprule
				\textbf{Model} & \textbf{Specification} & \multicolumn{1}{c}{\textbf{MAE}} & \multicolumn{1}{c}{\textbf{MAPE [\%]}}\bigstrut\\
				\midrule
				(4)H & Dynamic discrete choice model (best model) & 31.845 & 15.599 \bigstrut\\
				(a) & Hazard model (individual level) & 41.453 & 21.831 \\
				(b) & Myopic choice model (static NPV logit) & 78.158 & 43.376 \\
				(c) & Partially forward-looking model & 48.167 & 26.352 \\
				(d) & Bass diffusion model (aggregate data) & 106.172 & 47.208 \\
				\bottomrule
			\end{tabular}\\
		}
	\end{center}
\vspace{0.2cm}
{\footnotesize 
\setlength{\baselineskip}{10pt}
\noindent \emph{Note.} MAE, mean absolute error (number of adopters); MAPE, mean absolute percentage error.Metrics are computed on the five out-of-sample quarters following the 33-quarter training window.	
\par}

\end{table}
\clearpage
\section{Details of Differentiated Rebate Policy Simulations}\label{sec:ECtargeting}
\textbf{Geographic differentiation.} We evaluate rebate policies that differentiate the rebate rate across the geographic segments through a complete factorial: the net cost faced by urban, suburban, and rural households is each scaled by a multiplier in $\{0.8, 1.0, 1.2\}$, for $3^3 = 27$ policies; a lower multiplier corresponds to a higher rebate rate. Each policy is one forward simulation of the selected model, Model~(4)H, under a common rebate schedule. The simulation seeds on the first 16 quarters of observed adoption and generates new adoption from quarter~17 onward, so the reported horizon, period~56, corresponds to 10 years of simulated adoption. For each policy we record the cumulative adoption and the program budget at period~56. Figure~7 in the main text plots cumulative adoption against program budget for the 27 policies, the three-square glyph at each point encoding the urban, suburban, and rural multiplier (left, center, and right square) by color; Table~\ref{tab:ECtargeting} summarizes the results as regressions across the 27 policies.

Cumulative adoption rises almost linearly with the program budget, from \$38.4~million and 7{,}837 adopters at the most conservative policy to \$45.0~million and 8{,}480 adopters at the most generous: column~(1) of Table~\ref{tab:ECtargeting} regresses adoption on budget across the 27 policies and gives a slope of 94.9 adopters per million dollars with $R^2 = 0.97$. Column~(2) regresses adoption on the three multipliers directly: lowering any segment's net cost raises total adoption, with the suburban multiplier the strongest lever (75 additional adopters per $0.1$ discount, against 54 for urban and 31 for rural). Column~(3) asks the budget-constrained question by regressing each policy's deviation from the fitted line of column~(1) on the multipliers, and the urban multiplier separates the policies: a $0.1$ urban discount adds about 14 adopters beyond what its budget predicts, whereas suburban and rural discounts move a policy below the line. Accordingly, all nine policies that discount the urban net cost (multiplier $0.8$) lie above the fitted line, by $+13$ to $+52$ adopters ($+32$ on average), whereas the policies that hold the urban net cost at its observed level or increase it lie below the line on average ($-8$ and $-25$ adopters, respectively). This is the pattern visible in Figure~7 of the main text: the glyphs above the fitted relationship have the left (urban) square in green.

The worked examples in Section~6.4 of the paper come from the same simulations. At a budget of about \$41~million, the policy $(0.8, 1.2, 1.0)$ attains the highest adoption in its budget band, 8{,}120 adopters at \$40.87~million. At a budget of about \$42~million, the policy $(0.8, 1.0, 1.2)$ attains 8{,}212 adopters at \$41.90~million, against 8{,}154 adopters at \$41.79~million under the observed uniform policy $(1.0, 1.0, 1.0)$.

\begin{table}[H]
\begin{center}		
\caption{Regression summary of the 27 geographically differentiated rebate policies at period~56. \label{tab:ECtargeting}}
\centering{\footnotesize
\setlength{\tabcolsep}{8pt}\renewcommand{\arraystretch}{1.1}
\begin{tabular}{lrrr}
\toprule
 & (1) & (2) & (3) \\
Dependent variable & Adoption & Adoption & Deviation from budget line \\
\midrule
Budget (\$ million) & 94.86 &  &  \\
 & {\scriptsize (3.15)} &  &  \\
Urban multiplier $m_{\mathrm{urban}}$ &  & $-$544.05 & $-$141.87 \\
 &  & {\scriptsize (9.21)} & {\scriptsize (7.10)} \\
Suburban multiplier $m_{\mathrm{suburban}}$ &  & $-$752.41 & 52.74 \\
 &  & {\scriptsize (9.21)} & {\scriptsize (7.10)} \\
Rural multiplier $m_{\mathrm{rural}}$ &  & $-$310.90 & 41.63 \\
 &  & {\scriptsize (9.21)} & {\scriptsize (7.10)} \\
Constant & 4,198.34 & 9,764.80 & 47.50 \\
 & {\scriptsize (131.64)} & {\scriptsize (16.02)} & {\scriptsize (12.35)} \\
\midrule
$R^2$ & 0.973 & 0.998 & 0.955 \\
Policies & 27 & 27 & 27 \\
\bottomrule
\multicolumn{4}{l}{\scriptsize \emph{Note.} One observation per policy; the uniform policy $(1.0, 1.0, 1.0)$ is the observed rebate schedule.}\\
\end{tabular}}
\end{center}

{\footnotesize 
	\setlength{\baselineskip}{10pt}
	\noindent \emph{Note.} Column~(1) regresses cumulative adoption on the program budget; column~(2) regresses adoption on the three net-cost multipliers; column~(3) regresses the deviation of adoption from the fitted budget line of column~(1) on the multipliers. Standard errors in parentheses. A lower multiplier corresponds to a higher rebate rate, so the negative urban coefficient in column~(3) means that only the urban discount raises adoption beyond what its budget predicts.
	\par}

\end{table}

\textbf{Home-value differentiation.} The same design applies to the two economic segments: the net cost of households with home market value below \$300{,}000 and of those at or above it is each scaled by a multiplier in $\{0.8, 1.0, 1.2\}$, for $3^2 = 9$ policies, each a forward simulation of Model~(4)H evaluated at period~56. Figure~\ref{fig:ECeconglyph} plots cumulative adoption against program budget for the nine policies, the two-square glyph at each point encoding the low- and high-value multiplier (left and right square) by color. The uniform policy $(1.0, 1.0)$ reproduces the geographic baseline exactly, 8{,}154 adopters at \$41.79~million, a cross-check between the two simulation sets.

Adoption again tracks the budget almost one for one: regressing adoption on budget across the nine policies gives a slope of 95.7 adopters per million dollars with $R^2 = 0.998$, so the glyphs hug the fitted line and no home-value policy delivers adoption much beyond what its budget predicts. Relative to the uniform policy, discounting the high-value segment's net cost to $0.8$ adds about 246 adopters and discounting the low-value segment's adds about 82, at 99.1 and 111.3 marginal adopters per incremental million dollars, respectively: the segment that adds more adoption also absorbs proportionally more budget. Differentiating the rebate by home value therefore does not substantially improve adoption for a fixed budget, the conclusion reported in Section~6.4 of the paper.

\begin{figure}[H]
{\centering
\includegraphics[width=0.72\textwidth]{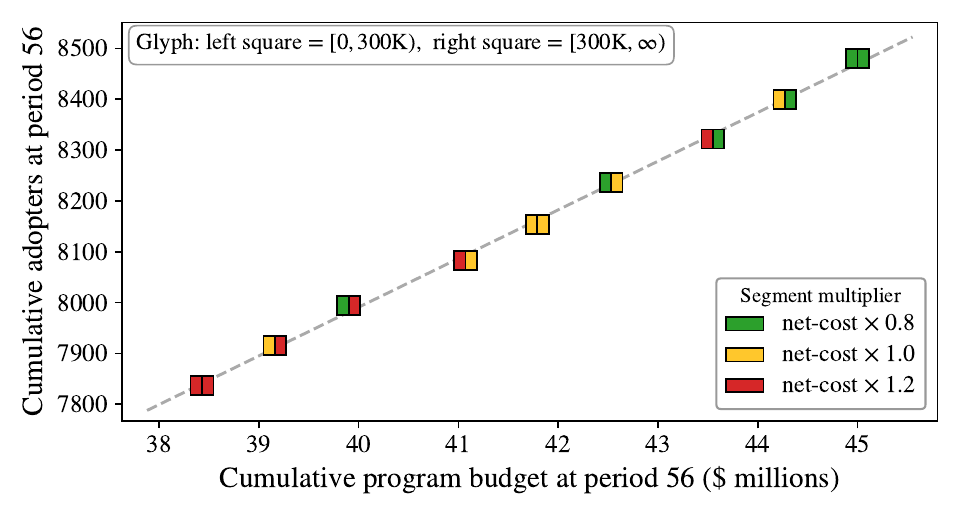}
\caption{Home-value differentiated rebate policies at period~56.\label{fig:ECeconglyph}}
}
{\footnotesize 
	\setlength{\baselineskip}{10pt}
	\noindent \emph{Note.} Cumulative adoption against program budget for the nine economic-segment policies. The two-square glyph at each point encodes the multiplier on the net cost of low-value households ($[0,300\mathrm{K})$, left square) and high-value households ($[300\mathrm{K},\infty)$, right square): green $0.8$, yellow $1.0$, red $1.2$.
	\par}
\end{figure}
\clearpage
\section{End-Horizon Sensitivity Analysis for Policy Simulations}\label{sec:ECrobustness}

The cost-effectiveness comparisons of the rebate switching policies in Section~6.1 are made at period~56, 23 quarters (nearly 6 years) beyond the estimation window. We test whether those comparisons are robust to the end horizon by reconstructing the frontier analyses at five horizons, periods 40, 44, 48, 52, and 56 at four-period steps, equal to 2 through 6 years beyond the last quarter used for estimation (quarter~33). Because a one-step policy RSP($\tau$) that switches at $\tau>T$ does not switch before the end horizon $T$ and hence coincides with the no-switch schedule, we evaluate at each horizon only the policies that switch within it, $\tau\le T$. Figure~\ref{fig:ECrobustness} reconstructs the frontier at each horizon: cumulative installed capacity against program budget on the left, cumulative carbon savings against program budget on the right; Table~\ref{tab:ECrobustness} lists the carbon-maximizing policy and its metrics at each horizon.

Both features of the period-56 frontier reappear across the horizons. Installed capacity rises with budget as a concave curve, so each later switch yields less capacity gain and the costliest policies occupy the flat upper arm. Carbon savings are non-monotone in budget: at every horizon the curve has an interior maximum, with the switches before and after the optimum both below it. At every horizon, the carbon-maximizing one-step policy is an interior switch, never the most expensive policy: RSP(32) at the 2-year horizon and RSP(36) from 3 years onward. The budget at the optimum is about \$30~million at 2 years and about \$33~million from 3 years on, while the carbon level grows from 174.2 at 2 years to 604.2 at 6 years.

The key reason for the interior optimum is the forward-looking behavior of households. Because households know the rebate is reduced at $\tau$, an earlier switch concentrates installations in the quarters before it; that capacity operates sooner and accrues carbon savings over more of the horizon. A later switch installs marginally more capacity by the horizon's end, 38.6~MW for RSP(48) against 36.4~MW for RSP(36) at period~56, yet accumulates less carbon, 552.3 against 604.2, because its capacity arrives later. Cumulative carbon savings, computed as the integral of installed capacity over time, reward the intermediate switch that front-loads adoption. Each policy's program budget plateaus once its rebate is switched off, so it changes little at horizons beyond the switch; the optimum's budget rises across horizons only because the optimal switch time itself moves later.

The counterfactual conclusion therefore holds independently of the period-56 horizon. Evaluated over any window from 2 to 6 years beyond the estimation window, an intermediate switch yields the most carbon savings per budget dollar, installed capacity exhibits diminishing returns to budget, and deferring the switch past the interior optimum raises spending while lowering cumulative carbon. The cost-effective switch advances from RSP(32) at 2 years to RSP(36) from 3 years onward, and it is interior at every horizon.

\begin{table}[H]
\begin{center}
	\caption{Horizon robustness of the one-step cost-effectiveness frontier.\label{tab:ECrobustness}}
	{\footnotesize
		\begin{tabular}{rrlrrr}
			\toprule
			Horizon (period) & Years beyond est. & Carbon-max policy & Budget (\$M) & Installed MW & Carbon savings \\
			\midrule
			40 & 2 & RSP(32) & 29.6 & 15.7 & 174.2 \\
			44 & 3 & RSP(36) & 32.6 & 21.8 & 250.8 \\
			48 & 4 & RSP(36) & 32.6 & 27.1 & 347.9 \\
			52 & 5 & RSP(36) & 32.6 & 32.3 & 466.5 \\
			56 & 6 & RSP(36) & 32.6 & 36.4 & 604.2 \\
			\bottomrule
	\end{tabular}}
\end{center}
\vspace{0.2cm}
{\footnotesize 
	\setlength{\baselineskip}{10pt}
	\noindent \emph{Note.} At each simulation horizon, the carbon-maximizing one-step policy with its program budget, cumulative installed capacity, and cumulative carbon savings (the carbon metric of Figure~5 of the paper). The optimum is RSP(32) at 2 years and RSP(36) from 3 years onward, and at every horizon it is an interior switch; period~56 is the headline horizon, nearly 6 years beyond the estimation window.
	\par}

\end{table}

\begin{figure}[H]
{\centering
\includegraphics[width=\textwidth]{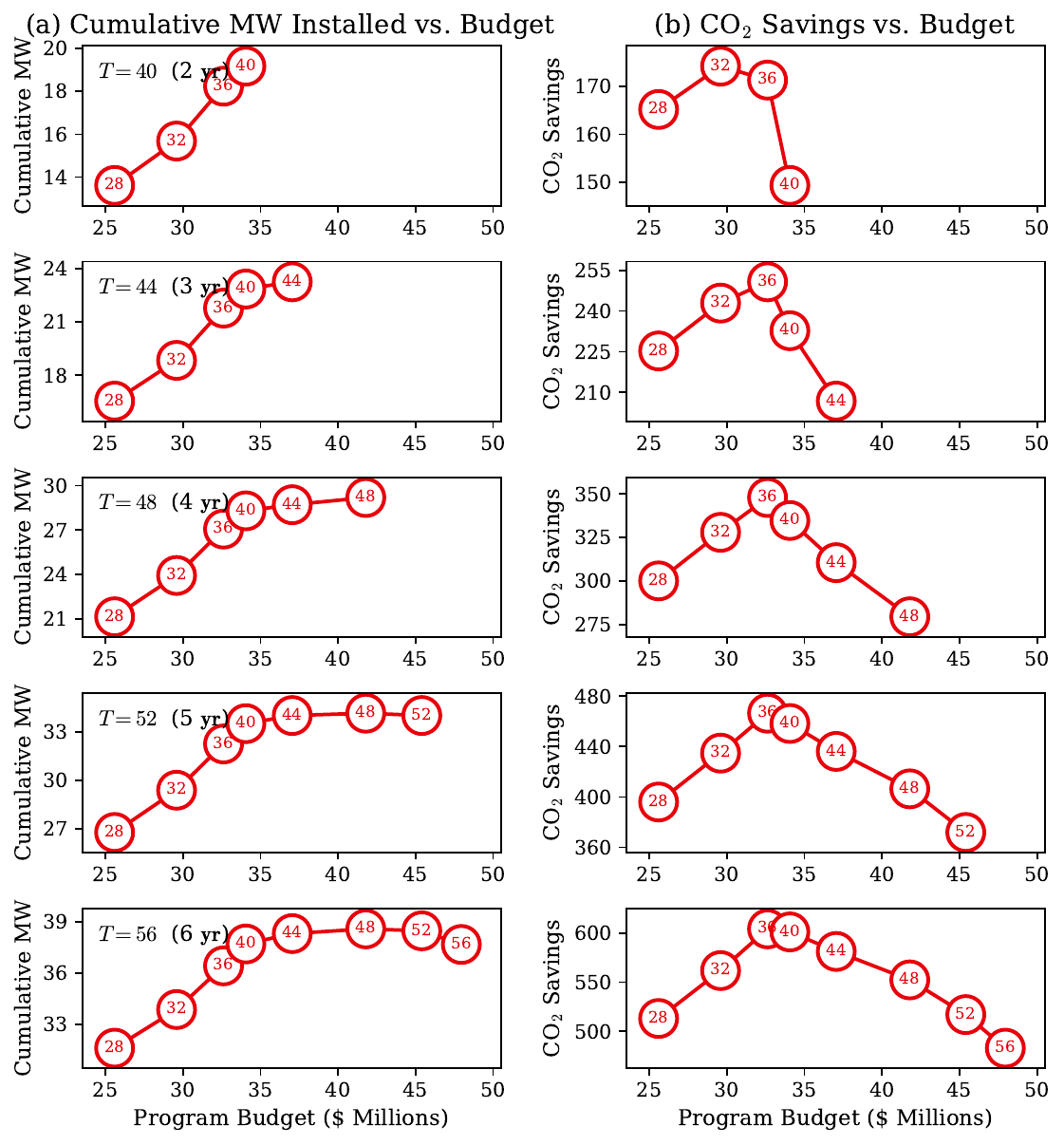}
\caption{One-step cost-effectiveness frontier at five simulation horizons.\label{fig:ECrobustness}}
}

{\footnotesize 
	\setlength{\baselineskip}{10pt}
	\noindent \emph{Note.} Periods 40 to 56 (2 to 6 years beyond the estimation window, whose last quarter is~33). Each row reconstructs the period-56 frontier at one horizon: cumulative installed capacity against program budget (left) and cumulative carbon savings against program budget (right), over the one-step policies RSP($\tau$) that switch within the horizon, $\tau\le T$, each point labeled by $\tau$. Installed capacity is concave in budget at every horizon; carbon savings peak at an interior switch, RSP(32) at 2 years and RSP(36) from 3 years onward.
	\par}
\end{figure}
\end{APPENDICES}



\clearpage
\bibliographystyle{informs2014}
\bibliography{library}
\end{document}